\documentclass[sigconf,twocolumn,nonacm,dvipsnames]{acmart}

\usepackage{booktabs}
\usepackage{geometry}
\usepackage{multirow}
\usepackage{graphicx}
\usepackage{colortbl}
\usepackage{hyperref}

\definecolor{twins}{HTML}{E6F2FF}
\definecolor{acs}{HTML}{FFF0F5}
\definecolor{german}{HTML}{F0FFF0}
\definecolor{so}{HTML}{FFF5E6}

\usepackage{appendix}
\usepackage{paralist}
\usepackage{subcaption}
\usepackage{standalone}
\usepackage{color, colortbl}
\usepackage{arydshln}
\usepackage{listings}
\usepackage{scalerel}
\usepackage{tikz}
\usepackage{amsmath}
\usepackage{pifont}
\usepackage[framemethod=TikZ]{mdframed}
\usepackage[capitalize,noabbrev]{cleveref}
\usepackage{enumitem}
\usepackage{xspace}
\usepackage{cprotect}
\usepackage{colortbl} 
 \usepackage{algorithmic}

\usepackage{forest}
\usepackage{multicol}
\usepackage{multirow}
\usepackage{array}
\usepackage{rotating}
\usepackage{xcolor}
\usepackage{makecell}
\usepackage{xspace}
\definecolor{vlightgray}{gray}{0.85}

\usepackage{pifont}

\Crefname{algocf}{Algorithm}{Algorithms}
\usepackage[colorinlistoftodos]{todonotes}

\usepackage{tikz}
\usepackage{graphicx}
\usetikzlibrary{positioning}

\usepackage{xcolor}

\definecolor{Lavender}{rgb}{0.9, 0.9, 0.98}

\definecolor{Salmon}{rgb}{0.98, 0.8, 0.85}

\usepackage{tikz,xfp} 

\usepackage{tikz,xfp}

\usepackage{tikz,xfp}

\newcommand{\paratitle}[1]{
\noindent{\bf #1.}}

\usepackage[vlined,commentsnumbered,ruled]{algorithm2e}
\let\oldnl\nl
\newcommand{\nonl}{\renewcommand{\nl}{\let\nl\oldnl}}

\SetCommentSty{mycommfont}

\def\HiLi{\leavevmode\rlap{\hbox to \hsize{\color{red!20}\leaders\hrule height .8\baselineskip depth .5ex\hfill}}}

\usepackage{amsfonts}
\usepackage{amsmath}
\usepackage{comment}
\usepackage{multirow}
\usepackage{paralist}
\usepackage{paralist}
\usepackage{mathtools}

\newtheorem{problem}{Problem}
\newtheorem{definition}{Definition}[section]
\newtheorem{theorem}{Theorem}[section]
\newtheorem{remark}{Remark}[section]

\newtheorem{example}{Example}[section]

\usepackage{xcolor}
\usepackage{soul}
\sethlcolor{yellow} 

\newcommand{\ignore}[1]{}

\newcommand{\cut}[1]{}

\newcommand{\probName}{\text{CaRET}}

\usepackage{xcolor}

\definecolor{deepPurple}{RGB}{120, 0, 180}

\newcommand{\algoNameTuple}{\textcolor{deepPurple}{\textsc{SubCure-tuple}}}
\newcommand{\algoNamePattern}{\textcolor{blue}{\textsc{SubCure-pattern}}}
\newcommand{\sysName}{\textsc{SubCure}}

\newcommand{\topk}{\textsc{Single-Update-SubCure-tuple}}
\newcommand{\scoded}{\textsc{SCODED}}
\newcommand{\ifoutlier}{\textsc{IF}}
\newcommand{\lof}{\textsc{LOF}}

\newcommand{\attrset}{\ensuremath{\mathbb{A}}}

\newcommand{\dom}{{\tt dom}}

\newcommand{\db}{\ensuremath{D}}

\newcommand{\pattern}{\ensuremath{\mathcal{\psi}}}

\newcommand{\sat}{\textsc{SAT}}

\newcommand{\bbN}{\mathbb{N}}

\definecolor{moonstoneblue}{rgb}{0.45, 0.66, 0.76}
\definecolor{oldlace}{rgb}{0.99, 0.96, 0.9}
\definecolor{mintcream}{rgb}{0.96, 1.0, 0.98}
\definecolor{mintgreen}{rgb}{0.6, 1.0, 0.6}
\definecolor{mistyrose}{rgb}{1.0, 0.89, 0.88}
\definecolor{palegold}{rgb}{0.9, 0.75, 0.54}
\definecolor{palechestnut}{rgb}{0.87, 0.68, 0.69}
\definecolor{darkgreen}{RGB}{0,100,0} 

\newcommand{\reva}[1]{{\leavevmode\color{black}{#1}}}
\newcommand{\revb}[1]{{\leavevmode\color{black}{#1}}}
\newcommand{\revc}[1]{{\leavevmode\color{black}{#1}}}

\usepackage{colortbl}
\usepackage{xcolor}

\def\HiLiG{\leavevmode\rlap{\hbox to \hsize{\color{green!30}\leaders\hrule height .8\baselineskip depth .5ex\hfill}}}
\def\HiLiY{\leavevmode\rlap{\hbox to \hsize{\color{yellow!50}\leaders\hrule height .8\baselineskip depth .5ex\hfill}}}
\usepackage{wrapfig}
\definecolor{light-gray}{gray}{0.95}
\usepackage{tcolorbox}

\usepackage{amsmath}

\begin{document}

\pagenumbering{gobble}  


\clearpage
\pagenumbering{arabic}  
\setcounter{page}{1}

\title{Stress-Testing Causal Claims via Cardinality Repairs}


\author{Yarden Gabbay}
\authornote{These authors contributed equally to this work.}
\email{yardengabbay@campus.technion.ac.il}
\affiliation{
  \institution{Technion}
  \country{Israel}
}
\author{Haoquan Guan$^*$}
\email{h3guan@ucsd.edu}
\affiliation{
  \institution{University of California, San Diego}
    \country{USA}
}
\author{Shaull Almagor}
\email{shaull@technion.ac.il}
\affiliation{
  \institution{Technion}
    \country{Israel}
}

\author{El Kindi Rezig}
\email{elkindi.rezig@utah.edu}
\affiliation{
  \institution{University of Utah}
    \country{USA}
}

\author{Brit Youngmann}
\email{brity@technion.ac.il}
\affiliation{
  \institution{Technion}
    \country{Israel}
}

\author{Babak Salimi}
\email{bsalimi@ucsd.edu}
\affiliation{
  \institution{University of California, San Diego}
    \country{USA}
}

\begin{abstract}
Causal analyses derived from observational data underpin high-stakes decisions in domains such as healthcare, public policy, and economics. Yet such conclusions can be surprisingly fragile: even minor data errors - duplicate records, or entry mistakes - may drastically alter causal relationships. This raises a fundamental question: \emph{how robust is a causal claim to small, targeted modifications in the data?} Addressing this question is essential for ensuring the reliability, interpretability, and reproducibility of empirical findings.

We introduce \sysName, a framework for \emph{robustness auditing via cardinality repairs}. Given a causal query and a user-specified target range for the estimated effect, \sysName\ identifies a small set of tuples or subpopulations whose removal shifts the estimate into the desired range. This process not only quantifies the sensitivity of causal conclusions but also pinpoints the specific regions of the data that drive those conclusions. We formalize this problem under both tuple- and pattern-level deletion settings and show both are NP-complete. To scale to large datasets, we develop efficient algorithms that incorporate machine unlearning techniques to incrementally update causal estimates without retraining from scratch.

We evaluate \sysName\ across four real-world datasets covering diverse application domains. In each case, it uncovers compact, high-impact subsets whose removal significantly shifts the causal conclusions—revealing vulnerabilities that traditional methods fail to detect. Our results demonstrate that cardinality repair is a powerful and general-purpose tool for stress-testing causal analyses and guarding against misleading claims rooted in ordinary data imperfections.
\end{abstract}




\maketitle




\section{Introduction}

Causal inference now underpins decisions in medicine, policy, and economics while powering fairness auditing and data-debiasing \citep{salimi2019interventional,zhu2023consistent}, explainability \citep{galhotra2021explaining,miller2019explanation}, domain-robust learning \citep{magliacane2018domain,scholkopf2021toward}, and core data-management tasks such as query explanation, discovery, and cleaning \citep{salimi2018bias,galhotra2023metam,pirhadi2024otclean,markakis2024logs,youngmann2024summarized,youngmann2023explaining}. Yet a growing body of work shows that seemingly routine data imperfections - such as misclassification, measurement noise, duplicate records, and composite corruption - can distort causal estimates and invalidate hypothesis tests \citep{Nab2020,MilesValeriCoull2024,Bogaert2025,SarracinoMikucka2017,LockElAnsari2025,AgarwalSingh2024,KadlecSainaniNimphius2023}. If left unaddressed, such distortions can lead to ineffective medical interventions, misallocated resources, or flawed algorithmic decisions, undermining the reliability of empirical findings and their downstream applications. These observations highlight a critical question: \textbf{How robust are our causal claims to common, low-level data errors?}

In this work, we ask a concrete version of the robustness question: given an observational dataset, a treatment variable \( T \), an outcome variable \( O \), and a causal estimation procedure (e.g., the Average Treatment Effect (ATE), defined as the expected difference in outcomes between treated and untreated groups), how many data records must be removed to shift the estimated effect into a user-specified target range? This question formalizes a data-centric notion of robustness-sensitivity not to modeling assumptions, but to minimal, targeted perturbations of the data itself. We refer to this as the problem of \emph{cardinality repair for causal effect targeting} (abbreviated as \probName\ for \underline{Ca}rdinality \underline{R}epair for causal \underline{E}ffect \underline{T}argeting): identifying the smallest subset of tuples (or subpopulations) whose removal pushes the estimate across a specified threshold. Beyond quantifying fragility, such repairs help reveal which records or regions of the data are most responsible for driving the causal conclusion.

We present \sysName\ (\underline{SUB}set \underline{C}a\underline{U}sal \underline{RE}pair), a framework that helps analysts assess how strongly their causal conclusions depend on specific parts of the data. \sysName\ supports two modes: \emph{tuple-level} repair, which identifies individual records whose removal shifts the estimated effect, and \emph{pattern-level} repair, which targets compact subpopulations defined by attribute–value predicates. It surfaces minimal edits that change the conclusion and highlights the regions of the data most responsible for a given estimate, enabling users to explore robustness interactively and with minimal assumptions. We now illustrate this workflow through two motivating examples, showing how small, targeted deletions can flip or fortify real-world causal conclusions.

\begin{example}
\label{ex:twins}
\textbf{Twins mortality.}  The \emph{Twins} dataset~\cite{louizos2017causal} records birth details for same-sex twins; the \emph{treatment} is being the heavier twin, and the \emph{outcome} is first-year mortality.  
Alex, a social scientist, adjusts for gestational age, birth weight, prenatal care, and maternal status and obtains an ATE of $-0.016$—suggesting a small protective effect for the heavier twin. \emph{Tuple-level insight:}  
Running \sysName\ in tuple mode, Alex discovers that removing just \textbf{270 records} (1.1\% of the sample) nudges the ATE to $-0.0009$—a \textbf{94.4\% move toward zero}.  These records are not statistical outliers under common heuristics; they appear representative on observed covariates yet collectively exert strong influence on the estimate. A quick inspection hints that many belong to extremely low-gestational-age births, where overall mortality is uniformly high—an avenue for further clinical review. \emph{Pattern-level insight:}  
Switching to pattern mode, \sysName\ highlights a subpopulation of \textbf{4\,401 first-born twins} weighing between \textbf{1786\,g and 1999\,g} (18.3\% of the data).  Deleting this slice flips the ATE to $0.0003$—a \textbf{101.8\% swing} that reverses the sign.  This weight band sits just below the 2 kg clinical low-birth-weight threshold; delivery complications concentrated in this range may offset any advantage of being the heavier twin, illustrating how a narrow but medically meaningful subgroup can dictate the overall conclusion. Together, these findings show how \sysName\ quantifies fragility and surfaces domain-specific hypotheses—here, extreme prematurity and a critical birth-weight band—that warrant deeper investigation before drawing policy or clinical inferences.
\end{example}


\begin{example}
\label{ex:acs}
\textbf{Disability and wages.}  
The \emph{American Community Survey} (ACS)~\cite{ACS_Data} is a nationwide instrument conducted by the U.S.\ Census Bureau. Alex studies the causal effect of \emph{not} having a disability on annual wages, controlling for education, public-health coverage, gender, and age.  The estimated ATE is \$8{,}774, indicating that, on average, people without disabilities earn roughly nine thousand dollars more per year.
\emph{Tuple-level insight:}  
Running \sysName\ in tuple mode, Alex learns that deleting just \textbf{8{,}400 records} (0.7\% of the data) lifts the ATE to \$11{,}512—a \textbf{31.1\% increase}. A quick look shows many of the removed records cluster in low-wage service occupations with high disability rates; trimming them disproportionately amplifies high-income, non-disabled earners, inflating the gap.
\emph{Pattern-level insight:} 
Shifting the estimate into the target range requires removing a subpopulation comprising 11.9\% of the data. In contrast, to achieve a comparable \textbf{downward} shift, \sysName\ must remove nearly 74\% of the data. This asymmetry signals that the observed wage advantage is highly sensitive to a small slice of low-income disabled workers, yet remarkably stable in the opposite direction. Such asymmetry raises fairness questions: the gap may be driven by a narrow pocket of the labor market rather than a uniform disadvantage. Together, these findings illustrate how \sysName\ pinpoints imbalance in representation—here, a small cohort of low-income disabled respondents—helping analysts decide whether additional weighting, stratification, or data collection is needed before drawing policy conclusions.
\end{example}

Our approach builds on the long-standing tradition of \emph{sensitivity analysis}, which asks how causal conclusions shift under model-based threats such as hidden confounding, selection bias, or measurement error~\citep{rosenbaum2002observational,vanderweele2017evalue,cinelli2020making,fjeldstad2021simex,blackwell2014selection}. Classical methods typically rely on \emph{controlled bias parameters}, or an additive-noise model to capture measurement error—while assuming the dataset itself is clean and internally coherent. 
\reva{Several works have proposed robustness metrics to evaluate how stable causal conclusions remain under data perturbations or removal, to overturn an inference~\cite{frank2013would}, quantifying worst-case estimate changes under data deletions ~\cite{broderick2020automatic}, and estimating the worst-case ATE across subpopulations under confounder shifts~\cite{jeong2020robust}}. \reva{Our approach quantifies ATE sensitivity by directly examining how causal conclusions shift when actual tuples are removed. This is particularly relevant for real-world datasets, which often contain messy and heterogeneous issues, such as duplicate records, miscoded or missing values, format inconsistencies, or corrupted subsets, that previous works do not capture. Our cardinality-repair strategy addresses this gap by treating the data as mutable and asking \emph{how many records - or which subpopulations - must be removed to move an effect into a user-chosen range?} The resulting repair size provides a model-agnostic, data-centric, and intuitive robustness metric that complements and guides classical sensitivity analyses.}



Our method also relates to prior work on constraint-based data cleaning. Integrity-oriented approaches enforce functional or conditional functional dependencies via tuple-level repairs~\citep{bohannon2006conditional,kolahi2009approximating,livshits2022shapley,livshits2020computing,horizon,DBLP:conf/icdt/CarmeliGKLT21,geerts2013llunatic}. More recent systems repair violations of conditional independence between treatment and outcome—typically to force a \emph{null} effect (\(\text{ATE}=0\))~\citep{pirhadi2024otclean,wu2013scorpion,salimi2019interventional}. \sysName\ subsumes these cases yet goes further on two fronts: it supports \emph{arbitrary} effect targets (not just zero), and it pinpoints \emph{influential subpopulations}, rather than only individual rule-breaking tuples~\citep{yan2020scoded,pirhadi2024otclean}. This broader scope lets analysts reason simultaneously about constraint correctness, structural bias, and the overall robustness of causal claims.

\setlist[itemize]{leftmargin=0pt,nosep} 

Our main contributions are as follows.\\
$\bullet$ \textbf{Problem formulation and hardness results  
(Section~\ref{sec:problem}).}  
We formalize the \emph{cardinality repair for causal effect targeting} (\probName) problem:  
Given a treatment, an outcome, and a desired effect interval, find the smallest data subset whose removal shifts the estimated effect into the range.  
We consider two cost models: (i) tuple-level deletions, and (ii) pattern-level deletions that remove entire subpopulations defined by attribute–value predicates.  
We prove that both variants are NP-complete. For tuple-level, this is shown by reduction from \textsc{Subset-Sum}. For pattern-level, we show that this problem is harder than tuple-level, regardless of the causal effect function. 




\noindent
$\bullet$ \textbf{Fast, incremental search algorithms
(Sections~\ref{sec:algo}–\ref{sec:optimizations}).}%
We develop fast, scalable cardinality-repair algorithms using two complementary search strategies.  
In tuple mode, we cluster the dataset using a two-stage \(k\)-means procedure to sample a proxy set, estimate each sampled tuple’s marginal influence, and iteratively delete the one with the highest impact. Influence scores are batch-refreshed every few steps, avoiding unnecessary rescoring and reducing runtime.  
In pattern mode, we perform bottom-up random walks over conjunctions of attribute–value predicates. A dynamic weighting mechanism steers the walk toward impactful predicates, and exploration stops early if the candidate subgroup grows too large. These strategies yield an efficient anytime search that surfaces high-leverage subpopulations without exhaustive enumeration.

Both strategies require re-evaluating the causal effect after each candidate deletion. To make this feasible, we use two standard causal effect estimators—linear regression~\cite{ding2018causal} and inverse-propensity weighting (IPW)~\cite{rosenbaum1983central}—and develop incremental update optimizations that eliminate the need to recompute from scratch after each change. For linear regression, we cache sufficient statistics such as the covariance matrix and response cross-product, and apply low-rank updates when rows are removed. This lets us recompute the treatment effect without rebuilding the full model, and in time that depends only on the number of confounders, not the size of the dataset. For IPW, we warm-start from the previous logistic regression fit and apply a single Fisher-scoring step that adjusts only for the removed rows. This avoids full optimization over all tuples and yields updated propensity scores in constant time per deletion. These incremental estimators ensure that each search step requires only localized computation, preserving estimator fidelity while enabling \sysName\ to scale to large, high-dimensional datasets with interactive latency.

\noindent
$\bullet$ \textbf{Comprehensive empirical evaluation  
(Section~\ref{sec:exp}).}  
We evaluate \sysName\ on four public datasets and one synthetic benchmark, \reva{comparing against eight baseline methods.}  
\sysName\ (i) consistently identifies smaller subsets of influential data,  
(ii) scales to million-row, high-dimensional datasets, and  
(iii) achieves up to order-of-magnitude speed-ups through incremental updates without compromising accuracy.  
Case studies demonstrate that \reva{\sysName\ complements existing sensitivity analysis methods by offering additional insights into ATE sensitivity}, while its identified subsets remain influential across multiple ATE estimators, highlighting the robustness of our approach.

\section{Related work}
\label{sec:related}





\noindent
\textbf{Constraint-based data cleaning}
Traditional data cleaning research focuses on enforcing integrity constraints~\cite{ilyas2019data,5767833}, such as functional dependencies~\cite{bohannon2006conditional,kolahi2009approximating,livshits2022shapley,livshits2020computing,horizon,DBLP:conf/icdt/CarmeliGKLT21}, conditional functional dependencies~\cite{bohannon2006conditional,geerts2013llunatic}, denial constraints~\cite{chomicki2005minimal,chu2013holistic}, and inclusion dependencies~\cite{bohannon2005cost}. Common repair operations include tuple deletion (cardinality repair~\cite{afrati2009repair,miao2020computation}), insertion~\cite{salimi2019interventional}, and value updates~\cite{pirhadi2024otclean}. 
\revb{
Recent work \cite{yan2020scoded, salimi2019interventional, pirhadi2024otclean} has explored enforcing statistical constraints that represent conditional dependence or independence between attributes. Capuchin \cite{salimi2019interventional} computes optimal repairs by adding or removing tuples to enforce conditional independence (CI) constraints. The authors of~\cite{ahuja2021conditionally} use GANs to generate data that satisfies CI constraints, with a focus on learning generative models rather than repairing the data.}
\revb{The works most closely related to ours are OTClean~\cite{pirhadi2024otclean} and SCODED~\cite{yan2020scoded}, both of which enforce statistical constraints representing conditional dependence or independence. OTClean employs optimal transport to probabilistically adjust data values to satisfy CI constraints; since it modifies values rather than removing tuples, its results are not directly comparable to ours. SCODED focuses on identifying the most influential tuples that violate a CI constraint. Its goal is to provide explanations rather than repairs, and thus it reports the top-$k$ most influential tuples instead of minimizing the subset to remove.
In contrast, \sysName not only enforces CIs (by driving the ATE to zero), but also supports targeted causal repair, shifting the ATE into a user-defined range. Beyond identifying influential tuples, \sysName\ is also capable of uncovering meaningful subpopulations whose removal has a significant effect on the ATE.}

\noindent
\textbf{Sensitivity analysis}: A large body of work has developed techniques to assess how causal conclusions might shift in the presence of various threats~\cite{robins2000sensitivity,diaz2013sensitivity}. For instance, \cite{rosenbaum2002observational} quantifies the influence of unmeasured confounding in observational studies. \cite{blackwell2014selection,vanderweele2017evalue} extended this perspective to selection bias, presenting sensitivity analysis tools for evaluating how causal estimates would change under varying assumptions about the selection process.  \cite{cinelli2020making} developed a framework for omitted variable bias, enabling interpretable and flexible assessments of robustness to hidden confounders. \cite{fjeldstad2021simex} addressed measurement error through the simulation extrapolation method, which corrects causal estimates by modeling the effect of added noise and extrapolating to the noise-free case.

\reva{
Several works have proposed robustness measures to assess the stability of causal inferences under data removal or perturbations.
KonFound \cite{frank2013would} is a robustness measure that quantifies the degree of hidden bias needed to overturn an inference. It assesses sensitivity to unobserved confounders through hypothetical value perturbations, identifying the “switch point” where an effect becomes null.
ZamInfluence \cite{broderick2020automatic} is a general-purpose robustness measure that measures how much an estimate can change when a fraction of the data is removed, treating the estimator as a black box. It searches for worst-case deletions that either flip the sign of the estimate or alter its statistical significance, making it applicable across diverse statistical models.
WTE \cite{jeong2020robust} is a robustness measure for confounders shift that estimates the worst-case ATE across subpopulations of a given size, providing conservative guarantees of external validity under distributional changes.
Our approach offers a complementary perspective. Assuming confounders are observed, we examine ATE sensitivity by removing actual tuples or subpopulations. Rather than identifying a single “switch point”, our method explores how the ATE responds to data removal and how it can be shifted toward any user-specified target range, offering a data-centric robustness measure that can support and inform traditional sensitivity analyses.}


\noindent
\textbf{Machine unlearning}: Machine unlearning focuses on efficiently removing the influence of data points from trained models without full retraining~\cite{ginart2019making,guo2020certified,wu2020deltagrad}. This is particularly relevant in settings requiring compliance with data privacy regulations (``the right to be forgotten''~\cite{rosen2011right}). Prior work has proposed exact and approximate unlearning techniques across various models, including regression models~\cite{guo2020certified,izzo2021approximate,wu2020priu}, and deep neural networks~\cite{nguyen2022survey,golatkar2020eternal}. In this work, we leverage unlearning techniques to accelerate causal effect estimation after data subset removal, enabling efficient exploration of data subsets that influence the ATE.

 \section{Background on Causal Inference}
\label{sec:prelim} 






\paratitle{Causal inference and ATE}
\revc{We use Pearl's model for {\em observational causal analysis} \cite{pearl2009causal}. Causal inference aims to quantify the effect of a \textit{treatment variable} ($T$) on an \textit{outcome variable} ($O$). For instance, to estimate the impact of disability on annual wage. The gold standard for establishing causality is through \textit{randomized controlled trials}. In such experiments, a population is randomly divided into two groups: the \textbf{treated group}, which receives the intervention ($\text{do}(T = 1)$ for a binary treatment), and the \textbf{control group}, which does not ($\text{do}(T = 0)$). This random assignment ensures that, on average, all other factors are balanced between the groups, isolating the effect of the treatment.
A widely used metric for estimating causal effects is the \textbf{Average Treatment Effect (ATE)}, which measures the difference between the average outcome observed in the treated and in the control group~\cite{rubin2005causal, pearl2009causal}.
\begin{equation}
    {\small ATE(T,O) = \mathbb{E}[Y \mid \text{do}(T=1)] -  
    \mathbb{E}[Y \mid \text{do}(T=0)]}
\label{eq:ate}
\end{equation}
The do-operator $\text{do}(T=1)$ formalizes the idea of manipulating the treatment, distinguishing between mere observation where $T=1$ and actively setting $T=1$ to analyze its causal impact.}


\revc{Randomized controlled experiments are often impractical or unethical in real-world scenarios. For instance, we cannot randomly assign disabilities to individuals to study its effect on wage. In such cases, \emph{observational causal analysis} provides a viable alternative, allowing for sound causal inference, albeit under additional assumptions.
The primary challenge in observational studies arises from \emph{confounding factors}. These are attributes that influence both the treatment assignment and the outcome, thereby obscuring the true causal effect. Consider our example of understanding the causal effect of disability on annual wage (Example~\ref{ex:acs}). Unlike a randomized experiment, disability status in this dataset is not randomly assigned. Instead, it is often correlated with other attributes such as gender and age. These attributes act as confounders because they can influence both the likelihood of having a disability and annual wage.}

\revc{To address confounding variables, Pearl's causal model offers a framework for obtaining unbiased causal estimates by accounting for these confounding attributes, denoted as $\mathbf{Z}$
. This approach relies on key assumptions, including the unconfoundedness and overlap assumptions which we do not elaborate on here. 
Given a set of confounding variables $\mathbf{Z}$, the ATE can be adjusted to the following form:
\begin{flalign}
& ATE(T,O) {=} \mathbb{E}_Z \left[\mathbb{E}[O \mid T{=}1, \boldsymbol{Z} {=} z] {-}
\mathbb{E}[O \mid T{=}0, \boldsymbol{Z} {=} z] \right] \label{eq:conf-ate}
\end{flalign}
This adjusted ATE, as shown in Equation (\ref{eq:conf-ate}), can be estimated directly from observed datasets.}

\revc{In this work, we assume that the user provides a sufficient set of confounding variables $\mathbf{Z}$. In practice, such a set can be identified using a causal discovery algorithm (e.g., \cite{zanga2022survey,glymour2019review}) and applying criteria such as Pearl's backdoor criterion~\cite{pearl2009causal}, or alternatively, through domain expertise.}

\smallskip
\paratitle{ATE Estimators}  
The ATE in Eq.~\ref{eq:conf-ate} can be estimated in high-dimensional settings using standard methods such as \emph{linear regression}~\cite{ding2018causal} and \emph{inverse propensity weighting (IPW)}~\cite{rosenbaum1983central}.  
We focus on these two estimators because their algebraic structure makes them especially amenable to the incremental updates required for efficient repair.  
Linear regression models the outcome as a linear function of the treatment and confounders and reads the ATE from the treatment coefficient.  
IPW estimates the probability of treatment (the \emph{propensity score}, typically via logistic regression) and re-weights each unit by the inverse of this probability to correct for treatment imbalance.  
While these methods rely on simplifying assumptions, our incremental downdate techniques are estimator-agnostic and can be readily transferred to more sophisticated methods such as doubly robust estimation~\cite{bang2005doubly} or double machine learning~\cite{chernozhukov2016double}

\section{Problem Formulation}
\label{sec:problem}

 

We consider a single-relation database over a schema $\attrset$. The schema is a vector of attributes $\attrset {=} (A_1, {\ldots}, A_m)$, where each $A_i$ is associated with a domain $\dom(A_i)$, which can be categorical or continuous. 
A database instance \db\ populates the schema with a set of tuples of the form $t {=} (a_1, \ldots, a_m)$ where $a_i {\in} \dom(A_i)$. 
We use bold letters to represent a subset of attributes $\mathbf{A} \subseteq \attrset$, and an uppercase letter to represent a subset of tuples from the database instance $\Gamma \subseteq \db$.


Consider a binary\footnote{For simplicity, we consider a binary treatment variable. However, our framework can be generalized to handle non-binary variables as well.} treatment variable $T \in \attrset$, and a numerical outcome variable $O \in \attrset$. The causal effect, determined by ATE (\cref{eq:conf-ate}), of $T$ on $O$ when evaluated on the database instance $\db$ is denoted by $ATE_{\db}(T,O)$,
where $\db$ indicates the causal effect was estimated over the database instance $\db$. We assume that the user provides a sufficient set of confounding variables to control for. For simplicity, we do not explicitly refer to them, although they are considered when estimating the causal effect of $T$ on $O$.

If the ATE of $T$ on $O$ is substantially different from the expected causal effect, it indicates a \emph{causal inconsistency} in the database instance $\db$. Our objective is to identify a set of tuples $\Gamma \subseteq \db$ responsible for this discrepancy, as they are the candidate tuples contributing to the shift.

We assume the presence of a desired causal effect, denoted as $ATE_d$, given by, e.g., a domain expert, relevant literature, or estimated over other dataset versions. Given an error bound $\epsilon > 0$,
we assume that $ATE_{\db}(T,O) \notin [ATE_d - \epsilon, ATE_d + \epsilon]$ (i.e., the current ATE is outside the target range).
Our goal is to find a minimal size set of tuples $\Gamma \subseteq \db$ such that, after removing $\Gamma$ from $\db$ we bring the database to a state where the causal effect of $T$ on $O$ is within the range $[ATE_d - \epsilon, ATE_d + \epsilon]$.
The set $\Gamma$ is the tuples responsible for the causal inconsistency. 

If $\Gamma$ is small relative to the dataset size, this indicates that the ATE can be shifted into the desired range by removing only a small fraction of the data. Examining these tuples allows the user to identify data regions that strongly influence the causal estimate and to assess the sensitivity of the causal conclusions to small changes in the data.

We define the Cardinality Repair for Causal Effect Targeting (abbreviated as \probName) problem.

\begin{problem}[\probName]
\label{prob:problem_def}
Let $\db$ be a database instance over a schema $\attrset$. We are given a binary treatment variable $T \in \attrset$, and an outcome variable $O \in \attrset$. 
Let $ATE_{\db}(T,O)$ denote the average treatment effect of $T$ on $O$, as estimated over the database $\db$.
Given a desired causal estimate $ATE_d$ and a threshold $\epsilon > 0$, find a minimal size subset of tuples $\Gamma \subseteq \db$ s.t:
$ATE_{\db \setminus \Gamma}(T,O) \in [ATE_d - \epsilon, ATE_d + \epsilon]$.

\end{problem}



\reva{Our framework assumes a sufficient set of observed confounding variables, in contrast to classical sensitivity analysis, which explicitly addresses unobserved confounders. Prior work typically models hidden confounding or measurement error using strong parametric bias models (e.g., additive-noise or multiplicative-bias frameworks) while assuming a clean dataset. By contrast, our approach is designed to handle messy data, including duplicates, outliers, and corrupted subsets, thereby relaxing assumptions about data quality.}


We consider two data repair settings, both restricted to tuple deletions, with the goal of removing the smallest possible subset. The first setting allows removal of arbitrary tuples, following cardinality repair research for database constraints~\cite{afrati2009repair,miao2020computation}. The second, more constrained setting, removes an entire subpopulation defined by a pattern, aiming to minimize the size of that subpopulation. For both settings, we analyze the computational complexity of the corresponding optimization problems and propose efficient algorithms to solve them.

\paragraph*{Tuple Removal}
\label{subsec:tuple_deletion}
We first focus on cardinality repair~\cite{afrati2009repair,miao2020computation}, which finds a minimal set of tuples to remove to achieve a target causal effect. 
We show that the \probName\ problem is NP-complete 
via a reduction from the SUBSET-SUM problem \cite{kleinberg2006algorithm}.

\begin{proposition}
\label{prop:problem_is_np_hard}
The \probName\ problem is NP-complete.
\end{proposition}


In Section~\ref{subsec:tuple_algo}, we introduce an efficient greedy heuristic algorithm to address this problem.

\paragraph*{Pattern Removal}
\label{sec:patterns}
In this part, we consider the removal of a subpopulation. 
To specify a subpopulation, we use
{\em patterns}~\cite{roy2014formal,wu2013scorpion,lin2021detecting,youngmann2024summarized} that comprise conjunctive predicates on attribute values. Formally, 
consider a database $\db$ with attributes $T,O,$ $A_1,{\ldots}, A_n$. Assume the domain of $A_i$ is $\dom(A_i)$.
A \emph{conjunctive pattern} $\pattern$ is a formula of the form $\bigwedge_{i\in I} A_i=a_i$ where $I\subseteq \{1,\ldots, n\}$ and $a_i\in \dom(A_i)$ for all $i\in I$. 
We only consider patterns that do not refer to the outcome variable $O$.

Our patterns are restricted to conjunctions of equality predicates, in line with past work on explanations that deem such predicates understandable~\cite{el2014interpretable,roy2015explaining,agmon2024finding}. The advantage of using such patterns lies in their ability to \textit{explain} the removed subset. We leave for future work the consideration of a richer class of patterns, including continuous variables and disjunction.



For a pattern $\pattern = (A_1 {=} a_1 {\wedge} {\ldots} {\wedge} A_l {=} a_l)$, we denote by $\pattern(\db)$ the set of tuples form $\db$ that satisfy $\pattern$. That is, $\pattern(\db)=\{t\mid \bigwedge_{i =  1}^l t[A_i]{=}a_i, t \in \db\}$.
In \emph{pattern deletion}, we delete from $\db$ the set of tuples $\pattern(\db)$ for some pattern $\pattern$.
Our goal is to identify the smallest subpopulation in terms of tuples from $\db$, defined by some pattern $\pattern$, such that after removing $\pattern(\db)$ from $\db$, the estimated ATE is within the target range. More formally,

\begin{problem}[\probName\ - pattern removal]
\label{prob:patterns}
Let $\db$ be a database instance over a schema $\attrset$. We are given a binary treatment $T \in \attrset$, and an outcome $O \in \attrset$. 
Let $ATE_{\db}(T,O)$ denote the causal effect of $T$ on $O$, as estimated over the database $\db$.
Given a desired causal estimate $ATE_d$ and a threshold $\epsilon > 0$, find a minimal size subpopulation defined by a pattern $\pattern$  such that:
$$ATE_{\db \setminus \pattern(\db)}(T,O) \in [ATE_d -\epsilon, ATE_d + \epsilon]$$
\end{problem}

\revc{In order to analyze the complexity of the problems, we implicitly refer to their decision-problem versions, where an upper bound on the size of the removed subpopulation is given. This analogy is sound in the sense that there is a polynomial-time solution to either Problem \ref{prob:problem_def} or Problem \ref{prob:patterns} if and only if there is a polynomial-time solution for the decision version.}

We can show that Problem \ref{prob:patterns} is NP-complete, by reducing it from the \probName\ problem (Problem \ref{prob:problem_def}) where any tuple can be deleted from the data (Proposition \ref{prop:problem_is_np_hard}).  

\begin{proposition}
\label{prop:problem_is_np_hard_pattern}
Problem \ref{prob:patterns} is NP-complete.
\end{proposition}

\revc{We remark that for both problems, membership in NP is witnessed by listing the subpopulation that is removed (for Problem~\ref{prob:problem_def}) and by the pattern removed (for Problem~\ref{prob:patterns}), where we can verify the solution in polynomial time by removing the subpopulation and computing the new ATE.}
We note that both Problem \ref{prob:problem_def} and Problem \ref{prob:patterns} remain hard even in the case where there are no confounding variables. 
In light of this observation, we have the following in particular.

\begin{proposition}
Problem \ref{prob:problem_def} and Problem \ref{prob:patterns} are NP-complete
for AVG and ATE.  
\end{proposition}


\section{The \sysName\ Framework}
\label{sec:algo}

In this section, we introduce the two underlying algorithms of \sysName\ - one for tuple-based removal (\algoNameTuple), and the second for the restricted setting where only a subpopulation can be removed (\algoNamePattern). Both algorithms adopt a heuristic approach to efficiently find a solution, motivated by the fact that the underlying optimization problems are NP-complete. As we demonstrate in our experimental evaluation (Section~\ref{sec:exp}), these algorithms scale well to large, high-dimensional datasets and consistently produce solutions that are smaller in size compared to existing and naive solutions.

\paragraph*{On ILP and SAT formulations}  
\revb{A natural question is whether our problem could be encoded as an Integer Linear Program (ILP) or SAT instance, similar to the Generalized Deletion Propagation (GDP) framework~\cite{makhija2025integer}. However, the key constraint 
$ATE_\db(T,O) \in [ATE_d - \epsilon,\, ATE_d + \epsilon]$ 
is inherently nonlinear. Even without confounders, the ATE is defined as the difference between two means, making it nonlinear in the decision variables. When confounders are included, the ATE is estimated, e.g., via linear regression, whose coefficients depend nonlinearly on the dataset (e.g., through matrix inversion). 
A possible linear relaxation assumes fixed numbers of tuples removed from each subgroup, which makes the constraint linear but requires solving a separate ILP for each combination of removal sizes, an infeasible approach as dimensionality grows. While a reduction to SAT is theoretically possible, it would require encoding nonlinear arithmetic and is computationally impractical. 
Therefore, instead of relying on ILP or SAT solvers, our methods focus on a practically scalable algorithmic design. For completeness, we also compare against a naive brute-force baseline that computes the exact optimum (see Section~\ref{sec:exp}). }



\subsection{\sysName\ for Tuple Removal}
\label{subsec:tuple_algo}

A naive brute-force approach would evaluate every possible tuple subset for removal, checking whether the resulting ATE falls within the target range. However, there can be exponentially many ways of shifting the ATE through tuple deletions (as is the case for functional dependencies~\cite{livshits2020computing}). To this end, we introduce \algoNameTuple, a greedy, optimized algorithm that iteratively removes tuples based on their estimated influence on the ATE, prioritizing those with the highest impact.

Inspired by ``leave-one-out''
influence in causal counterfactuals~\cite{bae2022if,pearl2009causal}, we define the influence of a tuple as follows. 

\begin{definition}[Tuple Influence]
The influence of a tuple $t \in \db$ on the ATE of $T$ on $O$ is defined as:
\[
\texttt{influence}(t) = ATE_{\db}(T, O) - ATE_{\db \setminus \{t\}}(T, O)
\]
where $T$ and $O$ denote the treatment and outcome variables, respectively, and $ATE_{\db}(T, O)$ represents the ATE of $T$ on $O$ computed over the dataset $\db$, while adjusting for confounding variables (omitted from the notation for clarity).
\end{definition}
A positive influence score indicates that removing $t$ from $\db$ would decrease the ATE, while a negative score indicates that its removal would increase the ATE.

A simple greedy approach iteratively removes the most beneficial tuple -either with the highest positive or most negative influence - based on the current ATE and the target range $[ATE_d - \epsilon, ATE_d + \epsilon]$. After each removal, influence scores are recomputed, and the process repeats until the ATE falls within the target range.
To reduce the computational burden of per-tuple ATE influence calculations while preserving structural diversity, we employ a two-stage cluster-based sampling strategy: an initial clustering-based representative selection, followed by an iterative, cluster-aware removal process.


\noindent
\textbf{Clustering.}  
Using the confounding variables $\mathbf{Z}$ taken into account for the ATE estimation, the treatment variable $T$, and the outcome $O$, we apply the $k$-means clustering algorithm~\cite{macqueen1967some} to form $k$ clusters. If the user does not specify $k$, we set  
\[
\small
  k = \max\Bigl(5,\;\min\bigl(\lfloor\sqrt{n}\rfloor,\;\frac{n}{10}\bigr)\Bigr)
\]
where $n = |\db|$, ensuring \(5 \le k \le n\),  that is, the number of clusters
$k$ is neither too small nor too large.


Within each cluster \(\mathcal{C}_k\) , we select \(s\) representative tuples (by default, we set \(s=2\)).  Let
\[
\small
  \mathcal{R}_k =
  \begin{cases}
    \mathcal{C}_k, & |\mathcal{C}_k|\le s\\
    \{\,r_{k,1},r_{k,2}\}, & \text{otherwise}
  \end{cases}
\]
where \(r_{k,1}\) is the tuple closest to the cluster centroid, and \(r_{k,2}\) (and additional reps if \(s>2\)) are chosen at evenly spaced percentiles of the distance-to-centroid distribution (e.g.\ 25th, 75th). We compute the influence score for all representative tuples, then assign each non-rep point to its nearest representative. 
This reduces the number of influence computations at each iteration from \(n\) to at most \(k\cdot s\).

\noindent
\textbf{Iterative Cluster-Based Sampling.}  
During each iteration \(i\), let $\mathcal{A}_i$ denote the set of indices of tuples still in $\db$.
For each cluster \(\mathcal{C}_k\), we define the available subset
\(\mathcal{A}_k^{(i)} = \mathcal{C}_k \cap \mathcal{A}_i\) and set
  $m_k = \min\bigl(5,\,|\mathcal{A}_k^{(i)}|\bigr)$.
  
We draw $m_k$ tuples from the cluster $ \mathcal{C}_k$
\(\{i_{k,1},\dots,i_{k,m_k}\}{\subseteq} \mathcal{A}_k^{(i)}\)
uniformly at random. We then compute their influence score and set the cluster score as the average influence score:
\[
\small
  s_k = \frac{1}{m_k}\sum_{j=1}^{m_k}\texttt{influence}(i_{k,j})
\]
Let $d = \operatorname{sign}\bigl(ATE_d - ATE_{\db^i}(T,O)\bigr)$,
where $ATE_d$ is the target ATE value and $ATE_{\db^i}(T,O)$ is the current ATE (at the $i$-th iteration). We select the cluster with the largest \(s_k \cdot d \). 
Finlay, within this cluster, we remove the tuple with the highest influence score and continue to the next iteration. 



\smallskip 
Additional optimizations we implemented include: \\
\textbf{(1) Sampling}: To improve scalability on large datasets, we operate on a random sample of the data. In our experiments, for large datasets (e.g., the ACS dataset), we sampled 10\% of the data. We then run the \algoNameTuple\ algorithm on this sample to identify tuples for removal. To amplify the effect, we also remove their neighboring tuples in the full dataset, identified using a k-nearest neighbors algorithm using \emph{scikit-learn}'s~\cite{pedregosa2011scikit} implementation. Here, $k$ was set to $100$. \\ 
\textbf{(2) Periodic influence recomputation}: 
We update influence scores every 10 iterations instead of after each removal, significantly reducing runtime while maintaining quality, as a tuple’s influence typically changes only slightly per iteration.\\
\textbf{(3) Offline Sampling}: Sampling is performed at preprocessing if the algorithm is expected to run on a large-scale dataset, as it depends only on the attributes used and not on the specific causal query. Thus, the same sampled-subset can be reused across causal questions involving the same attributes.

\subsection{\sysName\ for Pattern Removal}
\label{subsec:pattern_algo}
Next, we present an efficient algorithm called \algoNamePattern\ for identifying a small subpopulation, defined by a pattern, whose removal shifts the ATE into a target range.
As discussed in Section~\ref{sec:patterns}, this optimization problem is NP-complete, which motivates the need for a heuristic approach.

\smallskip
\noindent
\textbf{Algorithm Overview}
A naive exhaustive search would enumerate all possible subpopulations (i.e., patterns) to identify the smallest one whose removal shifts the ATE into the desired range. However, this approach is computationally inefficient due to the large number of candidate subpopulations.
To address this, our algorithm performs bottom-up random walks over the subpopulation lattice, effectively exploring the space without full enumeration. During these walks, we maintain a dynamic weighting mechanism to reflect the influence of individual predicates on the ATE. This guides the search toward regions of the lattice where predicate combinations are more likely to find impactful subpopulations.
Together, these strategies enable our algorithm to identify small subpopulations with significant influence on the ATE efficiently.

\smallskip

The pseudocode for \algoNamePattern\ is shown in Algorithm~\ref{algo:subgroup}.
It begins by computing the ATE over the full dataset (line~2), then generates all non-empty subgroups defined by combinations of attribute-values over all attributes in
$\db$ besides the treatment and outcome (line~3). This is done in a single pass over the data, inspired by the first step of the Apriori algorithm~\cite{agrawal1994fast}. These most specific patterns correspond to the leaves of the pattern lattice and serve as starting points for random walks.
Each walk starts from a leaf node (line~5) and iteratively removes one predicate at a time (line~11). For each subgroup, the algorithm checks whether its removal moves the ATE into the target range. If such a subgroup is found, it is returned immediately; otherwise, the walk continues. To limit runtime, the algorithm performs at most 
$k$ random walks.

\smallskip
\noindent
\textbf{Dynamic Weighting Mechanism}: Instead of uniformly selecting a predicate to remove from a pattern (line~11 in Algorithm\ref{algo:subgroup}), we maintain a caching mechanism that tracks how often the removal of a predicate shifts the ATE toward the desired direction, as well as the magnitude of the shift. This information is then used to assign probabilities to each predicate, prioritizing those whose removal is more likely to guide the ATE into the target range.

\smallskip 
Additional optimizations we implemented include: \\
\textbf{(1) Early termination of random walks} if the current subgroup exceeds a predefined size threshold $\tau$, prompting the algorithm to abandon that path and initiate a new random walk to find small-size subpopulations. \\
\textbf{(2) Caching}: Previously evaluated patterns are cached to reduce redundant computation. \\
\textbf{(3) Sampling:}
To ensure scalability on large datasets, we uniformly sample 10\% of the data and run the algorithm on this subset. The identified pattern is then applied to remove all matching tuples from the full dataset. This optimization is used only for large datasets (e.g., ACS in  Section~\ref{sec:exp}).

\begin{algorithm}[t]
  \small
  \DontPrintSemicolon
  \SetKwInOut{Input}{input}\SetKwInOut{Output}{output}
  \LinesNumbered
  \Input{A dataset $\db$ with a treatment variable $T$ and an outcome $O$, a target ATE range defined by $ATE_d$ and $\epsilon$, and a number $k >0$. }
  \Output{A subgroup to be removed defined by a pattern $\pattern$ or indication of a failure.} \BlankLine
  \SetKwFunction{EstimateATE}{\textsc{EstimateATE}}
  \SetKwFunction{GetGroups}{\textsc{GetMostSpecificGroups}}
  \SetKwFunction{RemovePredicate}{\textsc{RemovePredicate}}
  \SetKwFunction{IsValidPair}{\textsc{IsValidPair}}
  \SetKwFunction{HasEdge}{\textsc{HasEdge}}
  \SetKwFunction{GetPredecessors}{\textsc{GetPredecessors}}
  \SetKwFunction{GetSuccessors}{\textsc{GetSuccessors}}

  \tcc{\textcolor{blue}{Initial ATE value}}
  $v \gets$ \EstimateATE($T,O,\db$)\\

    $\mathcal{G} \gets$ \GetGroups($T,O,\db$)\\
 \For{$i \in [1,k]$}{

 $\pattern_i \gets $ a random group from $\mathcal{G}$\\

 $v_i \gets $ \EstimateATE($T,O,\db\setminus(\pattern_i(\db))$)\\

 \If{$v_i \in [ATE_d -\epsilon, ATE_d + \epsilon]$}{\Return $\pattern_i$}
 \While{$\pattern_i$ is not empty}{
  \tcc{\textcolor{blue}{A random step on the pattern lattice}}
$\pattern_i \gets \RemovePredicate(\pattern_i)$\\

 $v_i \gets $ \EstimateATE($T,O,\db\setminus(\pattern_i(\db))$)\\

 \If{$v_i \in [ATE_d -\epsilon, ATE_d + \epsilon]$}{\Return $\pattern_i$}
 
 }

         }
  \Return No solution was found

  \caption{The \algoNamePattern\ Algorithm}\label{algo:subgroup}
\end{algorithm}







\section{Incremental ATE Updates}
\label{sec:optimizations}
The main computational bottleneck of \algoNameTuple\ and \algoNamePattern\ lies in the repeated ATE calculation. To this end, we propose optimizations that estimate the ATE after removing data subsets without retraining the model. We focus on two widely used ATE estimators: (1) linear regression~\cite{imbens2015causal} and (2) inverse propensity weighting (IPW)~\cite{imbens2015causal}.
For linear regression, we leverage its closed-form solution to provide both exact and faster approximate ATE updates. For IPW, we build upon an existing machine unlearning technique~\cite{mahadevan2021certifiable}, to avoid full model retraining and yield updated propensity scores in constant time per deletion.


\subsection{ATE Update for Linear Regression}
\label{subsec:liner_unlearning}

We begin with the standard linear regression model:
{\setlength{\abovedisplayskip}{4pt}
 \setlength{\belowdisplayskip}{4pt}
\begin{equation}
  \mathbf{o} = X\boldsymbol{\beta} + \boldsymbol{\varepsilon}
\end{equation}
}
where $X \in \mathbb{R}^{n \times m}$ is the design matrix (one row per data record), $\mathbf{o} \in \mathbb{R}^{n}$ is the vector of outcomes, $\boldsymbol{\beta} \in \mathbb{R}^{m}$ are the model coefficients, and $\boldsymbol{\varepsilon}$ is a noise vector. The ATE of a treatment $T$ on $O$ is defined as the coefficient of $T$, where $T \in \mathbf{X}$, and $\mathbf{X}$ also includes all the confounding variables $\mathbf{Z}$.
To learn the model parameters $\boldsymbol{\beta}$, we use ordinary least squares, which minimizes the squared error between the predicted and actual values. This yields the following equation:
\begin{equation}
  X^{\top}X\boldsymbol{\beta} = X^{\top}\mathbf{o}
\end{equation}
This gives a closed-form expression for the optimal coefficients:
\begin{equation}
  \boldsymbol{\beta} = \left(X^{\top}X\right)^{-1}X^{\top}\mathbf{o}
\end{equation}

Now, suppose we remove a subset $\Gamma \subseteq \db$ of $r$ tuples, namely $|\Gamma| = r$. Let $X_{\mathrm{rmv}} \in \mathbb{R}^{r \times m}$ and $\mathbf{o}_{\mathrm{rmv}} \in \mathbb{R}^r$ denote the corresponding removed tuples of $X$ and $\mathbf{o}$. The updated data becomes:
\[
X_{\mathrm{new}} = X \setminus X_{\mathrm{rmv}}, \quad 
\mathbf{o}_{\mathrm{new}} = \mathbf{o} \setminus \mathbf{o}_{\mathrm{rmv}}
\]

Our goal is to obtain
$\mathbf{\beta}_{\mathrm{new}}$ defined by
\begin{equation}
    \label{eq:new_normal_eq}
    \mathbf{\beta}_{\mathrm{new}}
  \;=\;
  \bigl(X_{\mathrm{new}}^{\top}X_{\mathrm{new}}\bigr)^{-1}
  X_{\mathrm{new}}^{\top}\mathbf{o}_{\mathrm{new}}
\end{equation}
\emph{without} recomputing the inverse from scratch.

\vspace{0.5\baselineskip}
\noindent
For convenience, define
\[
  A := X^{\top}X,
  \qquad
  \Delta := X_{\mathrm{rmv}}^{\top}X_{\mathrm{rmv}},
  \qquad\text{so that}\qquad
  A_{\mathrm{new}} = A - \Delta
\]

\paragraph*{\textbf{Exact Update}}
\label{subsec:woodbury}
To get an exact update, we use the Woodbury matrix identity~\cite{hager1989updating} to get the coefficient update:

{\setlength{\abovedisplayskip}{1pt}
 \setlength{\belowdisplayskip}{1pt}
\begin{equation}
\small
  \label{eq:beta_woodbury}
  \mathbf{\beta}_{\mathrm{new}}
  {=}
  \bigl[
      A^{{-}1}
      {+}
      A^{{-}1}U
      \bigl(I_r {-} U^{\top}A^{{-}1}U\bigr)^{{-}1}
      U^{\top}A^{{-}1}
    \bigr]
    \bigl(X^{\top}\mathbf{y} {-} X_{\mathrm{rmv}}^{\top}\mathbf{o}_{\mathrm{rmv}}\bigr)
\end{equation}
}

\paragraph*{\textbf{Approximate Update}}
\label{subsec:neumann}


When the number of tuples to be removed is small, specifically, when the norm \(\lVert \Delta A^{-1} \rVert < 1\), we can use the \emph{Neumann series}~\cite{yosida2012functional} to approximate the inverse of the matrix \(A - \Delta\). By setting \(Z = \Delta A^{-1}\), we approximate \((I - \Delta A^{-1})^{-1}\) by summing a finite number of terms:
{\setlength{\abovedisplayskip}{1pt}
 \setlength{\belowdisplayskip}{1pt}
\[
(A - \Delta)^{-1} \approx A^{-1} \sum_{k=0}^K (\Delta A^{-1})^k 
\]
}


We set $K=1$, taking the first two terms:
{\setlength{\abovedisplayskip}{1pt}
 \setlength{\belowdisplayskip}{1pt}
\[
\mathbf{\beta}_{\mathrm{new}} \approx \bigl(A^{-1} + A^{-1} \Delta A^{-1}\bigr) \bigl(X^\top \mathbf{o} - X_{\mathrm{rmv}}^\top \mathbf{o}_{\mathrm{rmv}}\bigr)
\]
}

This method lets us efficiently update the solution without recalculating the full inverse. As we show in Section \ref{sec:exp}, this method is faster than the exact method and gives accurate results when the number of tuples we remove is small.






\subsection{ATE Update for IPW}
\label{subsec:ipw_unlearning}
To estimate ATE using IPW method, each tuple is weighted by the inverse probability of receiving the treatment they received, as estimated by a propensity score model. The ATE is computed as the difference in the weighted averages of outcomes between treated and control groups.
Here, we build upon an unlearning method for logistic regression ~\cite{mahadevan2021certifiable}, proposing an incremental update method for the IPW ATE estimator.

We define \(\textbf{t}_i\) and \(\textbf{o}_i\) as the treatment and outcome values of the \(i\)-th unit (tuple), and $z_i$ as its assignment for the confounding variables $\mathbf{Z}$. Let \(p_i = \Pr(\textbf{t}_i = 1 \mid \mathbf z_i;\boldsymbol\theta)=\sigma(\mathbf z_i^{\top}\boldsymbol\theta)\)
be the propensity score of the $i$-th unit obtained from a $\ell_2$-regularised
logistic model with parameters \(\boldsymbol\theta\), where $\mathbf z_i$ are the confounder assignments of the $i$-th unit.
For a dataset $\db$ where $\db| = n$, the IPW estimator estimate the ATE of $T$ on $Y$ as follows:
\begin{equation}
\small
\widehat{\operatorname{ATE}}_{\text{IPW}}(T,O)
   = \frac{\sum_{i=1}^{n} \dfrac{\textbf{t}_i \textbf{o}_i}{p_i}}
          {\sum_{i=1}^{n} \dfrac{\textbf{t}_i}{p_i}}
     \;-\;
     \frac{\sum_{i=1}^{n} \dfrac{(1-\textbf{t}_i)\textbf{o}_i}{1-p_i}}
          {\sum_{i=1}^{n} \dfrac{(1-\textbf{t}_i)}{1-p_i}}
\label{eq:ipw-ate}
\end{equation}

Suppose a subset \(D_{\text{rmv}}\subseteq \db\) is removed, where $|D_{\text{rmv}}|=r$. We need to update \(\boldsymbol\theta\) without
retraining the model from scratch and re-evaluate
Eq.~\eqref{eq:ipw-ate} efficiently.
To achieve this, we use the Fisher mini-batch incremental update algorithm from~\cite{guo2020certified}. The pseudocode is provided in the Appendix. This algorithm takes as input the dataset, the current model parameters, and the subset to be removed, and returns an estimate of the updated parameters.
Rather than performing multiple iterations of gradient descent, the algorithm relies on a single-step update. 
Finally, we use the updated parameters to recompute propensity scores and re-evaluate the ATE (Eq.~\eqref{eq:ipw-ate}).

\section{Experimental study}
\label{sec:exp}
We present an experimental evaluation that evaluates our \sysName\ effectiveness. We aim to address the following questions:  \textbf{Q1:} How does the quality of our proposed solutions compare to naive baselines and existing methods? \textbf{Q2:} How well do our algorithms scale when applied to large, high-dimensional datasets? \textbf{Q3:} To what extent does our proposed ATE incremental update optimizations improve performance? \textbf{Q4:} Is \sysName\ effective even when combined with alternative methods for computing ATE? \reva{\textbf{Q5:} How does \sysName\ compare to existing sensitivity analysis methods?}

\subsection{Experimental setup}
\label{subsec:exp_setup}
All experiments were conducted on a server with an Intel(R) Xeon(R) E5-2680 v3 CPU (2.50GHz, 12 cores) and 128GB of RAM. Our algorithms were implemented in Python3, and our code and the used datasets are publicly available in~\cite{code}. 

By default, we focus on the linear regression estimator for ATE computation, as described in Section \ref{subsec:liner_unlearning}.

\subsubsection*{Datasets \& examined causal questions}
We examine four commonly used datasets. The datasets' statistics and corresponding causal questions are summarized in Table~\ref{tab:datasets}.
\textbf{German Credit:}
This data~\cite{asuncion2007uci} contains details of bank account holders, including demographic and financial information, and credit risk scores. The causal query measures the impact of house ownership on the risk score. The goal is to shift this effect to zero (within a margin of $\pm$0.01), indicating no causal relationship. 
\textbf{Twins}: The Twins dataset~\cite{louizos2017causal} is a widely used benchmark for causal inference, containing data on same-sex twins. The treatment is being the heavier twin, and the outcome is mortality within the first year of life. The initial estimated ATE is -0.016, and our objective is to shift this effect to zero ($\pm$0.001).
\textbf{Stack Overflow:}
The Stack Overflow Developer Survey~\cite{stackoverflowreport} data contains responses from developers worldwide, covering topics such as professional experience, education, and salary-related information. The examined causal question is to quantify the effect of higher education on annual salary. The initial ATE is \$13,236. The goal is to reduce this effect by \$5,000 ($\pm$ \$100), resulting in a target ATE of \$8,236. 
\textbf{ACS:}
The American Community Survey (ACS)~\cite{ACS_Data} dataset is a nationwide survey conducted by the U.S. Census Bureau, providing detailed demographic, social, and economic data. 
The causal question examines the effect of not having a disability on annual wages. The initial ATE is \$8,774, meaning people without disabilities earn that much more on average compared to people with disability. 
The target ATE is set to \$12,000 ($\pm$\$500).


\begin{table*}[ht]
\centering
\footnotesize
\caption{Details of the datasets and corresponding causal questions we use for experiments and case studies.}
\label{tab:datasets}

\renewcommand{\arraystretch}{0.9}
\setlength{\tabcolsep}{4pt}

\resizebox{0.9\textwidth}{!}{
\begin{tabular}{lllp{22mm}p{14mm}p{70mm}ll}
\toprule
\textbf{Dataset}       & \textbf{\#Tuples} & \textbf{\#Atts} &\textbf{Treatment} & \textbf{Outcome} &\textbf{Confounding Variables} & \textbf{Org. ATE} & \textbf{Tar. ATE} \\ \midrule

German &1000 & 17& Owning a house  &Credit risk&Personal status, age&0.13& 0 ($\pm 0.01$)\\
\midrule

Twins &23,968 & 53& Heavier twin  &Mortality&Gestational age, birth weight, prenatal care, abnormal amniotic fluid, induced labor, gender, maternal marital status, year of birth, and total previous deliveries.& -0.016& 0 ($\pm 0.001$)\\
\midrule

SO     & 47,702 & 21& High Education &Annual Salary &Continent, gender, ethnicity&13236 &8236 ($\pm 100$)\\

\midrule

ACS &1,188,308 & 17 & Not having a disability & Annual wage  & Education, public health coverage, private health coverage, medicare for people 65 and older, insurance through employer, gender, age & 8774 & 12000 ($\pm 500$)\\
\bottomrule
\end{tabular}}

\vspace{2mm}
\end{table*}

\subsubsection*{Competing Baselines}
We consider the following baselines:


\noindent\textbf{\revb{OPT}}: \revb{We compare our algorithms against the optimal solution to analyze their quality and runtime trade-offs. We implemented OPT-tuple and OPT-pattern using exhaustive search to obtain the optimal solution.
}

\noindent\textbf{\topk}: This baseline evaluates the impact of each tuple in the data on the ATE and iteratively removes those with the highest influence until the target ATE is achieved. Unlike \algoNameTuple, the influence of each tuple is calculated only once at the start of the process.


\noindent\textbf{\scoded}~\cite{yan2020scoded}: \scoded\ is designed to identify the top-$k$ tuples that contribute most to the violation of a (in)dependence relationship between sets of variables. We use \scoded\ to identify tuples for removal in order to achieve a desired ATE. 
We evaluate both versions of \scoded, exact, and approximate. The value of $k$ is set to match the number of tuples removed by \topk.

\noindent
\reva{\textbf{ZamInfluence}~\cite{broderick2020automatic}:} \reva{ZamInfluence measures how much an estimate changes when a data fraction is removed, identifying deletions that flip its sign or significance. We compare against it in scenarios aiming to push the ATE toward 0.}

\noindent \textbf{Outlier Detection Baselines}: We compare \sysName\ with two widely used outlier detection methods: \textbf{Isolation Forest} (\ifoutlier)~\cite{liu2012isolation} and \textbf{Local Outlier Factor} (\lof)~\cite{cheng2019outlier}. We evaluate these algorithms in two ways: first, as standalone baselines by considering the ATE after their application; second, as preprocessing steps before \algoNameTuple, reporting the total number of removals needed to reach the target ATE. We use scikit-learn's implementation~\cite{pedregosa2011scikit}, setting the removal budget to match the scale of \topk.




For our \algoNameTuple\ and \algoNamePattern\ algorithms, as well as for the \topk\ baseline, we evaluate both their exact and approximate variants based on the ATE incremental update methods described in Section~\ref{subsec:liner_unlearning}.

The time cutoff for all algorithms was set to 10 hours. For the \algoNamePattern\ algorithm, we limited the size of the group to be removed to 20\% of the data, and the number of random walks considered was set to 1000.

\begin{table*}[ht]
\centering
\caption{Comparison of baseline methods. Results include the achieved ATE, whether the target range was met, the number of tuples removed, and runtime. \reva{Highlighted results correspond to the fewest tuples removed.}}
\label{tab:baseline_comparison}
\renewcommand{\arraystretch}{0.9}
\setlength{\tabcolsep}{4pt}

\resizebox{0.67\textwidth}{!}{%
\begin{tabular}{p{45mm}lcccc}
\toprule
\textbf{Dataset} & \textbf{Baseline} & \textbf{Achieved ATE} & \textbf{Hit Range} & \textbf{\#Removals} & \textbf{Time (s)} \\
\midrule
\rowcolor{german}
\multirow{12}{*}{} 
& \algoNameTuple\ (exact) & 0.006 ($\downarrow 95.4\%$) & \checkmark & \hl{\textbf{42}} (4.2\%) & 2.8 \\
\rowcolor{german}
& \algoNameTuple\ (approx) & 0.007 ($\downarrow 94.6\%$) & \checkmark & \hl{\textbf{42}}  (4.2\%)& 2.78 \\
\rowcolor{german}
& \topk\ (exact) & 0.006 ($\downarrow 95.4\%$) & \checkmark & \hl{\textbf{42}}  (4.2\%) & 0.36 \\
\rowcolor{german}
& \topk\ (approx) & 0.007 ($\downarrow 94.6\%$) & \checkmark & \hl{\textbf{42}}  (4.2\%)  & 0.36 \\
\rowcolor{german}
German Credit& SCODED (exact) & 0.02 ($\downarrow 84.6\%$) &$\times$  & \hl{\textbf{42}}  (4.2\%)& 0.08 \\
\rowcolor{german}
(org. ATE: 0.13, tar. ATE: 0 $\pm 0.01$) & SCODED (approx) & 0.08 ($\downarrow 38.5\%$) & $\times$ & \hl{\textbf{42}}  (4.2\%) & 0.03 \\
\rowcolor{german}
&\reva{ZamInfluence}&\reva{-0.025 ($\downarrow 292\%$)}&\reva{$\times$}&\reva{52 (5.2\%)}&\reva{0.02}\\

\rowcolor{german}
& IF + \algoNameTuple\ & 0.008 ($\downarrow 94.1\%$)  & \checkmark & 88  (8.8\%) & 2.8 \\
\rowcolor{german}
& LOF + \algoNameTuple\ & 0.008 ($\downarrow 94.1\%$)  & \checkmark & 89  (8.9\%)& 3.2 \\
\rowcolor{german}
& \algoNamePattern\ (exact) & 0.06 ($\downarrow 53.8\%$)  & $\times$ & 126  (12.6\%) & 50 \\
\rowcolor{german}
& \algoNamePattern\ (approx) & 0.07 ($\downarrow 46.1\%$)  & $\times$ & 196   (19.6\%)& 28 \\
\midrule

\rowcolor{twins}
\multirow{12}{*}{} 
& \algoNameTuple\ (exact) & -0.0009 ($\uparrow 94.4\%$) & \checkmark & 270 (1.1\%) & 92.7  \\
\rowcolor{twins}
& \algoNameTuple\ (approx) & -0.0009 ($\uparrow 94.4\%$) & \checkmark & 270 (1.1\%) & 99.9  \\
\rowcolor{twins}
& \topk\ (exact) & -0.0009 ($\uparrow 94.4\%$) & \checkmark & 367 (1.5\%) & 18.5 \\
\rowcolor{twins}
& \topk\ (approx) & -0.0009 ($\uparrow 94.4\%$) & \checkmark & 367 (1.5\%)  & 18.6 \\
\rowcolor{twins}
Twins& SCODED (exact) & -0.005 ($\uparrow 68.7\%$) &$\times$  & 367 (1.5\%)& 4.4 \\
\rowcolor{twins}
(org. ATE: -0.016, tar. ATE: 0 $\pm 0.001$)& SCODED (approx) & -0.015 ($\uparrow 6.2\%$) & $\times$ & 367 (1.5\%) & 4.02 \\
\rowcolor{twins}
&\reva{ZamInfluence}&\reva{-0.0005 ($\uparrow 99.7\%$)}&\reva{\checkmark} &\hl{\textbf{180}} \reva{$(0.75\%)$}& \reva{0.06}\\


\rowcolor{twins}
& IF + \algoNameTuple\ & -0.0009 ($\uparrow 94.4\%$) & \checkmark & 714 (2.9\%) & 93 \\
\rowcolor{twins}
& LOF + \algoNameTuple\ & -0.0009 ($\uparrow 94.4\%$) & \checkmark & 714 (2.9\%) & 88 \\
\rowcolor{twins}
& \algoNamePattern\ (exact) & -0.0005 ($\uparrow 96.8\%$) & \checkmark & 4692 (19.5\%) & 220 \\
\rowcolor{twins}
& \algoNamePattern\ (approx) & 0.0003 ($\uparrow 101.8\%$) & \checkmark & 4401 (18.3\%)  & 441 \\
\midrule

\rowcolor{so}
\multirow{12}{*}{} 
& \algoNameTuple\ (exact) & 8269 ($\downarrow 37.5\%$) & \checkmark & \hl{\textbf{67}} (0.14\%) & 45.1 \\
\rowcolor{so}
& \algoNameTuple\ (approx) & 8269 ($\downarrow 37.5\%$) & \checkmark & \hl{\textbf{67}} (0.14\%) & 43.7 \\
\rowcolor{so}
& \topk\ (exact) & 8269 ($\downarrow 37.5\%$) & \checkmark & \hl{\textbf{67}} (0.14\%) & 19.6 \\
\rowcolor{so}
& \topk\ (approx) & 8269 ($\downarrow 37.5\%$) & \checkmark & \hl{\textbf{67}} (0.14\%)  & 19.8 \\
\rowcolor{so}
Stack Overflow & SCODED (exact) & 12625 ($\downarrow 4.6\%$) &$\times$  & \hl{\textbf{67}} (0.14\%) & 17.7 \\
\rowcolor{so}
(org. ATE: \$13.2k, tar. ATE: & SCODED (approx) & 16255 ($\uparrow 22.6\%$)  & $\times$ & \hl{\textbf{67}} (0.14\%) & 10.8 \\
\rowcolor{so}
\$8.2k $\pm \$100$)& IF + \algoNameTuple\ & 8274 ($\downarrow 37.4\%$) & \checkmark & 465 (0.95\%) & 58 \\
\rowcolor{so}
& LOF + \algoNameTuple\ & 8290 ($\downarrow 37.3\%$) & \checkmark & 440 (0.94\%) & 44 \\
\rowcolor{so}
& \algoNamePattern\ (exact) & 8250  ($\downarrow 37.6\%$) & \checkmark & 3744 (7.8\%) & 115 \\
\rowcolor{so}
& \algoNamePattern\ (approx) & 8140 ($\downarrow 38.5\%$) & \checkmark & 5673 (11.9\%) & 141 \\

\midrule

\rowcolor{acs}
\multirow{12}{*}{\textbf{ACS}} 
&\algoNameTuple\ (exact) &11512	($\uparrow$ 31.1\%) & $\checkmark$ &\hl{\textbf{8400}} (0.7\%) &1149  \\
\rowcolor{acs}
& \algoNameTuple\ (approx) &11510 ($\uparrow$ 31.1\%) &\checkmark  & \hl{\textbf{8400}} (0.7\%) & 1102 \\
\rowcolor{acs}
& \topk\ (exact) & 11503 ($\uparrow$ 31.1\%) &\checkmark & 9600 (0.8\%)  &649  \\

\rowcolor{acs}
ACS & \topk\ (approx) & 11501 ($\uparrow$ 31.1\%) & \checkmark &9600 (0.8\%)  &647  \\
\rowcolor{acs}
(org. ATE: \$8.7k, tar. ATE: & IF + \algoNameTuple\ & 11583 ($\uparrow$ 31.8\%) & \checkmark &15582 (1.42\%)&867  \\
\rowcolor{acs}

\$12k $\pm \$500$)& LOF + \algoNameTuple\ & 11561 ($\uparrow$31.7\%) & \checkmark & 19582 (1.64\%)& 2072 \\
\rowcolor{acs}
& \algoNamePattern\ (exact) & 11869 ($\uparrow$ 35.3\%)					 & \checkmark &141168 (11.9\%)  & 229  \\
\rowcolor{acs}
& \algoNamePattern\ (approx) & 13629 ($\uparrow$ 55.3\%) &$\times$  &211497 (17.8\%)&  3063 \\

\bottomrule
\end{tabular}}

\end{table*}



\subsection{Quality Evaluation}
\label{subsec:ex_quality}
For each examined causal question, we ran all algorithms. For \algoNamePattern, which is a randomized algorithm, we ran it three times and report the average runtime and the best result achieved. 
In all cases, the outlier detection methods (\ifoutlier\ and \lof) failed to identify tuples that significantly affect the ATE. As a result, their standalone results are omitted from presentation; we report only their performance when used as a preprocessing step for our \algoNameTuple\ algorithm. 
The results are shown in Table~\ref{tab:baseline_comparison}. In the Appendix, we show heatmaps indicating the overlap between different solutions.

\noindent\underline{\textbf{Results Summary}}: 
When only a few hundred tuples need to be removed, \topk\ and \algoNameTuple\ yield nearly identical solutions, with \topk\ being faster. As the required removals increase, \topk\ degrades, while \algoNameTuple\ finds smaller solutions, \reva{highlighting the importance of recomputing influence scores}. \reva{ZamInfluence identified small subsets to shift the ATE to zero, but cannot target other ATE values.}
\algoNamePattern\ consistently identified compact influential subpopulations but removed more tuples overall. The approximate update method was faster and comparable to the exact one for small removals but less effective for larger ones. \scoded\ failed to reach the target ATE in all cases, and \ifoutlier\ and \lof\ removed tuples with minimal ATE impact, offering no benefit even as preprocessing for \algoNameTuple.

\noindent
\underline{German Credit}:
\topk\ and \algoNameTuple\ identified the same 42 tuples (4\% of the data) whose removal reduced the ATE from 0.13 to 0.007 - a 94.6\% drop, with no difference between exact and approximate update methods. \algoNamePattern\ failed to find any subgroup smaller than 20\% of the data that achieved the target ATE, though it discovered that removing a specific group (foreign workers with low checking balances and long residence) reduced the ATE by over 50\%. Outlier detection using IF and LOF did not help reduce the number of tuples \algoNameTuple\ needed to remove, as the same 42 deletions were still required. \reva{Tuples identified by \scoded\ and ZamInflunce did affect the ATE, but their removal alone did not bring it into the target range}.

\vspace{1mm}
\noindent
\underline{Twins}: Both \topk\ and \algoNameTuple\ achieved the target ATE, but \algoNameTuple\ did so more efficiently, removing only 270 tuples compared to 367, underscoring the importance of recalculating influence scores. Approximate and exact ATE update methods yielded similar results. \reva{Here, ZamInfluence found the smallest solution that achieves the target ATE, but it is only applicable when the target ATE is zero.} \algoNamePattern\ found relatively small subpopulations that achieved the target ATE, with the approximate variant identifying a group -- first-born twins with birth weights between 1,786g and 1,999g -- about 18\% of the data whose removal reversed the ATE sign from negative to positive, increasing it by 102\%.
\scoded\ again failed to identify a solution, and outlier detection methods did not reduce the number of required deletions, indicating that the tuples selected by our method are not merely outliers. 

\vspace{1mm}
\noindent
\underline{Stack Overflow}:
Here again \topk\ and \algoNameTuple\ identified the same 67 tuples (just 0.14\% of the data), whose removal decreases the ATE by 37.5\%. \algoNamePattern\ successfully identified impactful subpopulations; notably, the exact variant found that removing 3,744 individuals (7.8\% of the data) -- male respondents from the US aged 25-34, not students, with no dependents -- lowered the ATE by 38\%.
\scoded\ again failed to find tuples that shift the ATE into the target range, with its approximate variant even increasing the ATE. Here again using either IF or LOF as a preprocessing step for \algoNameTuple\ had no effect on reducing its workload.

\vspace{1mm}
\noindent
\underline{ACS}: Here, both \topk\ and \algoNameTuple\ achieved similar target ATE values, but \algoNameTuple\ required fewer deletions: 8.4k tuples vs 9.6k, demonstrating its advantage when large removal is needed, as it re-evaluates tuple influence during the process while \topk\ does not. The exact \algoNamePattern\ algorithm identified a meaningful subpopulation (11.9\% of the data) whose removal raised the ATE by 35\%, while the approximate version overshot the target range due to accumulated error from removing a large number of tuples (221k). \scoded\ results are omitted, as it can be applied only when the goal is to push the ATE towards zero, indicating an independence relationship.

As discussed in Example~\ref{ex:acs}, achieving a comparable downward shift of the ATE with \algoNameTuple\ requires removing only 7,200 tuples (0.6\% of the data). In contrast, \algoNamePattern\ can only find a solution that removes 889,626 tuples (74\% of the data), resulting in an ATE of 5,900 (a 32.1\% decrease). This subpopulation consists of individuals with no disability who hold Medicare. This asymmetry reveals an important insight: while the ATE is highly sensitive to upward shifts, it appears significantly more robust to downward ones when the intervention involves removing entire subpopulations.

\begin{table}[ht]
\centering
\caption{Quality results for using the incremental update optimization for IPW.}
\label{tab:baseline_comparison_ipw}
\renewcommand{\arraystretch}{0.9}
\setlength{\tabcolsep}{4pt}

\resizebox{0.47\textwidth}{!}{%
\begin{tabular}{p{10mm}lcccccccccc}
\toprule
\textbf{Dataset} & \textbf{Algorithm} & \textbf{Achieved ATE} & \textbf{\# Removals} & \textbf{Time} \\
\midrule

\rowcolor{german}

German& \sysName\ (no incremental update) & 0.009 ($\downarrow95.2\%$)  & 35 (3.5\%) & 35.8 \\
\rowcolor{german}
    Credit& \sysName\ & 0.009 ($\downarrow95.2\%$)  & 35  (3.5\%)& 25.5 \\

\midrule

\rowcolor{twins}

Twins &  \sysName\ (no incremental update) & 0 ($\uparrow100\%$)  & 300 (1.2\%) &  832 \\
\rowcolor{twins}
& \sysName\ & 0 ($\uparrow100\%$)  & 300 (1.2\%) & 712  \\

\bottomrule
\end{tabular}
}

\end{table}

\subsubsection{Estimating ATE with IPW}
We revisit our case study using IPW to estimate the ATE and apply the corresponding incremental update optimization. Since ATE update with IPW is significantly slower than with the linear model, we focus on the smaller German Credit and Twins datasets, as others exceed the time limit. 
Results are shown only for \algoNameTuple, as \algoNamePattern\ and \topk\ show similar trends.
Table~\ref{tab:baseline_comparison_ipw} shows that \algoNameTuple\ yields very similar results across the two ATE estimators and removes almost the same tuples. For example, in German Credit, it removes 35 tuples (compared to 42 with the linear model), achieving nearly identical target ATEs. This suggests that tuples influential under one estimator are typically influential under the other - a trend further supported in Section~\ref{subsec:ate_estimators}.
We also observe that using the optimization reduces runtime by approximately 20\%, compared to not using the IPW incremental update. However, the linear regression-based optimization is consistently faster and produces similar outcomes.

\subsubsection{\revb{Comparison to the optimal solution}}
\label{exp:opt}
\revb{Running an exhaustive brute-force search to obtain the optimal solution proved infeasible even for the small German Credit dataset. To nevertheless assess how closely \sysName\ approximates the optimal solution, we designed controlled experiments using synthetic data where an upper bound on the number of deletions is known, thus bounding the search space OPT.}
\revb{We generated synthetic datasets of increasing sizes (up to 50 tuples for OPT-tuples, as larger instances exceeded our time cutoff), each containing a treatment variable, an outcome, and three confounders. The target ATE was defined as the ATE computed on each dataset with $\epsilon = 0$. We then injected 20\% noisy records into the data, making the number of noisy tuples an upper bound on the solution size. We compared \algoNameTuple\ against OPT-tuple, and \algoNamePattern\ against OPT-pattern. The results are depicted in Figure~\ref{fig:opt}. We observe that \algoNameTuple\ achieves solutions comparable in quality to those of the OPT-tuple, while being significantly faster. 
For \algoNamePattern, the quality of the produced solutions is also comparable to that of OPT-pattern, but the runtime differences are smaller. This behavior stems from the experimental setup: in this synthetic data, there was a single dominant optimal solution to identify, rather than multiple similar solutions, as in the other experiments. Consequently, \algoNamePattern\ required more time to find this solution, and in some cases, it failed to find it. In real-world scenarios, however, where numerous subpopulations may yield comparable effects, the heuristic nature of \algoNamePattern\ allows it to find a good solution within a limited time, albeit without formal guarantees on optimality.}

     
       

\begin{figure*}[htbp]
    \centering
    \begin{minipage}[t]{0.23\textwidth}
        \centering
        \includegraphics[height=2.4cm]{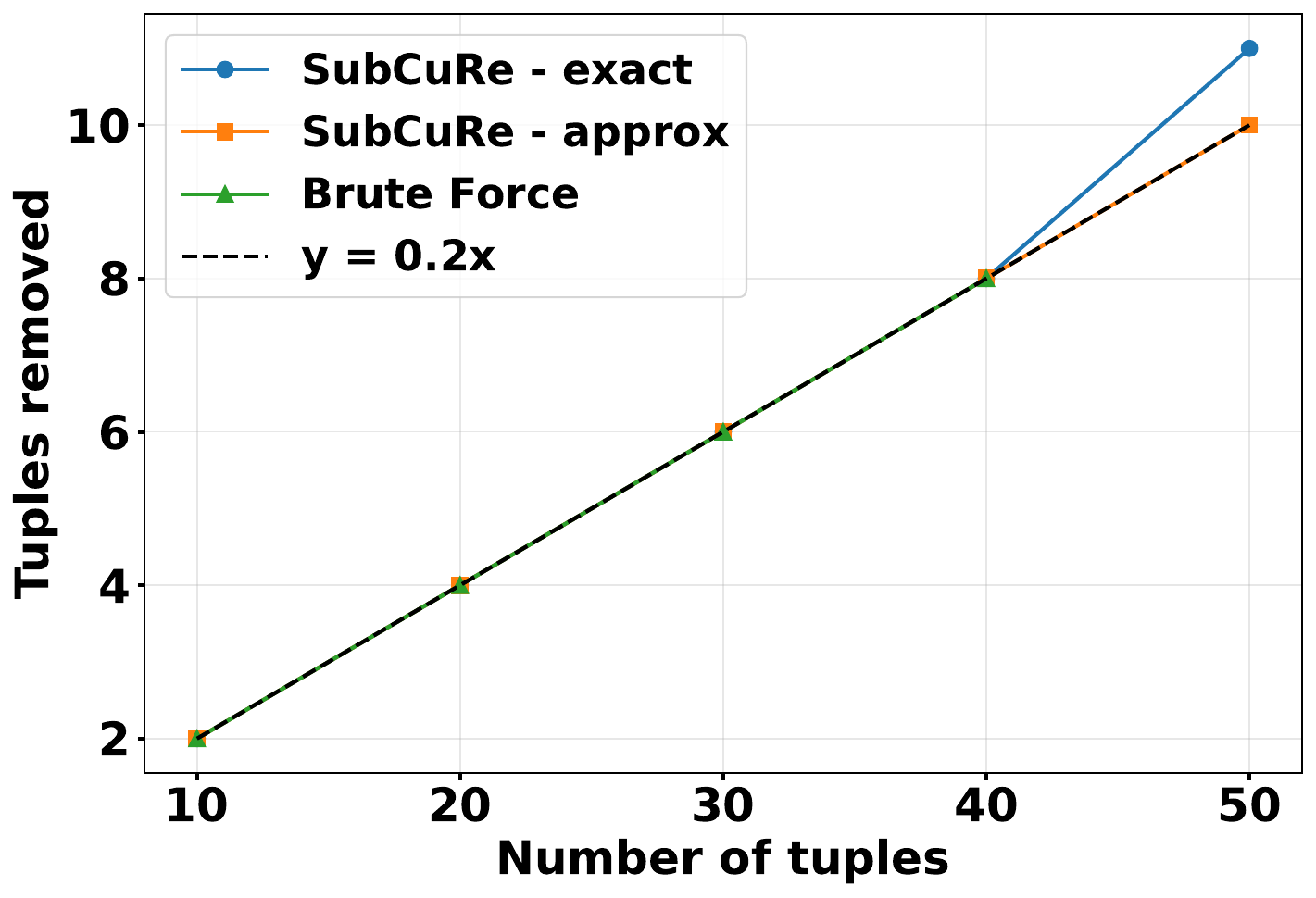}
        \caption*{\reva{(a) Quality - tuple}}
        \label{fig:image1}
    \end{minipage}
    \begin{minipage}[t]{0.23\textwidth}
        \centering
        \includegraphics[height=2.4cm]{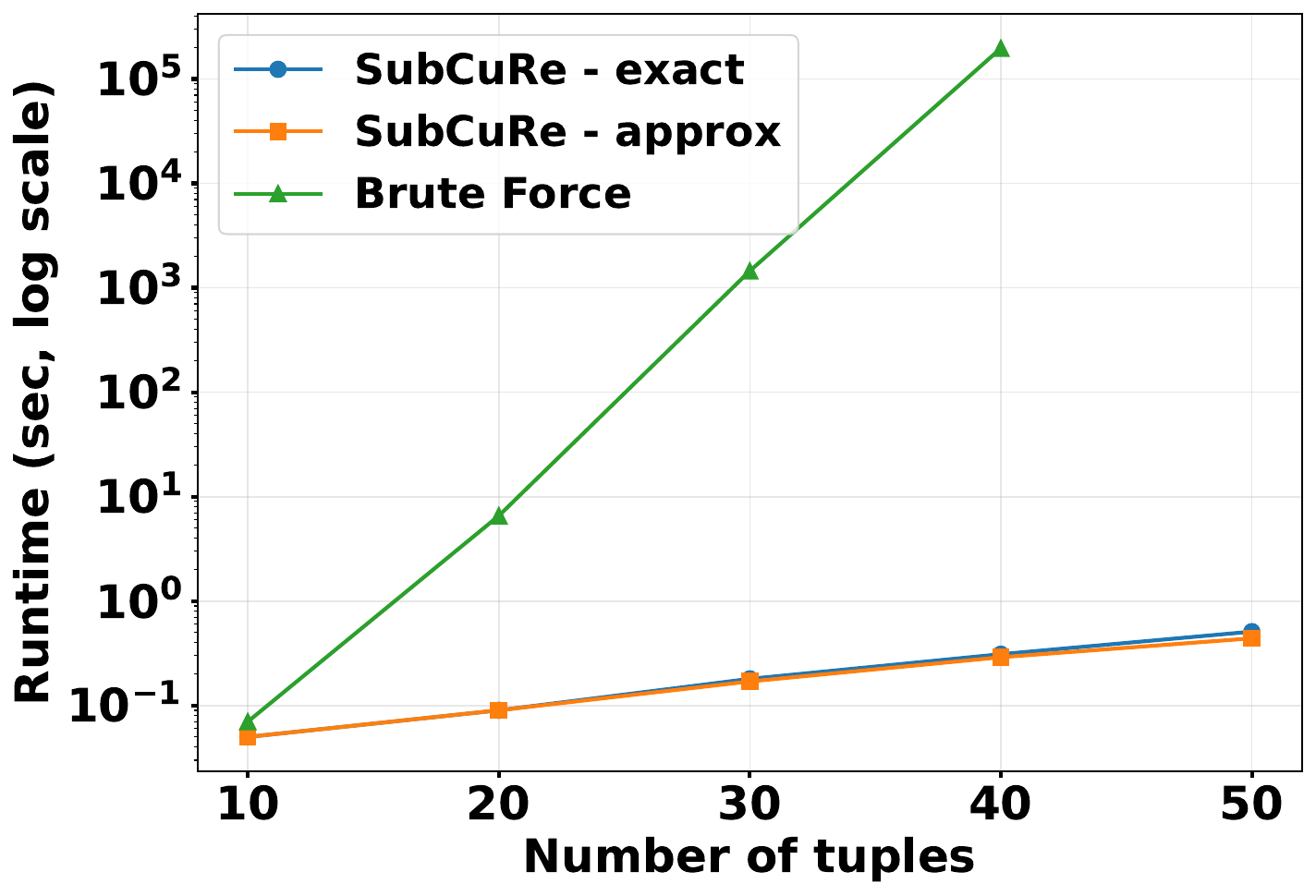}
        \caption*{\reva{(b) Runtime - tuple}}
        \label{fig:image2}
    \end{minipage}
    \begin{minipage}[t]{0.23\textwidth}
        \centering
        \includegraphics[height=2.4cm]{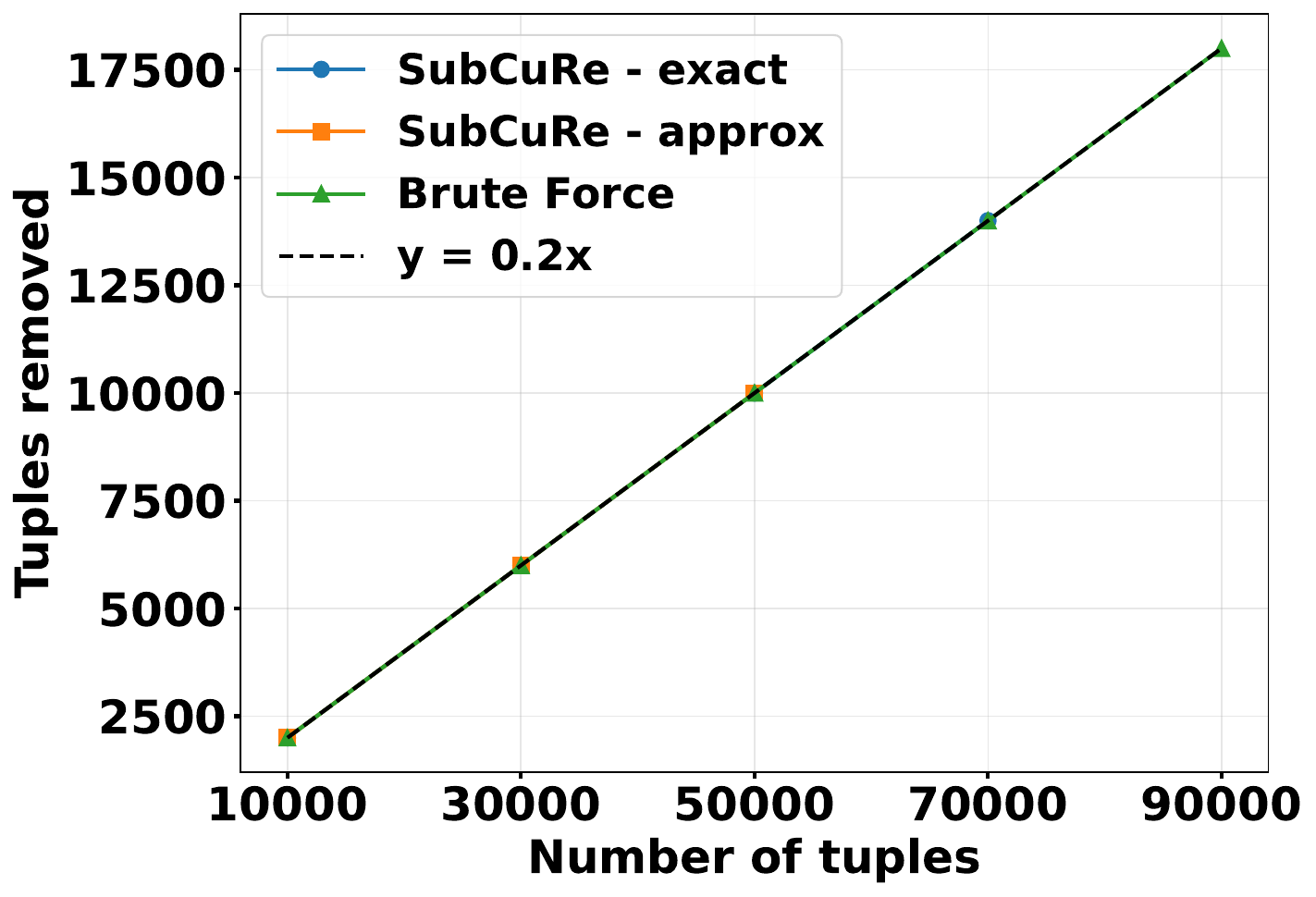}
        \caption*{\reva{(c) Quality - pattern}}
        \label{fig:image3}
    \end{minipage}
       \begin{minipage}[t]{0.23\textwidth}
        \centering
        \includegraphics[height=2.4cm]{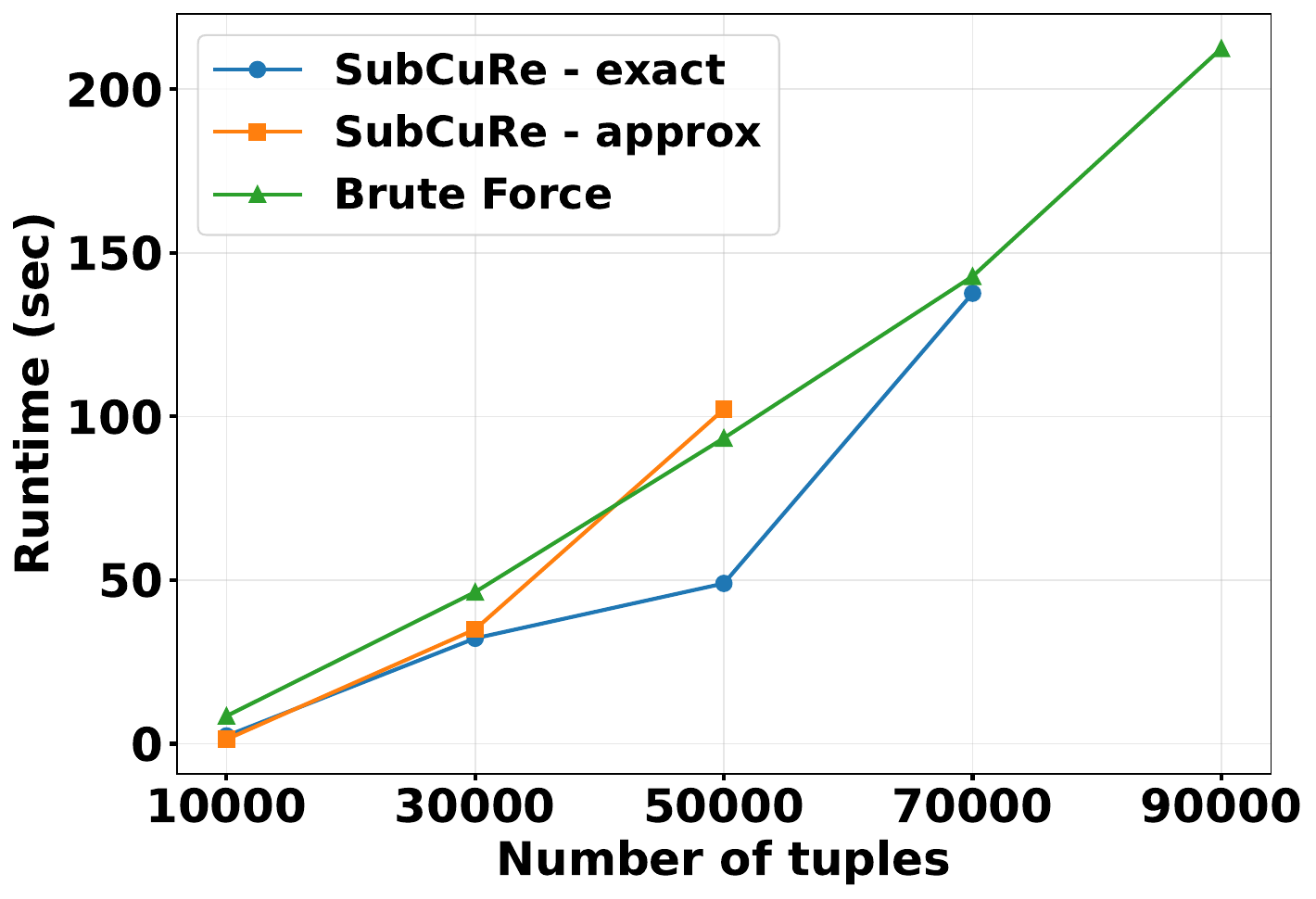}
        \caption*{\reva{(d) Runtime - pattern}}
        \label{fig:image3}
    \end{minipage}
    \caption{\reva{Comparison between OPT and \sysName\ on
synthetic data.}}
    \label{fig:opt}
\end{figure*}

We also compared \algoNameTuple\ and \topk\ on larger synthetic datasets. Both algorithms produced solutions smaller than the upper bound for up to 1.3k noisy tuples. However, as the number of noisy tuples increased - requiring more deletions to reach the target ATE - the performance of \topk\ degraded, since its single influence-score computation became outdated. Full details appear in the Appendix.

\subsubsection{\reva{Robustness to data quality issues}}
\label{subsec:data_quality}
\reva{To assess the robustness of \sysName\ to data quality issues, we conducted experiments that inject controlled noise into the datasets, following a procedure similar to prior work~\cite{horizon}. We consider three types of common data quality issues: outliers, duplicates, and missing values (with zeros used as placeholders), and varied the noise level from 5\% to 55\%. We evaluated the impact of these perturbations on two datasets, German Credit and Twins, measuring how the performance of \sysName\ degrades as the level of noise increases. This setup allows us to identify the point at which the system’s behavior begins to deteriorate, thereby quantifying its robustness to common data quality issues. The results are depicted in Figure~\ref{fig:data_quality}. Our results show that overall, both algorithms remain stable under varying levels and types of noise, demonstrating strong robustness. For duplicates and missing values, \algoNameTuple\ maintains relatively consistent performance as noise increases, indicating resilience to these perturbations. For outliers, it is somewhat less stable, since this type of noise directly affects the outcome values (by design) and, consequently, the ATE, requiring more tuples to be removed to reach the target ATE. \algoNamePattern\ is generally less stable across noise levels, and in some cases, failed to find a solution altogether. However, when a solution was found, its size was comparable to that obtained in the noise-free setting. This behavior is expected: because \algoNamePattern\ operates at the pattern level rather than the tuple level, it cannot selectively remove corrupted tuples that disproportionately influence the ATE.}


\begin{figure*}[htbp]
    \centering
    \begin{minipage}[t]{0.25\textwidth}
        \centering
        \includegraphics[height=2.4cm]{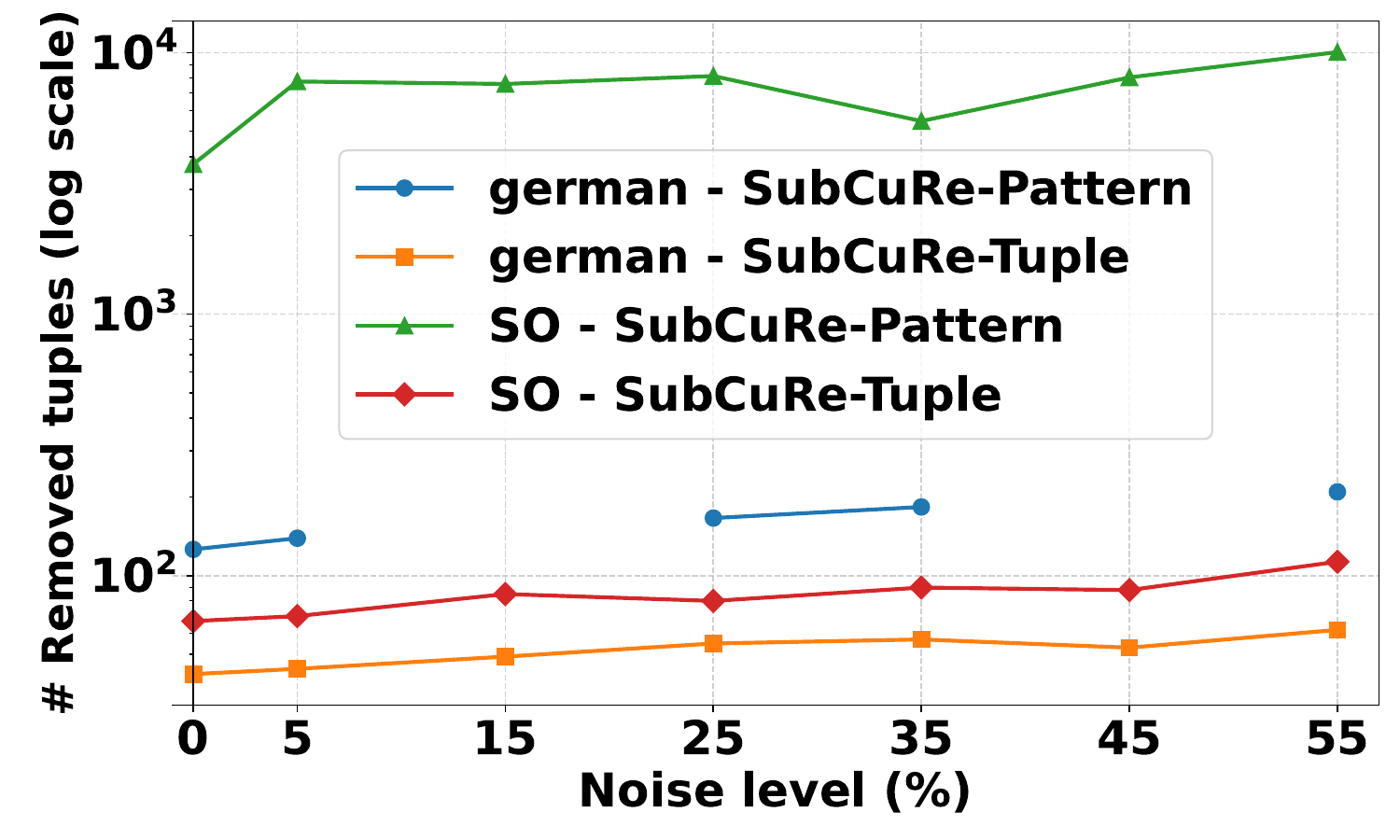}
        \caption*{\reva{(a) Duplicates}}
        \label{fig:image1}
    \end{minipage}
    \begin{minipage}[t]{0.25\textwidth}
        \centering
        \includegraphics[height=2.4cm]{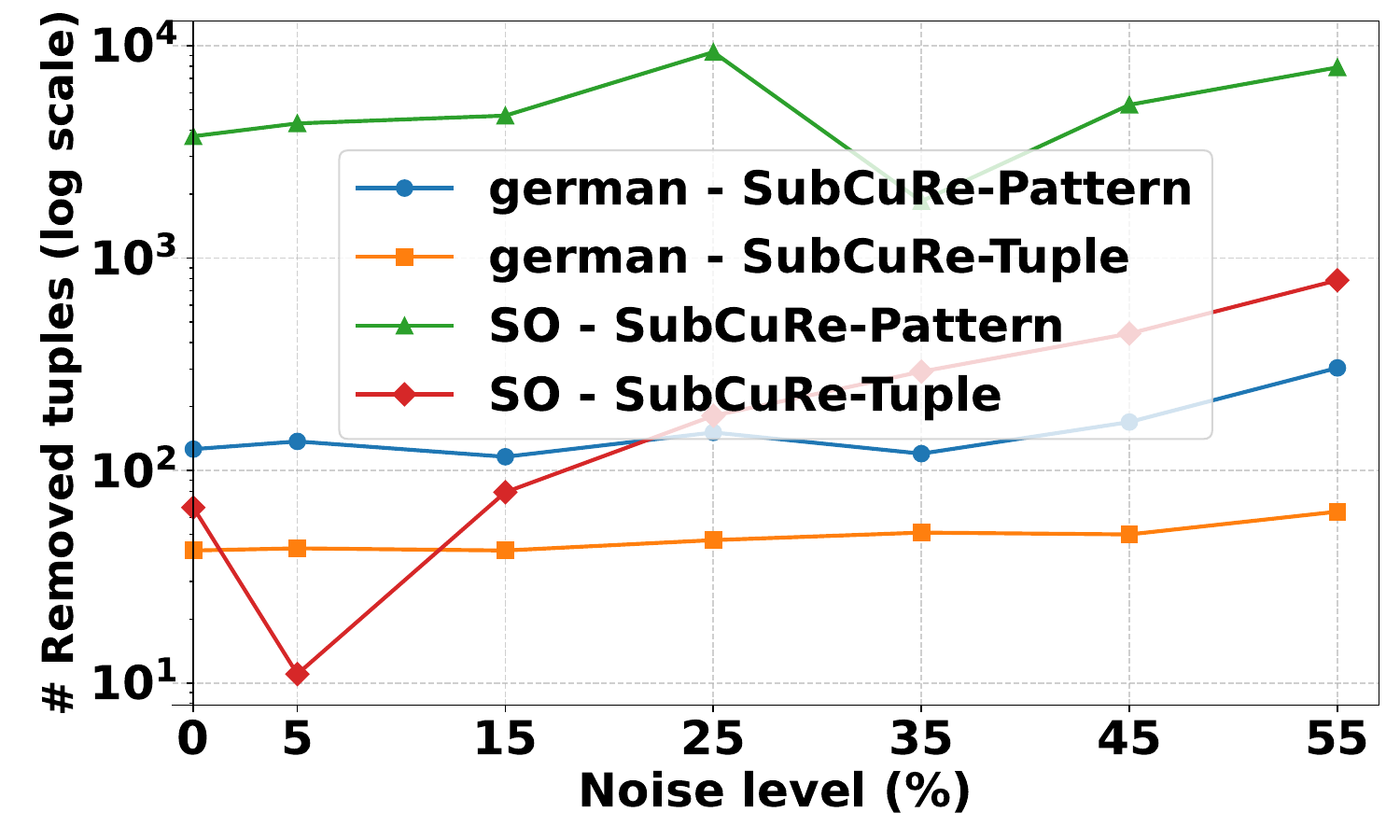}
        \caption*{\reva{(b) Missing values}}
        \label{fig:image2}
    \end{minipage}
    \begin{minipage}[t]{0.25\textwidth}
        \centering
        \includegraphics[height=2.4cm]{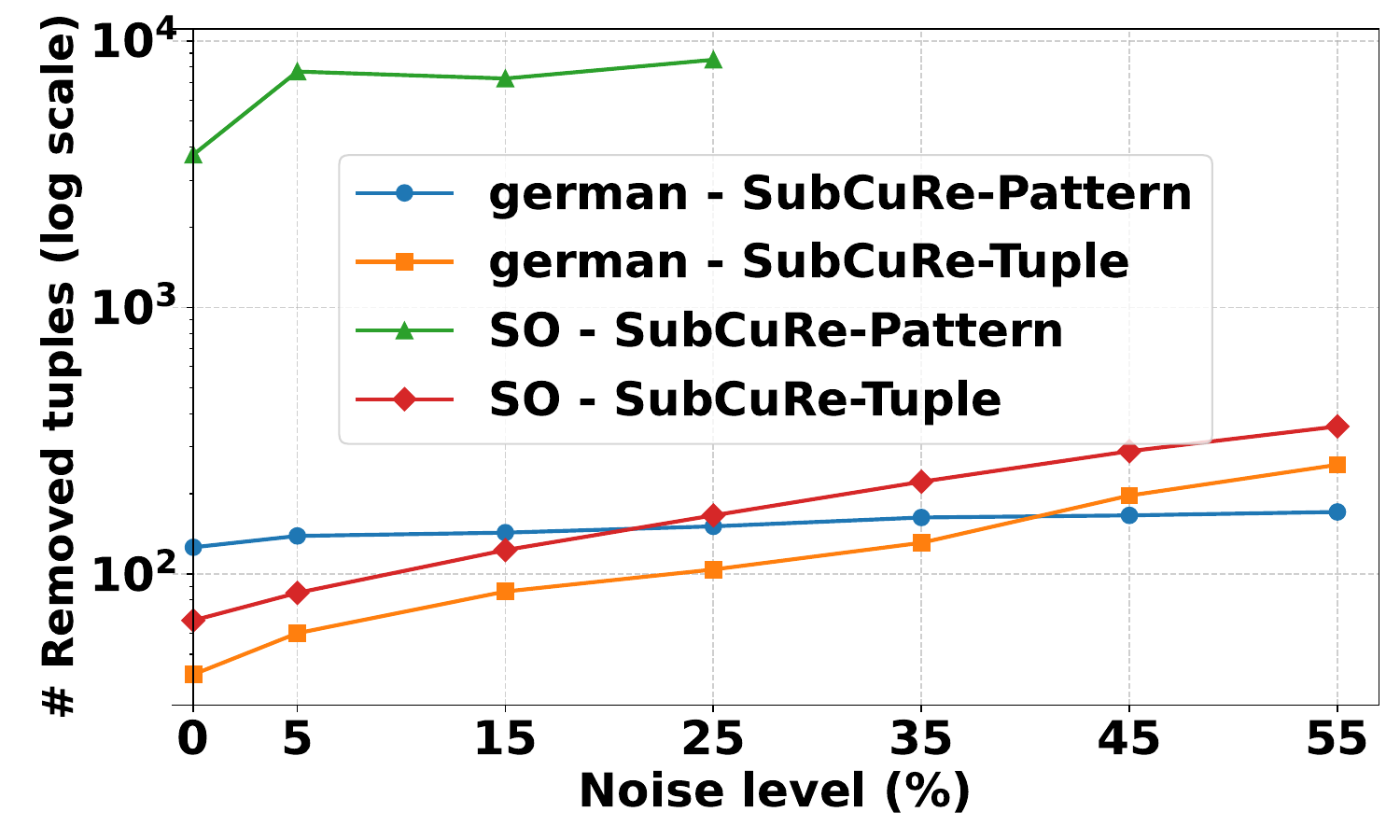}
        \caption*{\reva{(c) Outliers}}
        \label{fig:image3}
    \end{minipage}
    \caption{\reva{Number of removed tuples as a function of noise.}}
    \label{fig:data_quality}
\end{figure*}

\begin{figure*}[htbp]
    \centering
    \begin{minipage}[t]{0.25\textwidth}
        \centering
        \includegraphics[height=2.4cm]{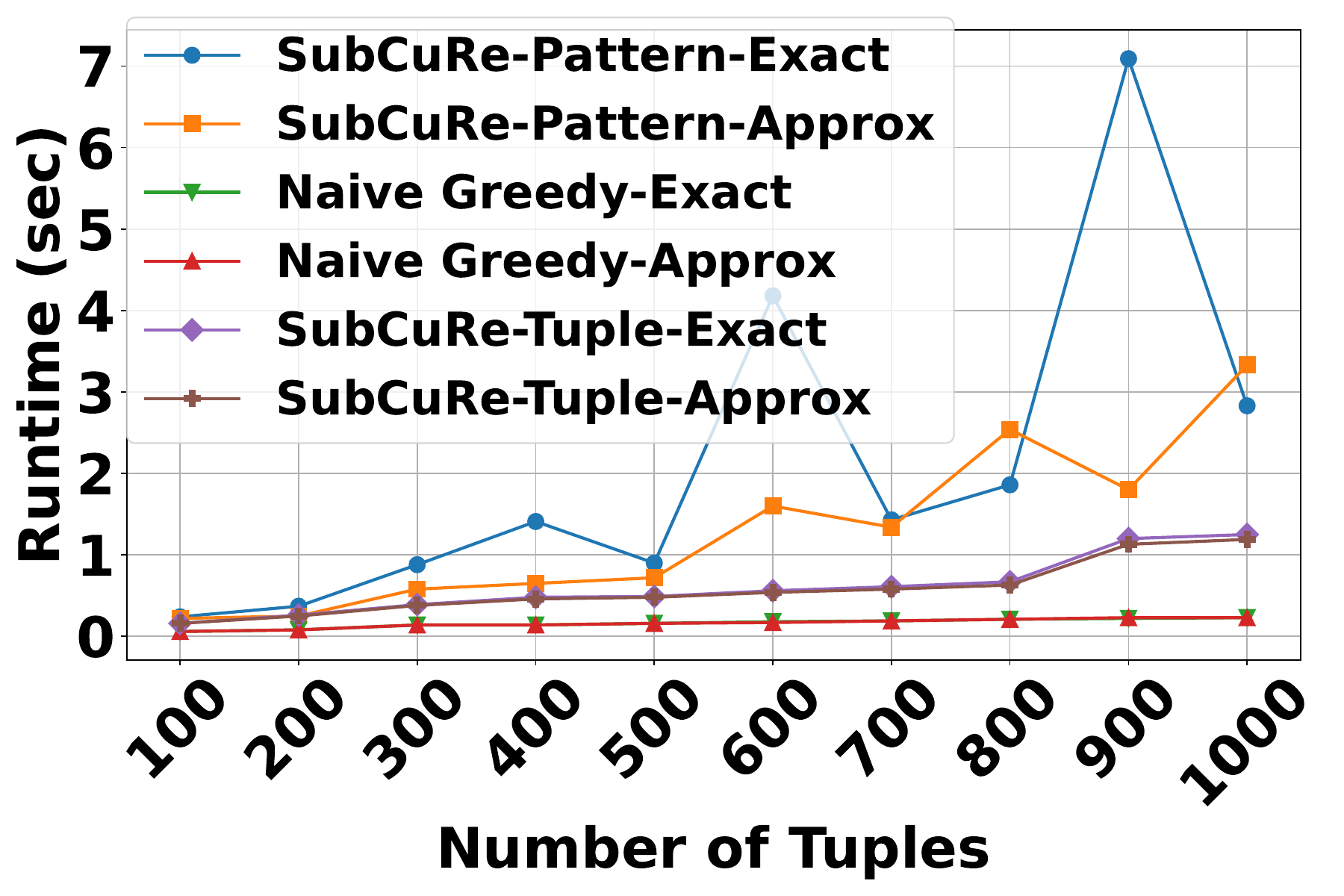}
        \caption*{(a) German Credit}
        \label{fig:image1}
    \end{minipage}
    \begin{minipage}[t]{0.25\textwidth}
        \centering
        \includegraphics[height=2.4cm]{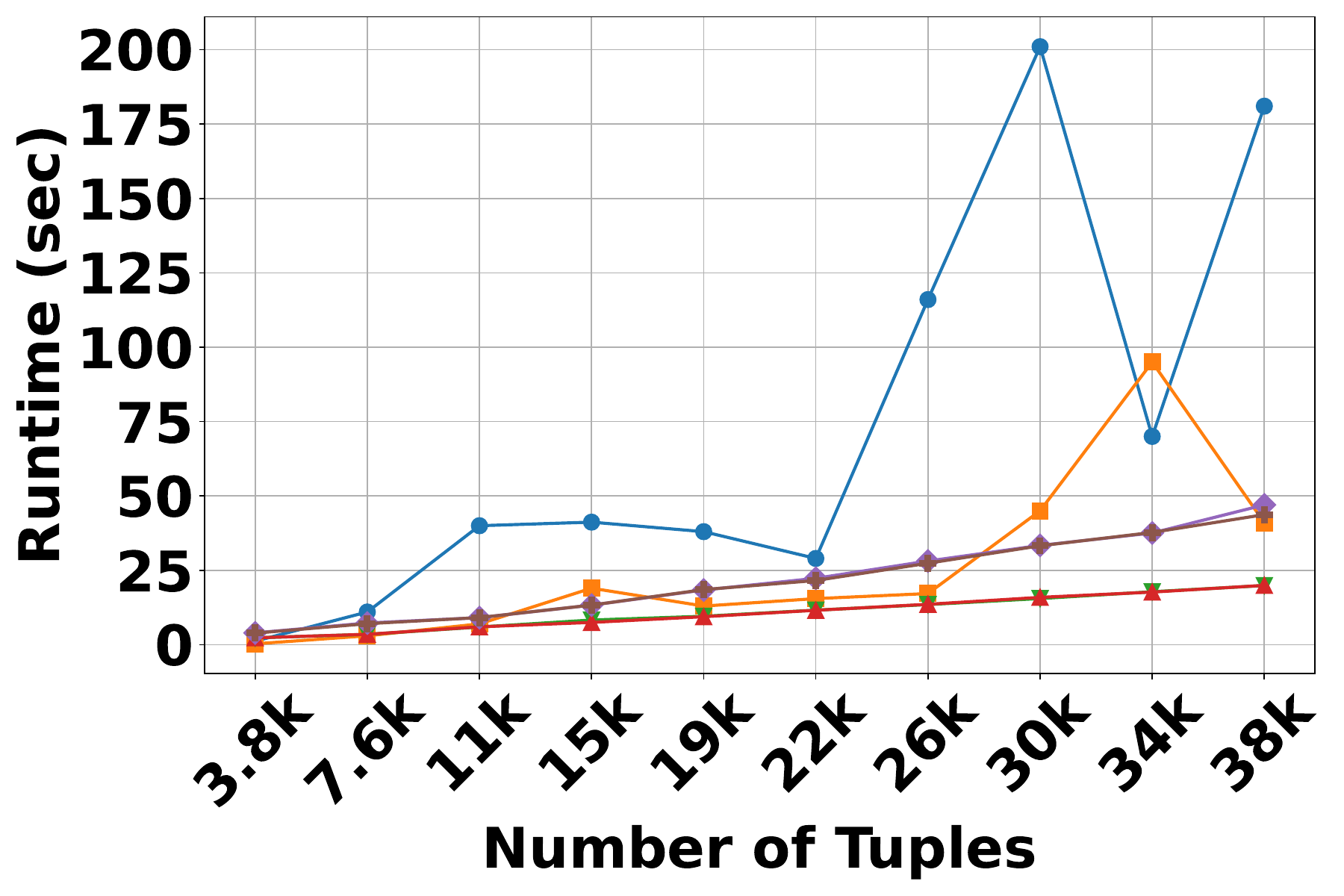}
        \caption*{(b) Stack Overflow}
        \label{fig:image2}
    \end{minipage}
    \begin{minipage}[t]{0.25\textwidth}
        \centering
        \includegraphics[height=2.4cm]{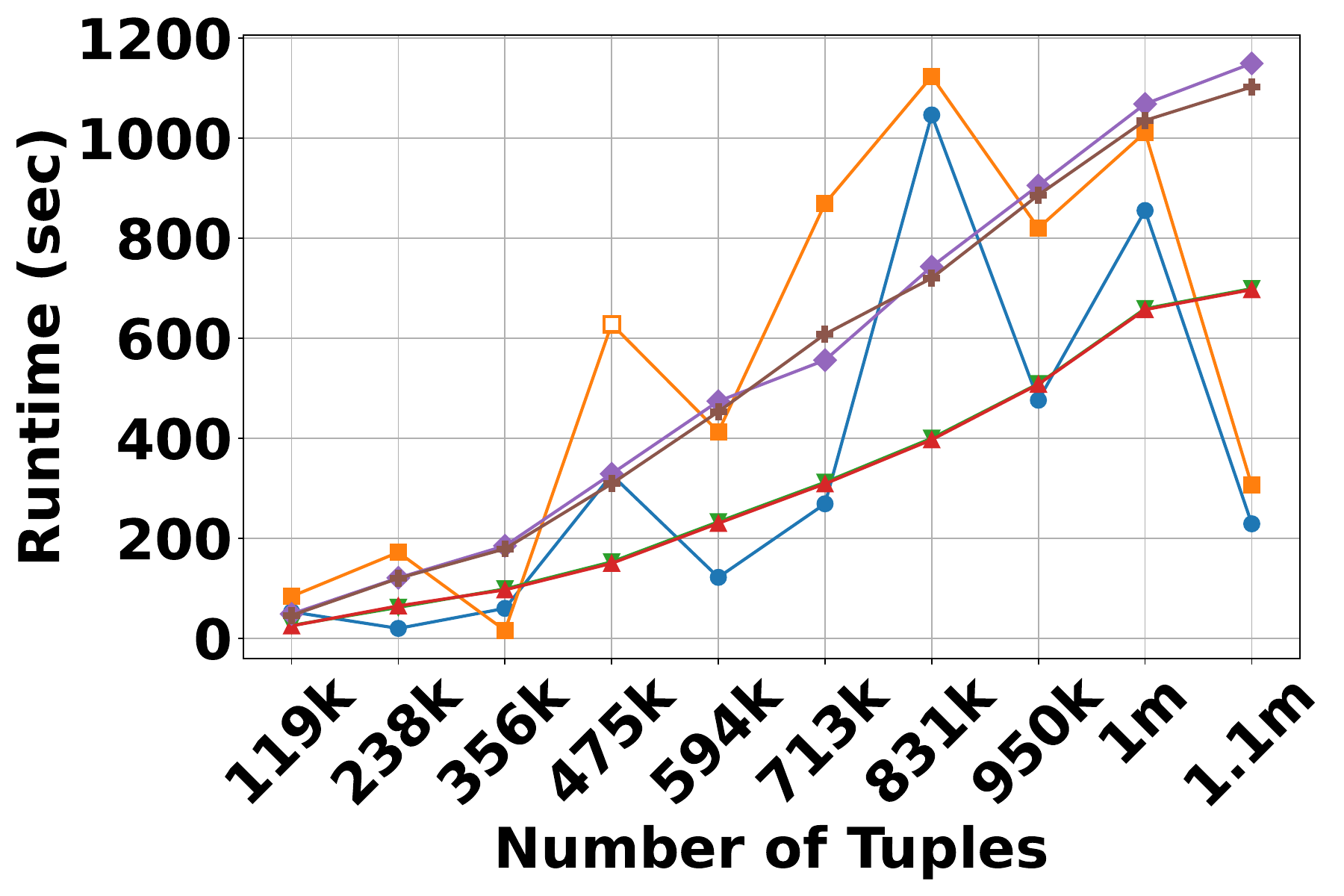}
        \caption*{(c) ACS}
        \label{fig:image3}
    \end{minipage}
    \caption{Runtime as a function of number of tuples in the data.}
    \label{fig:runtime_tuples}
\end{figure*}


\begin{figure*}[htbp]
    \centering
    \begin{minipage}[t]{0.25\textwidth}
        \centering
        \includegraphics[height=2.4cm]{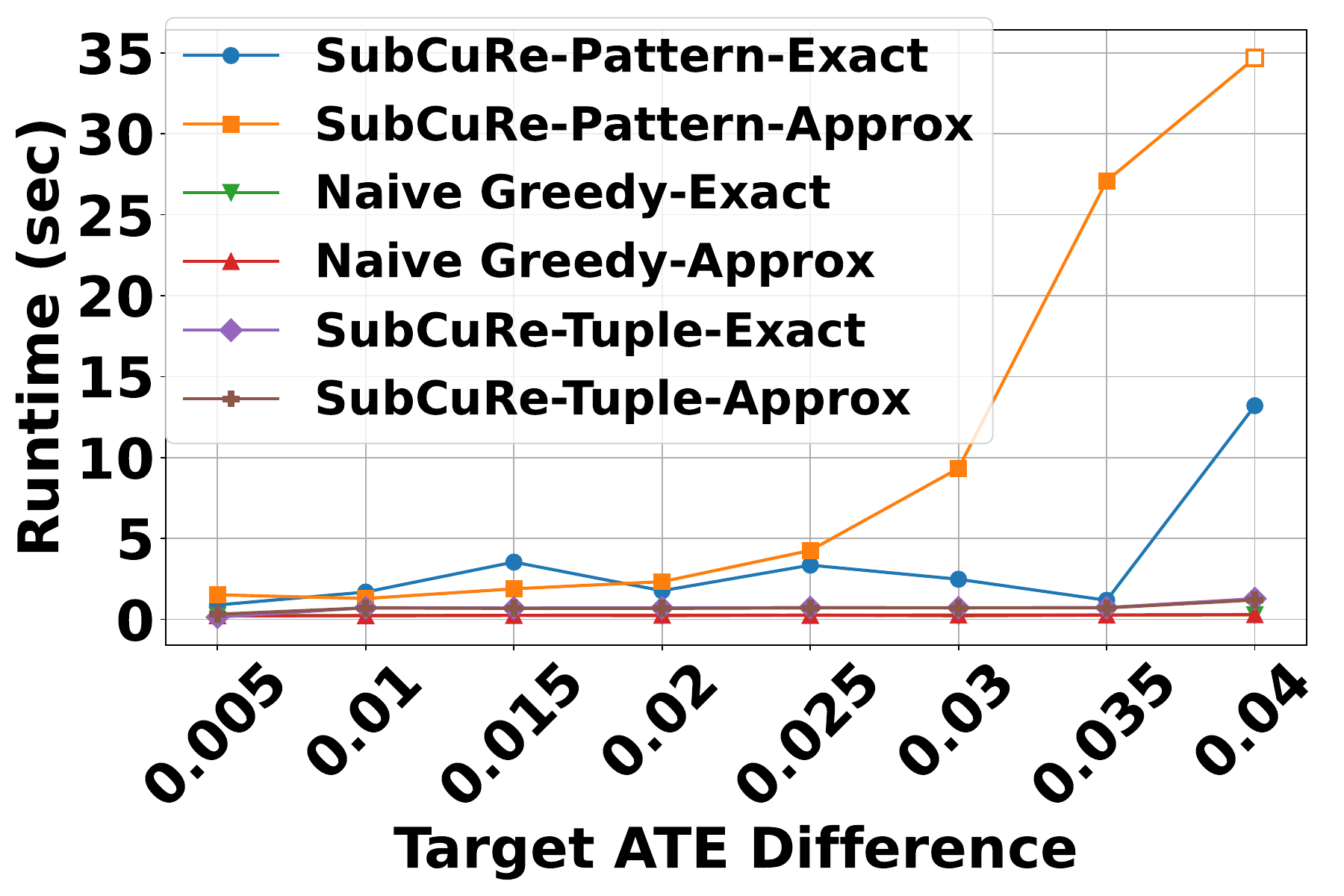}
        \caption*{(a) German Credit}
        \label{fig:image1}
    \end{minipage}
    \begin{minipage}[t]{0.25\textwidth}
        \centering
        \includegraphics[height=2.4cm]{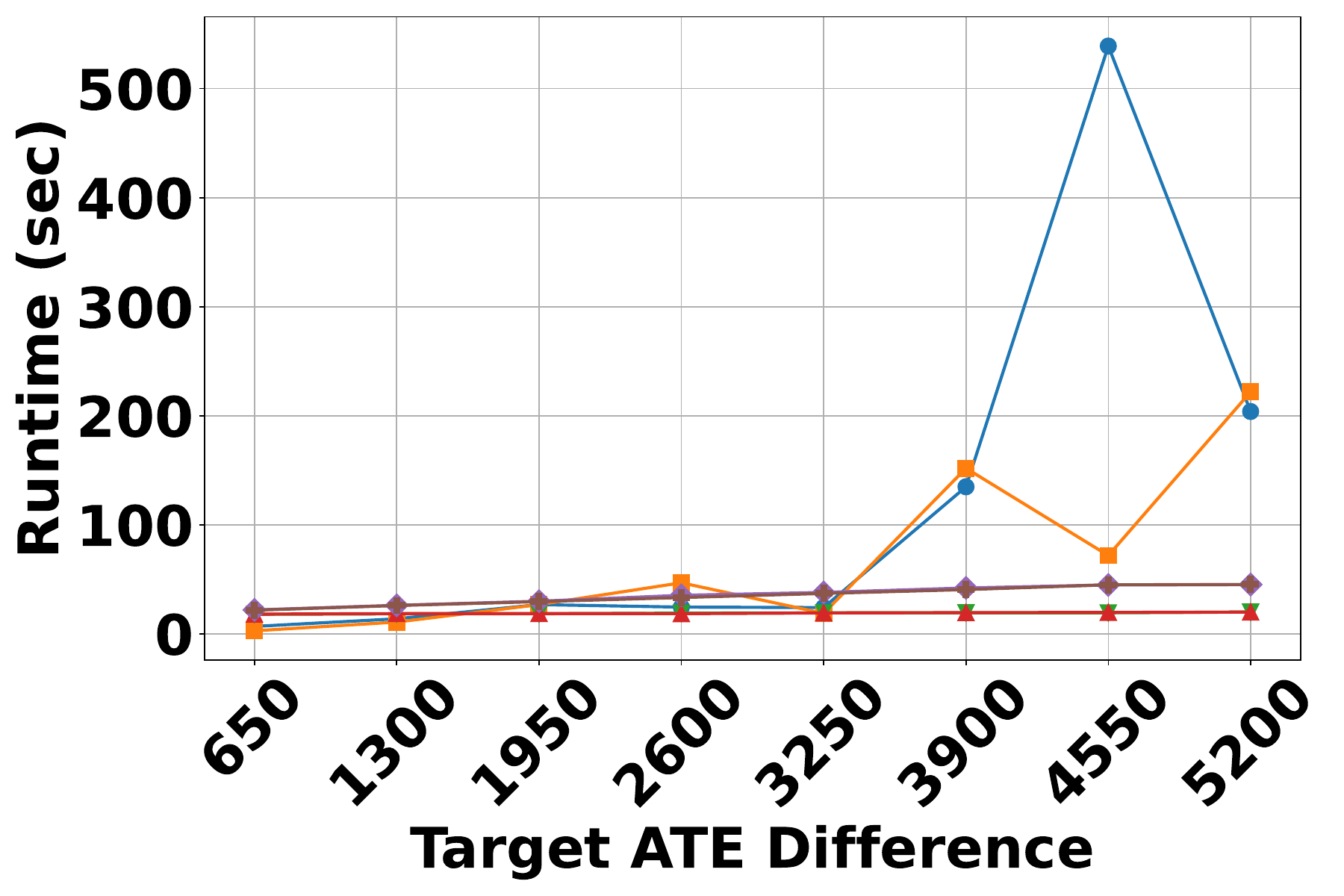}
        \caption*{(b) Stack Overflow}
        \label{fig:image2}
    \end{minipage}
    \begin{minipage}[t]{0.25\textwidth}
        \centering
        \includegraphics[height=2.4cm]{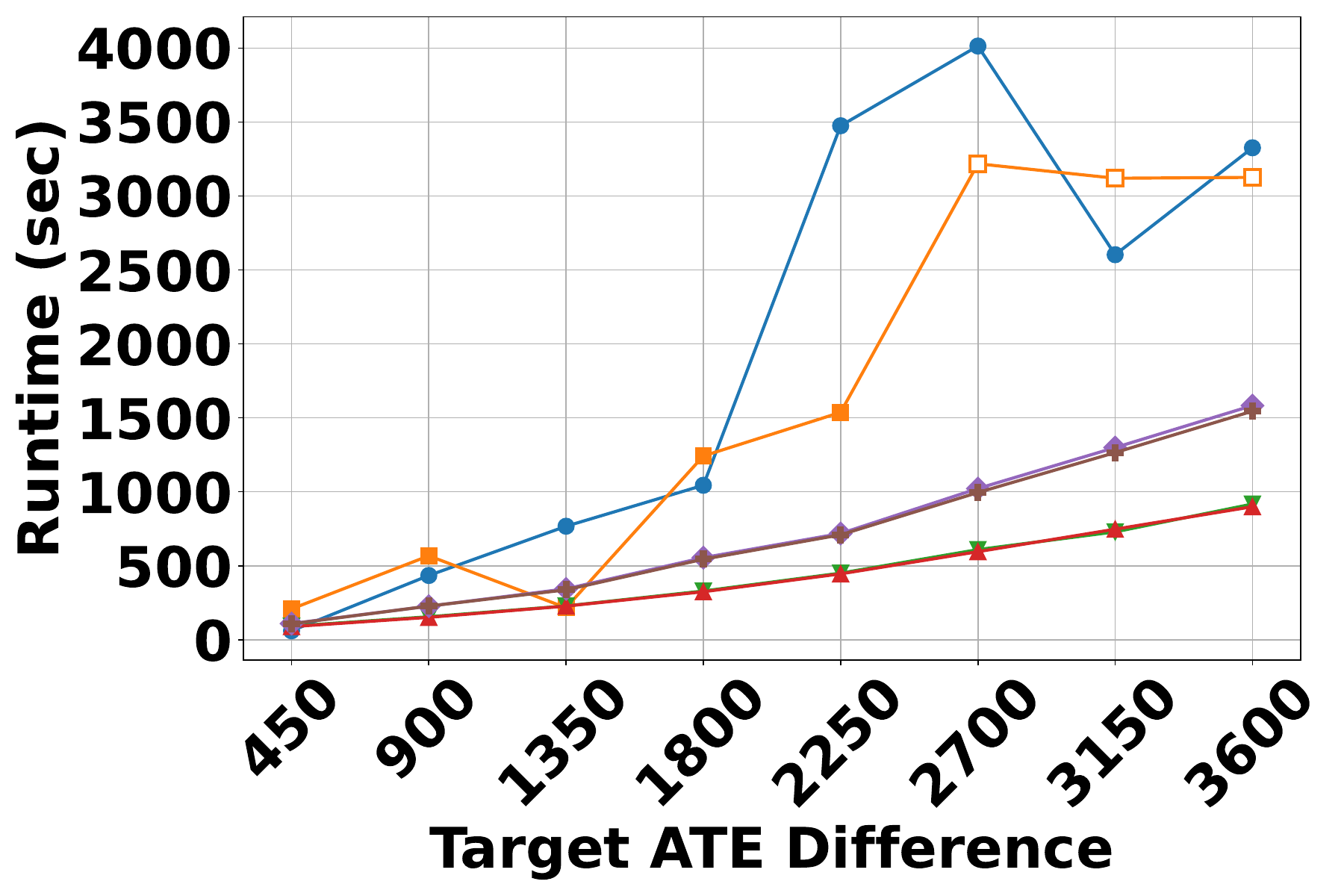}
        \caption*{(c) ACS}
        \label{fig:image3}
    \end{minipage}
    \caption{Runtime as a function of the distance between the initial and target ATE.}
    \label{fig:runtime_ate}
\end{figure*}

\subsection{Efficiency Evaluation}
\label{subsec:exp_efficiency}
Next, we evaluate how parameters affect algorithm runtime. Each experiment is repeated three times to account for randomization, and we report the average. We vary one parameter at a time (data size, confounders, or target ATE) while keeping others fixed. Results for the Twins dataset are omitted due to space, but similar trends were observed.

\vspace{1mm}
\noindent
\underline{Runtime vs. number of tuples}: We analyze how the number of tuples in the dataset impacts runtime. To do so, we randomly sample data subsets of increasing size. The results are shown in Figure~\ref{fig:runtime_tuples}. The runtime of the tuple-based algorithms (\algoNameTuple\ and \topk) increases approximately linearly with the number of tuples, as more influence scores need to be computed.
For \algoNamePattern, the runtime is less sensitive to data size. After processing around 40\% of the data, most patterns have already been considered. The spikes are due to the randomized nature of the algorithm: in some runs, it quickly finds a solution, while in others, it takes longer. Also, as discussed in Section \ref{subsec:ex_quality}, when the size of the subset to be removed is relatively big (as in ACS and Stack Overflow), the approximate version is less accurate and thus it takes longer to find a subgroup satisfying the condition.


\vspace{1mm}
\noindent
\underline{Runtime vs. target ATE}: We examine how the distance between the initial ATE and the target ATE affects runtime. To this end, we gradually increase the ATE shift. The results are presented in Figure~\ref{fig:runtime_ate}. Empty markers for \algoNamePattern\ indicate that no solution was found. As the gap between the initial and target ATE increases, the runtime of all algorithms grows accordingly. This is because a larger number of tuples must be removed to reach the target, leading to more iterations for \algoNameTuple\ and \topk, and forcing \algoNamePattern\ to explore larger subpopulations.

\vspace{1mm}
\noindent
\underline{Runtime vs. number of confounders}: We examine how the number of confounders affects runtime by randomly selecting confounder subsets. As more variables are included in the ATE estimation, runtimes increase across all algorithms. Full details appear in the Appendix.

\subsection{Ablation Study}
\label{subsec:exp_ablation}
We evaluate the effect of our ATE incremental update optimizations (Section~\ref{subsec:liner_unlearning}) by comparing against versions of our algorithms that recompute the ATE from scratch. We focus on the smaller German Credit and Twins datasets, as the unoptimized algorithms exceeded the time limit on larger datasets.
The unoptimized algorithms produced results nearly identical to their optimized counterparts. For example, on German Credit, the unoptimized \algoNameTuple\ (exact) selected the same 42 tuples but ran 18$\times$ slower (51s vs. 2.8s). Similarly, \algoNamePattern\ (exact) removed a slightly different subpopulation (132 vs. 126 tuples), reaching a nearly identical ATE (0.064 vs. 0.062), but with 2.3$\times$ longer runtime (118s vs. 50s).
As shown in Figure~\ref{fig:ablation}, runtime differences grow with data size. For \algoNameTuple, the optimizations yield up to a 100$\times$ speedup. For \algoNamePattern, the gain is around 10$\times$ on German Credit, where small subpopulations are removed (up to 200 tuples). On Twins, where larger groups are considered (of size up to thousands of tuples), the optimizations are less effective and can even slow performance. This shows that incremental update optimizations preserve solution quality while improving runtime, especially when only a small number of tuples are removed. When larger portions of the data are removed, the overhead can outweigh the benefits.

\begin{figure}[htbp]
    \centering
    \begin{minipage}{0.23\textwidth}
        \centering
        \includegraphics[width=\linewidth]{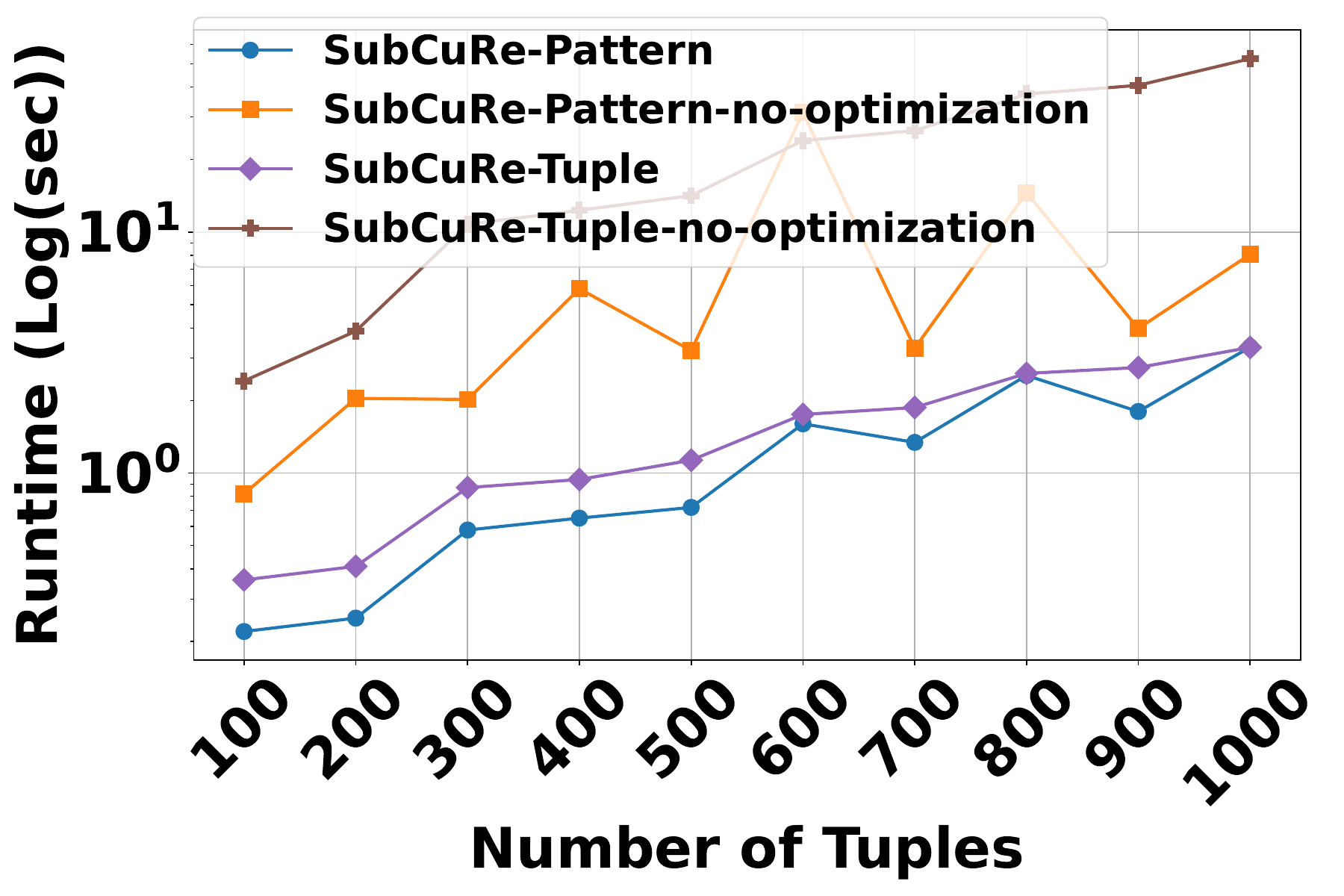}
        \caption*{(a) German Credit}
        \label{fig:image1}
    \end{minipage}
        \begin{minipage}{0.23\textwidth}
        \centering
\includegraphics[width=\linewidth]{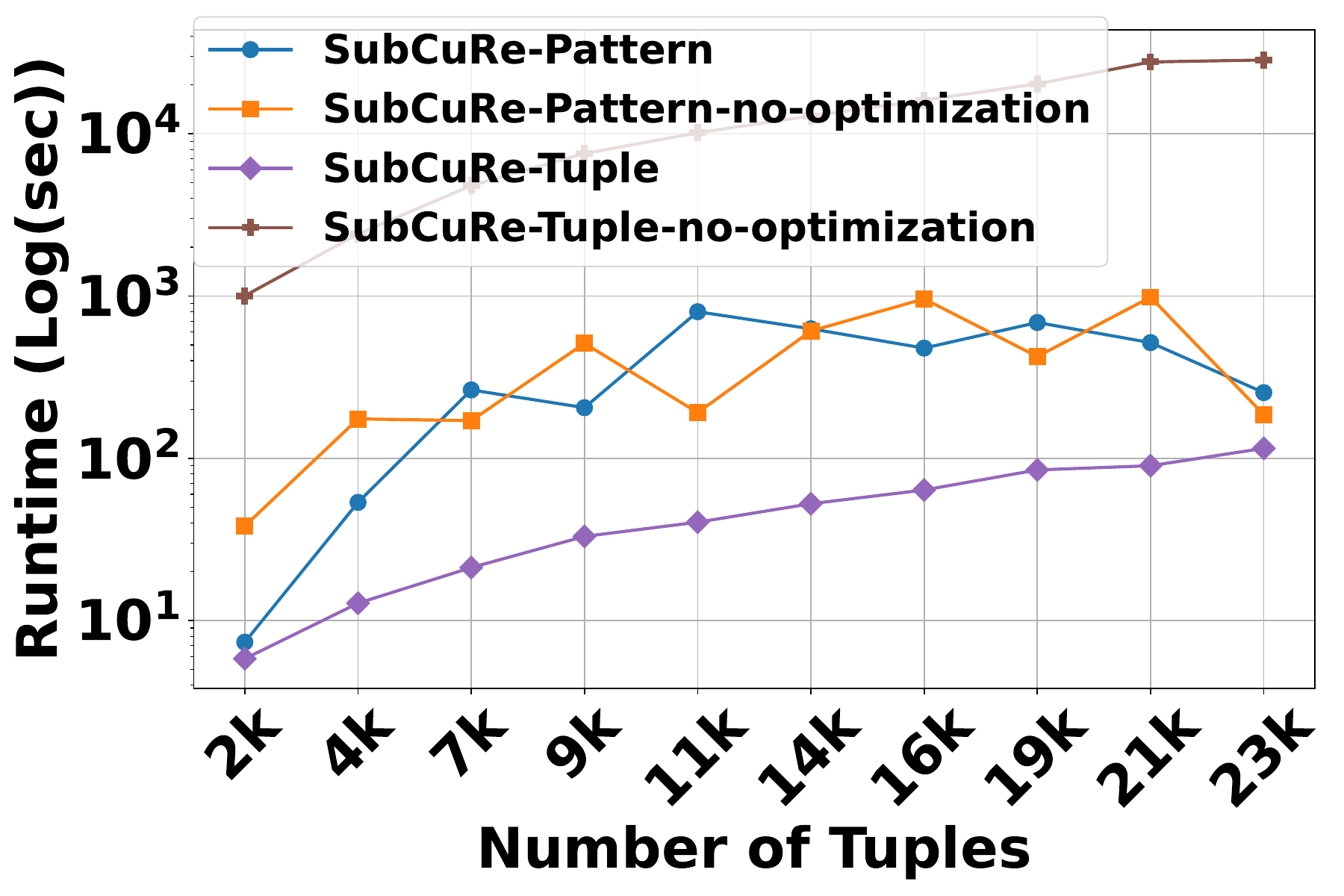}
        \caption*{(b) Twins}
        \label{fig:image1}
    \end{minipage}
\caption{Runtime as a function of number of tuples - with and without ATE incremental update optimizations}
\label{fig:ablation}
\end{figure}

\begin{figure}[htbp]
    \centering
    \begin{minipage}{0.48\textwidth}
        \centering
        \includegraphics[width=\linewidth]{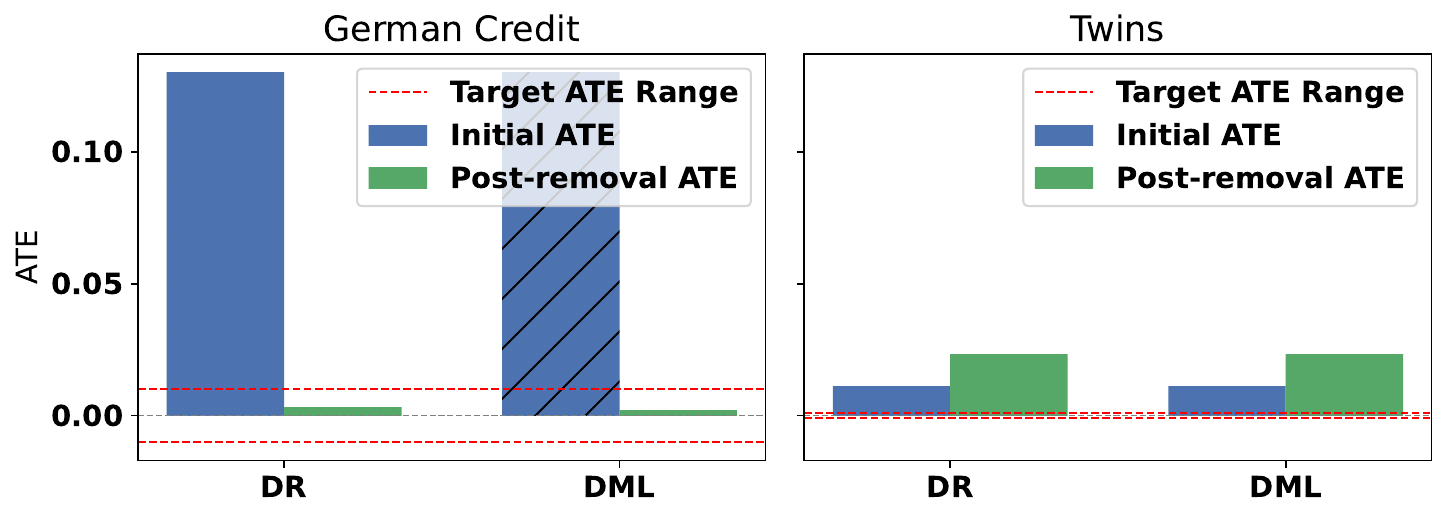}
     
    \end{minipage}
       
\caption{ATE estimates before and after removing tuples identified by \algoNameTuple\ (using linear regression update) from the German Credit and Twins datasets under two ATE estimators: DR and DML.}
\label{fig:generazability}
\end{figure}



\subsection{Generalizability of \sysName}
\label{subsec:ate_estimators}
Next, we show that \algoNameTuple\ consistently identifies influential tuples that affect the ATE, even when alternative estimation methods are used. Specifically, we evaluate its effectiveness with two widely used estimators: Doubly Robust (DR)~\cite{bang2005doubly} and Double Machine Learning (DML)~\cite{chernozhukov2016double}.
We conduct experiments on the German Credit and Twins datasets. For each scenario, we first compute the initial ATE using the DR and DML estimators\footnote{using the implementation of EconML~\cite{econml2019}.}. Then, we apply \algoNameTuple\ (with linear regression or IPW estimators) to identify influential tuples. After removing the selected tuples, we recompute the ATE using the same estimator (DR or DML).
Our results demonstrate that \algoNameTuple\ reliably finds tuples whose removal shifts the ATE in the desired direction, across both DR and DML. The results while using linear regression with the exact update optimization are shown in Figure~\ref{fig:generazability} 
In German Credit, the initial ATE is 
0.13 using all estimators. After removing tuples selected by \algoNameTuple - using either linear regression or IPW - the ATE, when recomputed with DR and DML, moves into the target range of $(0 \pm 0.01)$.
In Twins, the initial ATE estimated by DR and DML is positive (0.011), while the ATE estimated by linear regression and IPW is negative (-0.016). Consequently, \algoNameTuple\ selects tuples whose removal causes an increase in ATE. Although the post-removal ATE computed with DR and DML does not reach the target range (0 $\pm$0.001), the shift in ATE is in the correct direction, consistent with the effect predicted by \algoNameTuple.
\emph{These results highlight the generalizability of \algoNameTuple: its selected tuples influence ATE estimates consistently, even when evaluated using different estimation strategies}.

\subsection{\reva{Relation to Sensitivity Analysis}}
\label{exp:sensitivity_analysis}

\reva{We demonstrate that \sysName\ provides a complementary perspective on ATE sensitivity compared to existing methods. We focus on two closely related approaches (discussed in detail in Section \ref{sec:related}): (i) KonFound \cite{frank2013would}, which quantifies the degree of hidden bias required to overturn an inference, and (ii)
WTE \cite{jeong2020robust}, which measures robustness under confounder shift by estimating the worst-case ATE across subpopulations of a given size. As a case study, we evaluate both methods on relevant datasets, report their results, and highlight their differences from \algoNameTuple{}. We exclude \algoNamePattern{}, as both compared methods operate at the individual tuple level. }

\noindent
\underline{Konfound}:  \reva{KonFound evaluates sensitivity to unobserved confounders by identifying the “switch point” at which an effect becomes null, akin to pushing the ATE toward 0 (effectively nullifying the causal effect). Thus, we only evaluate this method on the German Credit and Twins datasets, where the goal is to push the ATE towards 0. In the German dataset, Konfound estimated that over 50\% of the data would need to be replaced to nullify the effect, while in Twins the estimate was 48\%. In contrast, \algoNameTuple\ identified just 4.2\% of the data in German and 1.1\% in Twins for removal to achieve no effect. This indicates that \emph{while results are generally robust to unobserved confounders, they can be highly sensitive to specific observed data points}, highlighting the complementary perspective \algoNameTuple\ provides for sensitivity analysis. }

\smallskip
\noindent
\underline{WTE}
\reva{For WTE, we set the subpopulation size threshold to match the number of tuples removed by \algoNameTuple\ in each dataset. Across all datasets, WTE shows minimal change in ATE, indicating strong robustness to confounder distribution shifts. In the German Credit dataset, \algoNameTuple{} reduces the ATE by ~95\%, whereas WTE changes it by only ~3\%. In Twins, \algoNameTuple{} reduces the ATE magnitude by ~94\%, while WTE shifts it by less than 1\%. On Stack Overflow, \algoNameTuple{} decreases the ATE by ~38\%, whereas WTE changes it by only ~7\%. Similarly, on the ACS dataset, \algoNameTuple{} increases the ATE by ~31\%, while WTE changes it by only ~2\%.
These results show that \emph{while the ATE may be largely stable under shifts in confounder distribution, it can be highly sensitive to the removal of specific tuples}, demonstrating a complementary notion of robustness captured by \algoNameTuple{}. }



 \section{Conclusion \& Future Work}
\label{sec:conc}
We presented \sysName, a framework for auditing the robustness of causal conclusions via cardinality repair. By identifying minimal subpopulations whose removal brings the ATE within a target range, \sysName\ offers a quantitative sensitivity measure and interpretable insights into influential data regions.


\sysName\ assumes that a sufficient set of confounders is provided by the user. In practice, however, errors in confounder selection are a major source of bias in causal analysis~\cite{youngmann2023explaining,mcnamee2003confounding,diaz2013sensitivity}. Additionally, the current pattern repair algorithm is limited to conjunctions of equality predicates, which may be too coarse for datasets with continuous or ordinal attributes, potentially missing more expressive or meaningful subpopulations. Finally, our implementation currently supports a single-relation schema, and does not yet extend to multi-table relational databases. 
\revc{ Previous work, such as~\cite {SalimiPKGRS20}, has extended causal models to handle multi-table data.
Extending our approach to multi-table relational data is not straightforward, as it requires careful handling of how treatments, outcomes, and confounders are distributed across tables, how interventions propagate through joins, and how relational dependencies affect influence measures. Developing such an extension is non-trivial and represents an important direction for future work. Notably, prior work leveraging causal inference in data management ~\cite{abs-2207-12718,salimi2018bias,youngmann2024summarized} has also focused on single-table settings.}
We leave all these extensions to future work.

Several other promising directions remain for future research. A natural extension is to handle multiple causal effect estimations simultaneously. Additionally, expanding the framework to support a wider variety of intervention models beyond tuple deletion, such as tuple updates and insertions, would increase its applicability to real-world data scenarios where data can be modified in diverse ways. 
\revb{We also plan to explore directions that could deepen the theoretical understanding of our problem. One interesting direction is to adopt a probabilistic perspective, treating the database as a sample from an underlying distribution (similar to~\cite{jeong2020robust}). This could enable reasoning about the effect of removing tuples without explicitly recomputing the ATE, for instance by using probabilistic bounds to estimate the impact of tuple removals. Such an approach might yield theoretical guarantees on the expected change in ATE.  
Another promising avenue is to explore approximate methods that provide (weaker) guarantees, such as using over- and under-approximations of the ATE obtained by sorting tuples based on outcome values and applying greedy selection strategies. 
Finally, studying the parameterized complexity of our problem, for example, with respect to output size, is another promising direction that could reveal fundamental limitations of our approach. We leave these investigations for future work.
}




\bibliographystyle{ACM-Reference-Format}
\bibliography{vldb_sample.bib}

@article{rosenbaum2002observational,
  title={Observational study},
  author={Rosenbaum, Paul R},
  journal={Encyclopedia of statistics in behavioral science},
  volume={3},
  pages={1451--1462},
  year={2005},
  publisher={John Wiley and Sons New York}
}

@article{diaz2013sensitivity,
  title={Sensitivity analysis for causal inference under unmeasured confounding and measurement error problems},
  author={D{\'\i}az, Iv{\'a}n and van der Laan, Mark J},
  journal={The international journal of biostatistics},
  volume={9},
  number={2},
  pages={149--160},
  year={2013},
  publisher={De Gruyter}
}

@incollection{robins2000sensitivity,
  title={Sensitivity analysis for selection bias and unmeasured confounding in missing data and causal inference models},
  author={Robins, James M and Rotnitzky, Andrea and Scharfstein, Daniel O},
  booktitle={Statistical models in epidemiology, the environment, and clinical trials},
  pages={1--94},
  year={2000},
  publisher={Springer}
}

@article{vanderweele2017evalue,
  title={Sensitivity analysis for unmeasured confounding in meta-analyses},
  author={Mathur, Maya B and VanderWeele, Tyler J},
  journal={Journal of the American Statistical Association},
  year={2020},
  publisher={Taylor \& Francis}
}

@article{cinelli2020making,
  title={Making sense of sensitivity: Extending omitted variable bias},
  author={Cinelli, Carlos and Hazlett, Chad},
  journal={Journal of the Royal Statistical Society Series B: Statistical Methodology},
  volume={82},
  number={1},
  pages={39--67},
  year={2020},
  publisher={Oxford University Press}
}

@article{fjeldstad2021simex,
  title={Measurement error adjustment using the SIMEX method: An application to student growth percentiles},
  author={Shang, Yi},
  journal={Journal of Educational Measurement},
  volume={49},
  number={4},
  pages={446--465},
  year={2012},
  publisher={Wiley Online Library}
}

@article{blackwell2014selection,
  title={A selection bias approach to sensitivity analysis for causal effects},
  author={Blackwell, Matthew},
  journal={Political Analysis},
  volume={22},
  number={2},
  pages={169--182},
  year={2014},
  publisher={Cambridge University Press}
}

@article{rosenbaum1983central,
  title={The central role of the propensity score in observational studies for causal effects},
  author={Rosenbaum, Paul R and Rubin, Donald B},
  journal={Biometrika},
  volume={70},
  number={1},
  pages={41--55},
  year={1983},
  publisher={Oxford University Press}
}

@article{ding2018causal,
  title={Causal inference},
  author={Ding, Peng and Li, Fan},
  journal={Statistical Science},
  volume={33},
  number={2},
  pages={214--237},
  year={2018},
  publisher={JSTOR}
}

@article{Nab2020,
  title={Quantitative bias analysis for a misclassified confounder: a comparison between marginal structural models and conditional models for point treatments},
  author={Nab, Linda and Groenwold, Rolf HH and van Smeden, Maarten and Keogh, Ruth H},
  journal={Epidemiology},
  volume={31},
  number={6},
  pages={796--805},
  year={2020},
  publisher={LWW}
}

@article{MilesValeriCoull2024,
  title={Measurement error-robust causal inference via constructed instrumental variables},
  author={Miles, Caleb H and Valeri, Linda and Coull, Brent},
  journal={arXiv preprint arXiv:2406.00940},
  year={2024}
}

@inproceedings{SarracinoMikucka2017,
  title={Bias and efficiency loss in regression estimates due to duplicated observations: a Monte Carlo simulation},
  author={Sarracino, Francesco and Mikucka, Malgorzata},
  booktitle={Survey Research Methods},
  volume={11},
  number={1},
  pages={17--44},
  year={2017}
}

@article{LockElAnsari2025,
  title={New world of big data—new challenges for evidence synthesis: impact of data duplication on estimates generated by meta-analyses and the development of a framework for its identification and management},
  author={Lock, Merilyn and El Ansari, Walid},
  journal={Journal of Clinical Epidemiology},
  volume={179},
  pages={111641},
  year={2025},
  publisher={Elsevier}
}

@article{AgarwalSingh2024,
  title={Causal inference with corrupted data: Measurement error, missing values, discretization, and differential privacy},
  author={Agarwal, Anish and Singh, Rahul},
  journal={arXiv preprint arXiv:2107.02780},
  year={2021},
  publisher={arXiv}
}

@article{KadlecSainaniNimphius2023,
  title={With great power comes great responsibility: common errors in meta-analyses and meta-regressions in strength \& conditioning research},
  author={Kadlec, Daniel and Sainani, Kristin L and Nimphius, Sophia},
  journal={Sports Medicine},
  volume={53},
  number={2},
  pages={313--325},
  year={2023},
  publisher={Springer}
}

@article{Bogaert2025,
  title={The Effect of Measurement Error on Hypothesis Testing in Small Sample Structural Equation Modeling: A Comparison of Various Estimation Approaches},
  author={Bogaert, Jasper and Loh, Wen Wei and Schuberth, Florian and Rosseel, Yves},
  journal={Structural Equation Modeling: A Multidisciplinary Journal},
  volume={32},
  number={2},
  pages={215--236},
  year={2025},
  publisher={Taylor \& Francis}
}

@article{horizon,
author = {Rezig, El Kindi and Ouzzani, Mourad and Aref, Walid G. and Elmagarmid, Ahmed K. and Mahmood, Ahmed R. and Stonebraker, Michael},
title = {Horizon: scalable dependency-driven data cleaning},
year = {2021},
issue_date = {July 2021},
publisher = {VLDB Endowment},
volume = {14},
number = {11},
issn = {2150-8097},
url = {https://doi.org/10.14778/3476249.3476301},
doi = {10.14778/3476249.3476301},
journal = {Proc. VLDB Endow.},
month = {jul},
pages = {2546–2554},
numpages = {9}
}

@book{kleinberg2006algorithm,
  title={Algorithm design},
  author={Kleinberg, Jon and Tardos, Eva},
  year={2006},
  publisher={Pearson Education India}
}

@inproceedings{DBLP:conf/icdt/CarmeliGKLT21,
  author       = {Nofar Carmeli and
                  Martin Grohe and
                  Benny Kimelfeld and
                  Ester Livshits and
                  Muhammad Tibi},
  editor       = {Ke Yi and
                  Zhewei Wei},
  title        = {Database Repairing with Soft Functional Dependencies},
  booktitle    = {24th International Conference on Database Theory, {ICDT} 2021, March
                  23-26, 2021, Nicosia, Cyprus},
  series       = {LIPIcs},
  volume       = {186},
  pages        = {16:1--16:17},
  publisher    = {Schloss Dagstuhl - Leibniz-Zentrum f{\"{u}}r Informatik},
  year         = {2021}
}

@INPROCEEDINGS{5767833,
  author={Chiang, Fei and Miller, Renée J.},
  booktitle={2011 IEEE 27th International Conference on Data Engineering}, 
  title={A unified model for data and constraint repair}, 
  year={2011},
  volume={},
  number={},
  pages={446-457},
  doi={10.1109/ICDE.2011.5767833}}

@misc{pirhadi2024otclean,
      title={OTClean: Data Cleaning for Conditional Independence Violations using Optimal Transport}, 
      author={Alireza Pirhadi and Mohammad Hossein Moslemi and Alexander Cloninger and Mostafa Milani and Babak Salimi},
      year={2024},
      eprint={2403.02372},
      archivePrefix={arXiv},
      primaryClass={cs.LG}
}

@inproceedings{salimi2018bias,
  title={Bias in olap queries: Detection, explanation, and removal},
  author={Salimi, Babak and Gehrke, Johannes and Suciu, Dan},
  booktitle={Proceedings of the 2018 International Conference on Management of Data},
  pages={1021--1035},
  year={2018}
}

@misc{asuncion2007uci,
  title={UCI machine learning repository},
  author={Asuncion, Arthur and Newman, David},
  year={2007},
  publisher={Irvine, CA, USA}
}

@inproceedings{agrawal1994fast,
  title={Fast algorithms for mining association rules},
  author={Agrawal, Rakesh and Srikant, Ramakrishnan and others},
  booktitle={Proc. 20th int. conf. very large data bases, VLDB},
  volume={1215},
  pages={487--499},
  year={1994},
  organization={Santiago, Chile}
}

@inproceedings{SalimiPKGRS20,
  author       = {Babak Salimi and
                  Harsh Parikh and
                  Moe Kayali and
                  Lise Getoor and
                  Sudeepa Roy and
                  Dan Suciu},
  editor       = {David Maier and
                  Rachel Pottinger and
                  AnHai Doan and
                  Wang{-}Chiew Tan and
                  Abdussalam Alawini and
                  Hung Q. Ngo},
  title        = {Causal Relational Learning},
  booktitle    = {Proceedings of the 2020 International Conference on Management of
                  Data, {SIGMOD} Conference 2020, online conference [Portland, OR, USA],
                  June 14-19, 2020},
  pages        = {241--256},
  publisher    = {{ACM}},
  year         = {2020},
  url          = {https://doi.org/10.1145/3318464.3389759},
  doi          = {10.1145/3318464.3389759},
  bibsource    = {dblp computer science bibliography, https://dblp.org}
}

@Misc{stackoverflowreport,
note = {\url{https://insights.stackoverflow.com/survey/2021}},
title = {2021 Stackoverflow Developer Survey},
year = {2021}
}

@book{imbens2015causal,
  title     = {Causal Inference for Statistics, Social, and Biomedical Sciences: An Introduction},
  author    = {Imbens, Guido W. and Rubin, Donald B.},
  year      = {2015},
  publisher = {Cambridge University Press}
}

@article{zanga2022survey,
  title={A survey on causal discovery: Theory and practice},
  author={Zanga, Alessio and Ozkirimli, Elif and Stella, Fabio},
  journal={International Journal of Approximate Reasoning},
  volume={151},
  pages={101--129},
  year={2022},
  publisher={Elsevier}
}

@article{lin2021detecting,
  title={On detecting cherry-picked generalizations},
  author={Lin, Yin and Youngmann, Brit and Moskovitch, Yuval and Jagadish, HV and Milo, Tova},
  journal={Proceedings of the VLDB Endowment},
  volume={15},
  number={1},
  pages={59--71},
  year={2021},
  publisher={VLDB Endowment}
}

@article{pedregosa2011scikit,
  title={Scikit-learn: Machine learning in Python},
  author={Pedregosa, Fabian and Varoquaux, Ga{\"e}l and Gramfort, Alexandre and Michel, Vincent and Thirion, Bertrand and Grisel, Olivier and Blondel, Mathieu and Prettenhofer, Peter and Weiss, Ron and Dubourg, Vincent and others},
  journal={the Journal of machine Learning research},
  volume={12},
  pages={2825--2830},
  year={2011},
  publisher={JMLR. org}
}

@article{zhu2023consistent,
  title={Consistent Range Approximation for Fair Predictive Modeling},
  author={Zhu, Jiongli and Galhotra, Sainyam and Sabri, Nazanin and Salimi, Babak},
  journal={Proceedings of the VLDB Endowment},
  volume={16},
  number={11},
  pages={2925--2938},
  year={2023},
  publisher={VLDB Endowment}
}

@article{louizos2017causal,
  title={Causal effect inference with deep latent-variable models},
  author={Louizos, Christos and Shalit, Uri and Mooij, Joris M and Sontag, David and Zemel, Richard and Welling, Max},
  journal={Advances in neural information processing systems},
  volume={30},
  year={2017}
}

@book{yosida2012functional,
  title={Functional analysis},
  author={Yosida, K{\^o}saku},
  volume={123},
  year={2012},
  publisher={Springer Science \& Business Media}
}

@article{hager1989updating,
  title={Updating the inverse of a matrix},
  author={Hager, William W},
  journal={SIAM review},
  volume={31},
  number={2},
  pages={221--239},
  year={1989},
  publisher={SIAM}
}

@inproceedings{macqueen1967some,
  title={Some methods for classification and analysis of multivariate observations},
  author={MacQueen, James},
  booktitle={Proceedings of the Fifth Berkeley Symposium on Mathematical Statistics and Probability, Volume 1: Statistics},
  volume={5},
  pages={281--298},
  year={1967},
  organization={University of California press}
}

@inproceedings{cheng2019outlier,
  title={Outlier detection using isolation forest and local outlier factor},
  author={Cheng, Zhangyu and Zou, Chengming and Dong, Jianwei},
  booktitle={Proceedings of the conference on research in adaptive and convergent systems},
  pages={161--168},
  year={2019}
}

@article{miller2019explanation,
  title={Explanation in artificial intelligence: Insights from the social sciences},
  author={Miller, Tim},
  journal={Artificial intelligence},
  volume={267},
  pages={1--38},
  year={2019},
  publisher={Elsevier}
}

@article{markakis2024logs,
  title={From Logs to Causal Inference: Diagnosing Large Systems},
  author={Markakis, Markos and Youngmann, Brit and Gao, Trinity and Zhang, Ziyu and Shahout, Rana and Chen, Peter Baile and Liu, Chunwei and Sabek, Ibrahim and Cafarella, Michael},
  journal={Proceedings of the VLDB Endowment},
  volume={18},
  number={2},
  pages={158--172},
  year={2024},
  publisher={VLDB Endowment}
}

@article{scholkopf2021toward,
  title={Toward causal representation learning},
  author={Sch{\"o}lkopf, Bernhard and Locatello, Francesco and Bauer, Stefan and Ke, Nan Rosemary and Kalchbrenner, Nal and Goyal, Anirudh and Bengio, Yoshua},
  journal={Proceedings of the IEEE},
  volume={109},
  number={5},
  pages={612--634},
  year={2021},
  publisher={IEEE}
}

@article{magliacane2018domain,
  title={Domain adaptation by using causal inference to predict invariant conditional distributions},
  author={Magliacane, Sara and Van Ommen, Thijs and Claassen, Tom and Bongers, Stephan and Versteeg, Philip and Mooij, Joris M},
  journal={Advances in neural information processing systems},
  volume={31},
  year={2018}
}

@inproceedings{galhotra2021explaining,
  title={Explaining black-box algorithms using probabilistic contrastive counterfactuals},
  author={Galhotra, Sainyam and Pradhan, Romila and Salimi, Babak},
  booktitle={Proceedings of the 2021 International Conference on Management of Data},
  pages={577--590},
  year={2021}
}

@article{liu2012isolation,
  title={Isolation-based anomaly detection},
  author={Liu, Fei Tony and Ting, Kai Ming and Zhou, Zhi-Hua},
  journal={ACM Transactions on Knowledge Discovery from Data (TKDD)},
  volume={6},
  number={1},
  pages={1--39},
  year={2012},
  publisher={ACM New York, NY, USA}
}

@misc{ACS_Data,
  author       = {{U.S. Census Bureau}},
  title        = {American Community Survey (ACS) - Data},
  year         = {2025},
  url          = {https://www.census.gov/programs-surveys/acs/data.html},
  note         = {Accessed: 2025-02-11}
}

@article{agmon2024finding,
  author       = {Shunit Agmon and
                  Amir Gilad and
                  Brit Youngmann and
                  Shahar Zoarets and
                  Benny Kimelfeld},
  title        = {Finding Convincing Views to Endorse a Claim},
  journal      = {Proc. {VLDB} Endow.},
  volume       = {18},
  number       = {2},
  pages        = {439--452},
  year         = {2024},
  url          = {https://www.vldb.org/pvldb/vol18/p439-agmon.pdf},
  bibsource    = {dblp computer science bibliography, https://dblp.org}
}

@article{youngmann2024summarized,
  title={Summarized Causal Explanations For Aggregate Views},
  author={Youngmann, Brit and Cafarella, Michael and Gilad, Amir and Roy, Sudeepa},
  journal={Proceedings of the ACM on Management of Data},
  volume={2},
  number={1},
  pages={1--27},
  year={2024},
  publisher={ACM New York, NY, USA}
}

@article{glymour2019review,
  title={Review of causal discovery methods based on graphical models},
  author={Glymour, Clark and Zhang, Kun and Spirtes, Peter},
  journal={Frontiers in genetics},
  volume={10},
  pages={524},
  year={2019},
  publisher={Frontiers Media SA}
}

@article{el2014interpretable,
  title={Interpretable and informative explanations of outcomes},
  author={El Gebaly, Kareem and Agrawal, Parag and Golab, Lukasz and Korn, Flip and Srivastava, Divesh},
  journal={Proceedings of the VLDB Endowment},
  volume={8},
  number={1},
  pages={61--72},
  year={2014},
  publisher={VLDB Endowment}
}

@article{wu2013scorpion,
  title={Scorpion: Explaining away outliers in aggregate queries},
  author={Wu, Eugene and Madden, Samuel},
  year={2013},
  publisher={Association for Computing Machinery (ACM)}
}

@inproceedings{roy2014formal,
  title={A formal approach to finding explanations for database queries},
  author={Roy, Sudeepa and Suciu, Dan},
  booktitle={Proceedings of the 2014 ACM SIGMOD international conference on Management of data},
  pages={1579--1590},
  year={2014}
}

@article{rubin2005causal,
  title={Causal inference using potential outcomes: Design, modeling, decisions},
  author={Rubin, Donald B},
  journal={Journal of the American Statistical Association},
  volume={100},
  number={469},
  pages={322--331},
  year={2005},
  publisher={Taylor \& Francis}
}

@article{pearl2009causal,
  title={Causal inference in statistics: An overview},
  author={Pearl, Judea},
  year={2009}
}

@article{roy2015explaining,
  title={Explaining query answers with explanation-ready databases},
  author={Roy, Sudeepa and Orr, Laurel and Suciu, Dan},
  journal={Proceedings of the VLDB Endowment},
  volume={9},
  number={4},
  pages={348--359},
  year={2015},
  publisher={VLDB Endowment}
}

@article{abs-2207-12718,
author = {Ma, Pingchuan and Ding, Rui and Wang, Shuai and Han, Shi and Zhang, Dongmei},
title = {XInsight: EXplainable Data Analysis Through The Lens of Causality},
year = {2023},
issue_date = {June 2023},
publisher = {Association for Computing Machinery},
address = {New York, NY, USA},
journal = {Proc. ACM Manag. Data},
month = {jun},
articleno = {156},
numpages = {27}
}

@article{chernozhukov2016double,
  title={Double/debiased machine learning for treatment and causal parameters},
  author={Chernozhukov, Victor and Chetverikov, Denis and Demirer, Mert and Duflo, Esther and Hansen, Christian and Newey, Whitney and Robins, James},
  journal={arXiv preprint arXiv:1608.00060},
  year={2016}
}

@article{bang2005doubly,
  title={Doubly robust estimation in missing data and causal inference models},
  author={Bang, Heejung and Robins, James M},
  journal={Biometrics},
  volume={61},
  number={4},
  pages={962--973},
  year={2005},
  publisher={Oxford University Press}
}

@article{bae2022if,
  title={If influence functions are the answer, then what is the question?},
  author={Bae, Juhan and Ng, Nathan and Lo, Alston and Ghassemi, Marzyeh and Grosse, Roger B},
  journal={Advances in Neural Information Processing Systems},
  volume={35},
  pages={17953--17967},
  year={2022}
}

@inproceedings{ahuja2021conditionally,
  title={Conditionally independent data generation},
  author={Ahuja, Kartik and Sattigeri, Prasanna and Shanmugam, Karthikeyan and Wei, Dennis and Ramamurthy, Karthikeyan Natesan and Kocaoglu, Murat},
  booktitle={Uncertainty in Artificial Intelligence},
  pages={2050--2060},
  year={2021},
  organization={PMLR}
}

@book{ilyas2019data,
  title={Data cleaning},
  author={Ilyas, Ihab F and Chu, Xu},
  year={2019},
  publisher={Morgan \& Claypool}
}

@article{mcnamee2003confounding,
  title={Confounding and confounders},
  author={McNamee, Roseanne},
  journal={Occupational and environmental medicine},
  volume={60},
  number={3},
  pages={227--234},
  year={2003},
  publisher={BMJ Publishing Group Ltd}
}

@article{miao2020computation,
  title={The computation of optimal subset repairs},
  author={Miao, Dongjing and Cai, Zhipeng and Li, Jianzhong and Gao, Xiangyu and Liu, Xianmin},
  journal={Proceedings of the VLDB Endowment},
  volume={13},
  number={12},
  pages={2061--2074},
  year={2020},
  publisher={VLDB Endowment}
}

@inproceedings{wu2020priu,
  title={Priu: A provenance-based approach for incrementally updating regression models},
  author={Wu, Yinjun and Tannen, Val and Davidson, Susan B},
  booktitle={Proceedings of the 2020 ACM SIGMOD international conference on management of data},
  pages={447--462},
  year={2020}
}

@inproceedings{wu2020deltagrad,
  title={Deltagrad: Rapid retraining of machine learning models},
  author={Wu, Yinjun and Dobriban, Edgar and Davidson, Susan},
  booktitle={International Conference on Machine Learning},
  pages={10355--10366},
  year={2020},
  organization={PMLR}
}

@misc{econml2019,
  author       = {{Microsoft Research}},
  title        = {{EconML: A Python Package for ML-Based Heterogeneous Treatment Effects Estimation}},
  year         = {2019},
  howpublished = {\url{https://github.com/microsoft/EconML}}
}

@article{makhija2025integer,
  title={Is Integer Linear Programming All You Need for Deletion Propagation? A Unified and Practical Approach for Generalized Deletion Propagation},
  author={Makhija, Neha and Gatterbauer, Wolfgang},
  journal={Proceedings of the VLDB Endowment},
  volume={18},
  number={8},
  pages={2667--2680},
  year={2025},
  publisher={VLDB Endowment}
}

@misc{code,
  author       = {{Anonymous Author(s)}},
  title        = {{Git Repository}},
  year         = {2015},
  howpublished = {\url{https://anonymous.4open.science/r/SubCuRe-662F/README.md}}
}

@inproceedings{jeong2020robust,
  title={Robust causal inference under covariate shift via worst-case subpopulation treatment effects},
  author={Jeong, Sookyo and Namkoong, Hongseok},
  booktitle={Conference on Learning Theory},
  pages={2079--2084},
  year={2020},
  organization={PMLR}
}

@article{frank2013would,
  title={What would it take to change an inference? Using Rubin’s causal model to interpret the robustness of causal inferences},
  author={Frank, Kenneth A and Maroulis, Spiro J and Duong, Minh Q and Kelcey, Benjamin M},
  journal={Educational Evaluation and Policy Analysis},
  volume={35},
  number={4},
  pages={437--460},
  year={2013},
  publisher={Sage Publications Sage CA: Los Angeles, CA}
}

@article{broderick2020automatic,
  title={An automatic finite-sample robustness metric: when can dropping a little data make a big difference?},
  author={Broderick, Tamara and Giordano, Ryan and Meager, Rachael},
  journal={arXiv preprint arXiv:2011.14999},
  year={2020}
}

@inproceedings{guo2020certified,
  title={Certified data removal from machine learning models},
  author={Guo, Chuan and Goldstein, Tom and Hannun, Awni and Van Der Maaten, Laurens},
  booktitle={Proceedings of the 37th International Conference on Machine Learning},
  pages={3832--3842},
  year={2020}
}

@inproceedings{golatkar2020eternal,
  title={Eternal sunshine of the spotless net: Selective forgetting in deep networks},
  author={Golatkar, Aditya and Achille, Alessandro and Soatto, Stefano},
  booktitle={Proceedings of the IEEE/CVF conference on computer vision and pattern recognition},
  pages={9304--9312},
  year={2020}
}

@article{mahadevan2021certifiable,
  title={Certifiable machine unlearning for linear models},
  author={Mahadevan, Ananth and Mathioudakis, Michael},
  journal={arXiv preprint arXiv:2106.15093},
  year={2021}
}

@article{nguyen2022survey,
  title={A survey of machine unlearning},
  author={Nguyen, Thanh Tam and Huynh, Thanh Trung and Ren, Zhao and Nguyen, Phi Le and Liew, Alan Wee-Chung and Yin, Hongzhi and Nguyen, Quoc Viet Hung},
  journal={arXiv preprint arXiv:2209.02299},
  year={2022}
}

@article{rosen2011right,
  title={The right to be forgotten},
  author={Rosen, Jeffrey},
  journal={Stan. L. Rev. Online},
  volume={64},
  pages={88},
  year={2011},
  publisher={HeinOnline}
}

@article{ginart2019making,
  title={Making ai forget you: Data deletion in machine learning},
  author={Ginart, Antonio and Guan, Melody and Valiant, Gregory and Zou, James Y},
  journal={Advances in neural information processing systems},
  volume={32},
  year={2019}
}

@inproceedings{salimi2019interventional,
  title={Interventional fairness: Causal database repair for algorithmic fairness},
  author={Salimi, Babak and Rodriguez, Luke and Howe, Bill and Suciu, Dan},
  booktitle={Proceedings of the 2019 International Conference on Management of Data},
  pages={793--810},
  year={2019}
}

@inproceedings{afrati2009repair,
  title={Repair checking in inconsistent databases: algorithms and complexity},
  author={Afrati, Foto N and Kolaitis, Phokion G},
  booktitle={Proceedings of the 12th International Conference on Database Theory},
  pages={31--41},
  year={2009}
}

@inproceedings{bohannon2005cost,
  title={A cost-based model and effective heuristic for repairing constraints by value modification},
  author={Bohannon, Philip and Fan, Wenfei and Flaster, Michael and Rastogi, Rajeev},
  booktitle={Proceedings of the 2005 ACM SIGMOD international conference on Management of data},
  pages={143--154},
  year={2005}
}

@inproceedings{chu2013holistic,
  title={Holistic data cleaning: Putting violations into context},
  author={Chu, Xu and Ilyas, Ihab F and Papotti, Paolo},
  booktitle={2013 IEEE 29th International Conference on Data Engineering (ICDE)},
  pages={458--469},
  year={2013},
  organization={IEEE}
}

@article{chomicki2005minimal,
  title={Minimal-change integrity maintenance using tuple deletions},
  author={Chomicki, Jan and Marcinkowski, Jerzy},
  journal={Information and Computation},
  volume={197},
  number={1-2},
  pages={90--121},
  year={2005},
  publisher={Elsevier}
}

@article{geerts2013llunatic,
  title={The LLUNATIC data-cleaning framework},
  author={Geerts, Floris and Mecca, Giansalvatore and Papotti, Paolo and Santoro, Donatello},
  journal={Proceedings of the VLDB Endowment},
  volume={6},
  number={9},
  pages={625--636},
  year={2013},
  publisher={VLDB Endowment}
}

@inproceedings{bohannon2006conditional,
  title={Conditional functional dependencies for data cleaning},
  author={Bohannon, Philip and Fan, Wenfei and Geerts, Floris and Jia, Xibei and Kementsietsidis, Anastasios},
  booktitle={2007 IEEE 23rd international conference on data engineering},
  pages={746--755},
  year={2006},
  organization={IEEE}
}

@inproceedings{yan2020scoded,
  title={Scoded: Statistical constraint oriented data error detection},
  author={Yan, Jing Nathan and Schulte, Oliver and Zhang, MoHan and Wang, Jiannan and Cheng, Reynold},
  booktitle={Proceedings of the 2020 ACM SIGMOD International Conference on Management of Data},
  pages={845--860},
  year={2020}
}

@inproceedings{izzo2021approximate,
  title={Approximate data deletion from machine learning models},
  author={Izzo, Zachary and Smart, Mary Anne and Chaudhuri, Kamalika and Zou, James},
  booktitle={International Conference on Artificial Intelligence and Statistics},
  pages={2008--2016},
  year={2021},
  organization={PMLR}
}

@article{livshits2020computing,
  title={Computing optimal repairs for functional dependencies},
  author={Livshits, Ester and Kimelfeld, Benny and Roy, Sudeepa},
  journal={ACM Transactions on Database Systems (TODS)},
  volume={45},
  number={1},
  pages={1--46},
  year={2020},
  publisher={ACM New York, NY, USA}
}

@article{livshits2022shapley,
  title={The Shapley value of inconsistency measures for functional dependencies},
  author={Livshits, Ester and Kimelfeld, Benny},
  journal={Logical Methods in Computer Science},
  volume={18},
  year={2022},
  publisher={Episciences. org}
}

@inproceedings{kolahi2009approximating,
  title={On approximating optimum repairs for functional dependency violations},
  author={Kolahi, Solmaz and Lakshmanan, Laks VS},
  booktitle={Proceedings of the 12th International Conference on Database Theory},
  pages={53--62},
  year={2009}
}

@String{Computing = "Computing" }

@String{Computer = "{IEEE} Computer" }

@String{Springer = "Springer-Verlag" }

@inproceedings{youngmann2023explaining,
  title={On Explaining Confounding Bias},
  author={Youngmann, Brit and Cafarella, Michael and Moskovitch, Yuval and Salimi, Babak},
  booktitle={ICDE},
  pages={1846--1859},
  year={2023},
  organization={IEEE}
}

@inproceedings{galhotra2023metam,
  title={Metam: Goal-oriented data discovery},
  author={Galhotra, Sainyam and Gong, Yue and Fernandez, Raul Castro},
  booktitle={2023 IEEE 39th International Conference on Data Engineering (ICDE)},
  pages={2780--2793},
  year={2023},
  organization={IEEE}
}

 \appendix

\section*{Appendix}
We provide missing proofs, details, and additional experiments.

\section{Proofs}
In this part, we provide formal proofs of the propositions stated in Section \ref{sec:problem}.

\subsubsection{Proof of Proposition \ref{prop:problem_is_np_hard}}

\begin{proof}[Proof of Proposition \ref{prop:problem_is_np_hard}]
We show a reduction from SUBSET-SUM to the decision version of \probName. 
Recall that in SUBSET-SUM we are given a set of numbers $S = \{x_1, \ldots, x_n\}\subseteq \bbN$ and a target $k$, and the problem is to decide whether there exists a subset $S' \subseteq S$ such that $\sum_{x\in S'}x = k$. 

Recall that in \probName\ we are given a database instance $\db$ over a schema $\attrset$, a binary treatment variable $T \in \attrset$, an outcome variable $O \in \attrset$, a desired value $ATE_d$ and $\epsilon>0$. In the decision version of this problem we are also given an additional bound $n$, and we need to decide whether there is a subset of tuples $\Gamma \subseteq \db$ s.t. $|\Gamma|\le n$ and it holds that
$ATE_{\db \setminus \Gamma}(T,O) \in [ATE_d - \epsilon, ATE_d + \epsilon]$.

We turn to describe the reduction. See \cref{xmp:reduction subset sum} for an example.
Given a SUBSET-SUM instance $(S, k)$, we define an instance of \probName\ as follows: 
The schema $\attrset$ is defined as $\attrset = \{T,O\}$ (i.e., there are no confounding variables). 
For each element $x \in S$, we add to $\db$ two corresponding tuples $t_x^1$ and $t_x^2$ with the following values:
\begin{itemize}
    \item $t_x^1[T] = 1, t_x^1[O] = x$
     \item $t_x^2[T] = 0, t_x^2[O] = 0$
\end{itemize}
In addition, we add to $\db$ the tuple $t$ where $t[T] = 1$ and $t[O] = - k$.
We set the target ATE to $ATE_d=0$ and set $\epsilon = 0$, i.e., we must reach exactly the target ATE of $0$. Finally, we bound on the allowed number of tuples to delete to be at most $|S|$. 
Clearly the reduction can be implemented in polynomial time. We turn to show correctness.

($\Rightarrow$) Assume the SUBSET-SUM instance has a solution $S'\subseteq S$. We remove from $\db_d$ all tuples $t_x^1$ for $x\notin S'$ (intuitively, we retain only the $t_x^1$ tuples that are in $S'$). Note that since $|S'|\le |S|$, we are within the allowed number of deleted tuples. We claim that the ATE after this removal is $0$.
Indeed, we have $AVG(O|T= 0) = 0$ (since all the tuples with $T=0$ have $O=0$), and since $\sum_{x\in S'}x=k$, then $AVG(O|T=1)=(-k+\sum_{x\in S'}x)/(|S'|+1)=(-k+k)/(|S'|+1)=0$ (accounting for the remaining tuples with $T=1$). It follows that the resulting ATE is $0$.

($\Leftarrow$) Assume that the instance of our problem is solvable, with a deleted set of tuples $\Gamma$. 
Intuitively, we do not care about deleted tuples where $t[T] = 0$, as these tuples also have $t[O]=0$ and do no affect the ATE. 
Denote by $\Gamma'$ the set of deleted tuples with $t[T]=1$. Since the new ATE is $0$, and since the contribution of $t[T]=0$ tuples is $0$, it follows that after deleting the $\Gamma'$, the average of the tuples with $t[T]=1$ is $0$, and therefore also the sum.

Since $S$ contains only non-negative numbers, it follows that the tuple $t$ where $t[O] = -k$ is not deleted (otherwise the sum would be strictly positive). Thus, all tuples in $\Gamma'$ are of the form $t[T] = 1$ and $t[O] >0$. Denote by $S'=\{x\in S\mid t_x\notin \Gamma'\}$ the remaining tuples, it follows that $-k+\sum_{x\in S'}x=0$, so $\sum_{x\in S'}x=k$, and therefore the SUBSET-SUM instance is solvable.

\end{proof}

    \begin{table}[H]
\caption{The database $\db$ created by the reduction in \cref{xmp:reduction subset sum}. The tuple naming convention is as per the reduction in \cref{prop:problem_is_np_hard}.}
    \label{tab:instance_for_subsetsum}
    \begin{tabular}{|c|c|c|}
      \hline
      &T & O \\
      \hline
   $t_1^1$  & 1 & 1 \\
     $t_3^1$& 1 & 3 \\
     $t_5^1$& 1 & 5 \\
      $t$& 1 & -4 \\
     $t_1^2$& 0 & 0 \\
     $t_3^2$& 0 & 0 \\
     $t_5^2$& 0 & 0 \\
      \hline
    \end{tabular}

 

  \end{table}

  \begin{example}
  \label{xmp:reduction subset sum}
We illustrate the reduction in \cref{prop:problem_is_np_hard}.
Consider the following instance of SUBSET-SUM: $S = \{1,3,5\}$ and $k = 4$. We construct an instance of the cardinality repair for causal effect targeting problem as depicted in \cref{tab:instance_for_subsetsum}, and set $\epsilon = 0$. 
Observe that $ATE(T,O)=\frac14(1+3+5-4)-\frac13 0=\frac54$.





A solution to the SUBSET-SUM instance is the set $S' = \{1,3\}$ since their sum is equal to $4$. Therefore, by deleting from $\db$ the tuples $t_5^1$ (namely keeping $t_1^1$ and $t_3^1$), the new ATE is $\frac13(1+3-4)-\frac13 0=0$, as required.
\end{example}

\subsubsection{Proof of \cref{prop:problem_is_np_hard_pattern}}

We now turn to prove \cref{prop:problem_is_np_hard_pattern}, showing that \cref{prob:patterns} is NP-hard. Again, the decision version of the problem includes a parameter $n$ that bounds the size of the removed population.
Our reduction is from the setting of tuple deletion, namely \probName. 

Consider therefore an instance of the cardinality repair for causal effect targeting problem, i.e., a database $\db$ over schema $\attrset$, variables $T,O$, a desired value $ATE_d$, $\epsilon>0$ and a bound $n$. We can assume without loss of generality that there are no confounding variables, as the problem remains NP-hard in this case, as we show in the proof of \cref{prop:problem_is_np_hard} (this assumption only simplifies the writing, but has no meaningful impact on the construction).

Intuitively, the reduction introduces new attributes, in such a way that every subset of tuples can be defined by a pattern, thus showing an equivalence between tuple deletion and pattern deletion.




We start with the explicit construction.

Denote by $m$ the number of tuples in $\db$. We obtain a new database $\db_2$ by adding to $\db$ fresh attributes $S_1,\ldots, S_m$ (where $S_i$ stands for ``Select $i$''), and by modifying the tuples such that for tuple $t_i$ the values of the new attributes are $t_i(S_j)=0$ if $i\neq j$ and $t_i(S_i)=1$.
We refer to $\db_2$ as the \emph{identifier augmentation} of $\db$. We demonstrate the construction in \cref{tab:identifier augmentation example}. 
\begin{table}[ht]
    \centering
    \caption{A database $\db$ and its identifier augmentation $\db_2$.}
    \label{tab:identifier augmentation example}
    \begin{minipage}[b]{0.2\textwidth}
        \centering
        \begin{tabular}{|c|c|c|}
            \hline
            &T & O \\
            \hline
            $t_1$&1    & 10    \\
            $t_2$&1    & 8    \\
            $t_3$&1    & 6    \\
            $t_4$&1    & 3    \\
            \hline
        \end{tabular}
        \caption{$\db$}
    \end{minipage}
    \begin{minipage}[b]{0.25\textwidth}
        \centering
        \begin{tabular}{|c|c|c|c|c|c|c|}
            \hline
            & $S_1$ & $S_2$ & $S_3$ & $S_4$ &T & O \\
            \hline
           $t_1$&1&0&0&0&0    & 12    \\
           $t_2$&0&1&0&0&0    & 9    \\
           $t_3$&0&0&1&0&0    & 1    \\
           $t_4$&0&0&0&1&0    & 1    \\
            \hline
        \end{tabular}
        \caption{$\db_2$}
    \end{minipage}
\end{table}

\begin{lemma}
    \label{lem:identifier augmentation property}
    Consider a database $\db$ and its identifier augmentation $\db_2$. For every set of tuples $I\subseteq \{1,\ldots, n\}$, there exists a pattern $\pattern_I$ such that $\sat(\pattern_I)=I$.
\end{lemma}
\begin{proof}
    Define the pattern $\pattern_I=\bigwedge_{i\notin I} S_i=0$. We claim that $\sat(\pattern_I)=I$. Indeed, for $i\in I$ we have that $t_i(S_j)=0$ for all $j\neq i$, and in particular for all $j\notin I$. Therefore $i\in \sat(\pattern_I)$.

    Conversely, let $i\in \sat(\pattern_I)$, then $t_i(S_j)=0$ for all $j\notin I$. Since $t_i(S_i)=1$, it follows that $i\in I$.
\end{proof}
As an example of \cref{lem:identifier augmentation property}, consider the set $\{t_2,t_3\}$ in \cref{tab:identifier augmentation example}. The pattern $S_1=0\wedge S_4=0$ exactly captures it.

\begin{remark}
    \label{rmk:disjunctive patterns}
    A similar construction can be obtained with \emph{disjunctive} patterns, by setting the pattern to capture all entries \emph{inside} the set $I$.
\end{remark}

We can now show the correctness of the reduction. 

$(\Rightarrow).$ If there exists a subset $\Gamma\subseteq \db$ such that $ATE_{\db\setminus \Gamma}(T,O)\in [ATE_d-\epsilon,ATE+\epsilon]$ and $|\Gamma|\le n$, consider the pattern $\pattern_\Gamma$ for $\db_2$ as per \cref{lem:identifier augmentation property}, then removing $\pattern_\Gamma$ results in the same database $\db\setminus \Gamma$, and therefore yields an ATE also within $\epsilon$ from $ATE_d$. 

$(\Leftarrow).$ If there is a pattern in $\db_2$ whose removal yields an ATE within $\epsilon$ from $ATE_d$, clearly we can also select the tuples in this pattern directly, without specifying them as a tuple, in the original database $\db$.

\begin{remark}
    \label{rmk:tuple deletion works with arbitrary function}
    The reduction above does not use any property of ATE. Indeed, the reduction of tuple removal to pattern removal remains correct regardless of the optimization function.
\end{remark}
In light of \cref{rmk:tuple deletion works with arbitrary function}, we have the following in particular.
\begin{theorem}
\label{thm:patter deletion NPc}
The pattern-deletion problem is NP-complete for AVG and ATE.
\end{theorem}

\section{ATE Update for IPW}

Algorithm~\ref{algo:fisher-unlearning} depicts the Fisher mini-batch incremental update algorithm for logistic regression that we use to update the model parameters. 

\begin{algorithm}[t]
  \small
  \DontPrintSemicolon
  \SetKwInOut{Input}{Input}\SetKwInOut{Output}{Output}
  \LinesNumbered
  \Input{Current model parameters \(\theta\in\mathbb{R}^d\), original dataset \(\db\), subset to be removed \(\db_{\text{rmv}}\subseteq \db\), mini-batch size $m'$, loss function \(L(\theta;\db)\)}
  \Output{Unlearned parameters \(\theta_{\text{new}}\)}
  \BlankLine
  \SetKwFunction{PartitionBatches}{PartitionBatches}
  \SetKwFunction{EstimateGrad}{EstimateGrad}
  \SetKwFunction{EstimateHess}{EstimateHess}
  \SetKwFunction{SampleNoise}{SampleNoise}

 \tcc{Generating Mini-Batches}
  $s \leftarrow \left\lceil |D_{\text{rmv}}|/m' \right\rceil$\;
  Split $D_{\text{rmv}}$ into $\{D_{\text{rmv}}^1,\dots,D_{\text{rmv}}^s\}$\;
  \(D_{\text{new}}\!\gets\!D\setminus D_{\text{rmv}}\)\;
  \(\theta_{\text{new}}\!\gets\theta\)\;
  \For{$i \leftarrow 1$ \KwTo $s$}{
    \tcc{Compute gradient and empirical Fisher (Hessian)}
  \(\Delta\gets\)\(\nabla L(\theta_{\text{new}},D_{\text{new}}^i)\), \(F\;\gets\)\(\nabla^2 L(\theta_{\text{new}},D_{\text{new}}^i)\)\;
  \(\theta_{\text{new}}\!\gets\!\theta_{\text{new}} - F^{-1}\,\Delta\)\;
      \tcc{add calibrated noise for \(\sigma\)-certified unlearning}
  \If{\(\sigma>0\)}{
    \(b\gets\)\SampleNoise{\(d\)} \tcc{\(b\sim\mathcal{N}(0,I_d)\)}
    \(\theta_{\text{new}}\!\gets\!\theta_{\text{new}} + \sigma\,F^{-1/4}\,b\)\;
  }
  }
  \Return \(\theta_{\text{new}}\)\;
  \caption{Fisher Update for Logistic Regression}
  \label{algo:fisher-unlearning}
\end{algorithm}



\section{Additional Experiments}
\label{app:exp}


\subsubsection{Experiments with Synthetic Data under Known Deletion Budgets}\label{subsec:synthetic}

Running an exhaustive brute-force search algorithm to find the optimal solution was infeasible even on large datasets. To nonetheless evaluate how close \algoNameTuple\ comes to the optimal solution in large scale data, we designed experiments using synthetic data where an upper bound on the number of tuples to remove is known.
Specifically, we generated a synthetic dataset with 10k tuples, including a binary treatment variable, a continuous outcome, and three confounders. We set the target ATE to be the one computed over this data with $\epsilon = 0$. We then introduced an increasing number of noisy records into the data. In this setup, the number of inserted noisy tuples serves as an upper bound on the number of deletions required to reach the target ATE.

The results are shown in Figure \ref{fig:synthetic}. The approximate versions of \algoNameTuple\ and \topk\ are omitted as we observed similar trends. Up to 1300 noisy tuples, all algorithms found solutions smaller than the upper bound, indicating high accuracy. However, as the number of noisy tuples increases, implying that many tuples must be removed to reach the target ATE, the performance of \topk\ degrades. This inefficiency arises because \topk\ computes influence scores only once, which eventually becomes outdated. A similar trend was observed in the ACS and Twins use cases. 

\begin{figure}[ht]
  \centering
  \begin{minipage}[b]{0.25\textwidth}
    \includegraphics[width=\textwidth]{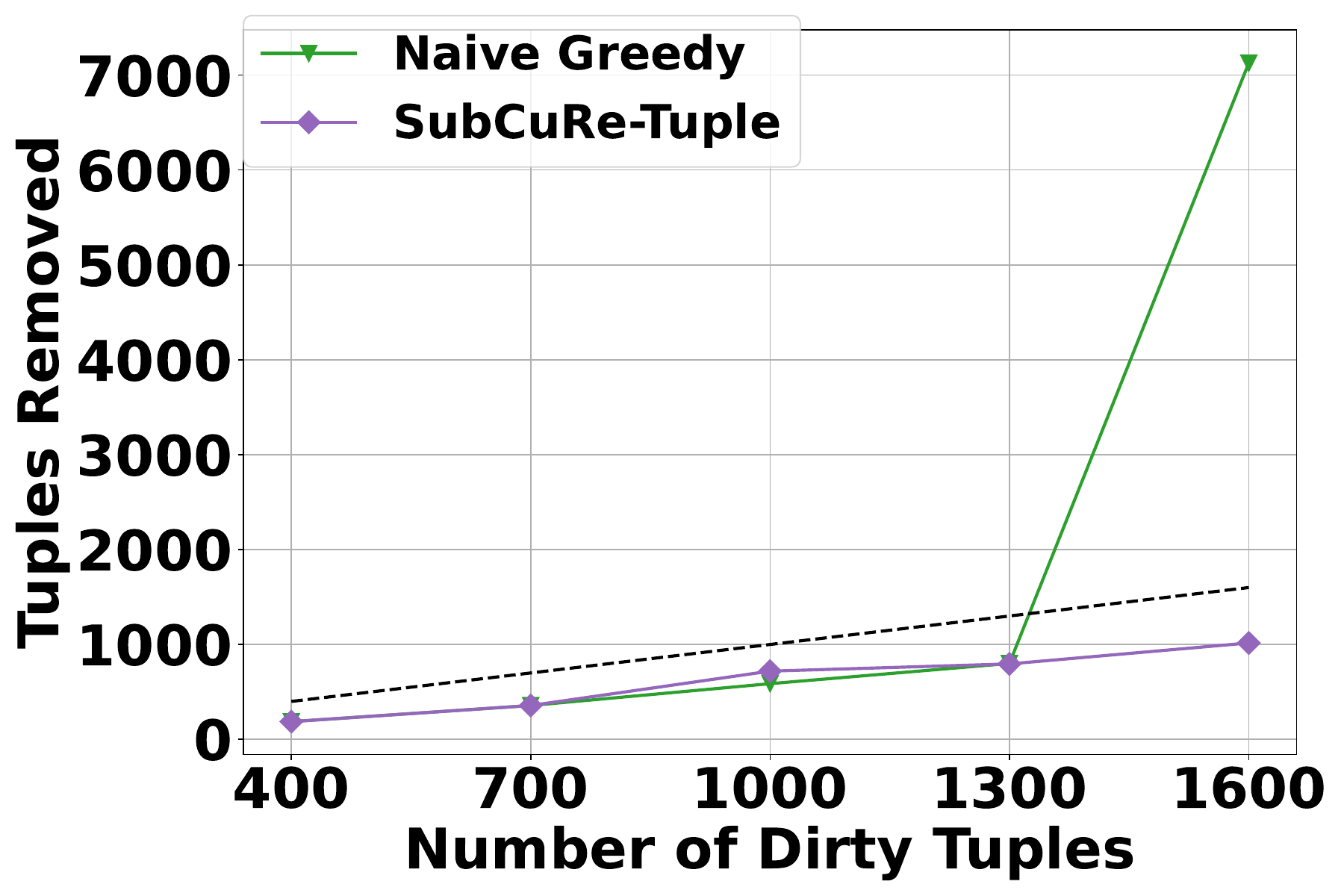}
  \end{minipage}
  \caption{Number of removed tuples vs. number of noisy tuples. An upper bound on the solution size is marked by a dashed black line.}
  \label{fig:synthetic}
\end{figure}

\begin{figure*}[t]
    \centering
    \begin{minipage}[t]{0.25\textwidth}
        \centering
        \includegraphics[height=2.4cm]{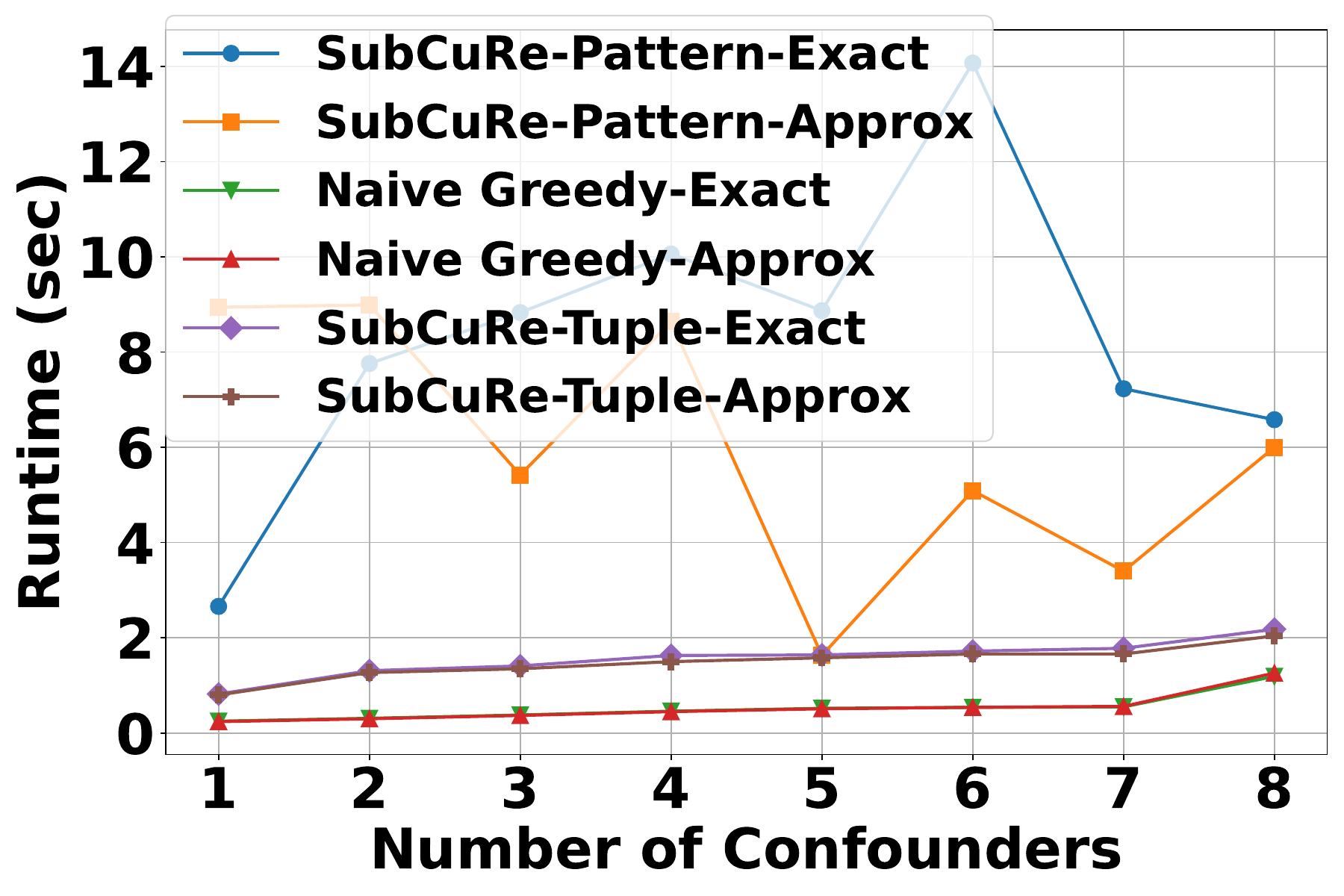}
        \caption*{(a) German Credit}
        \label{fig:image1}
    \end{minipage}
    \begin{minipage}[t]{0.25\textwidth}
        \centering
        \includegraphics[height=2.4cm]{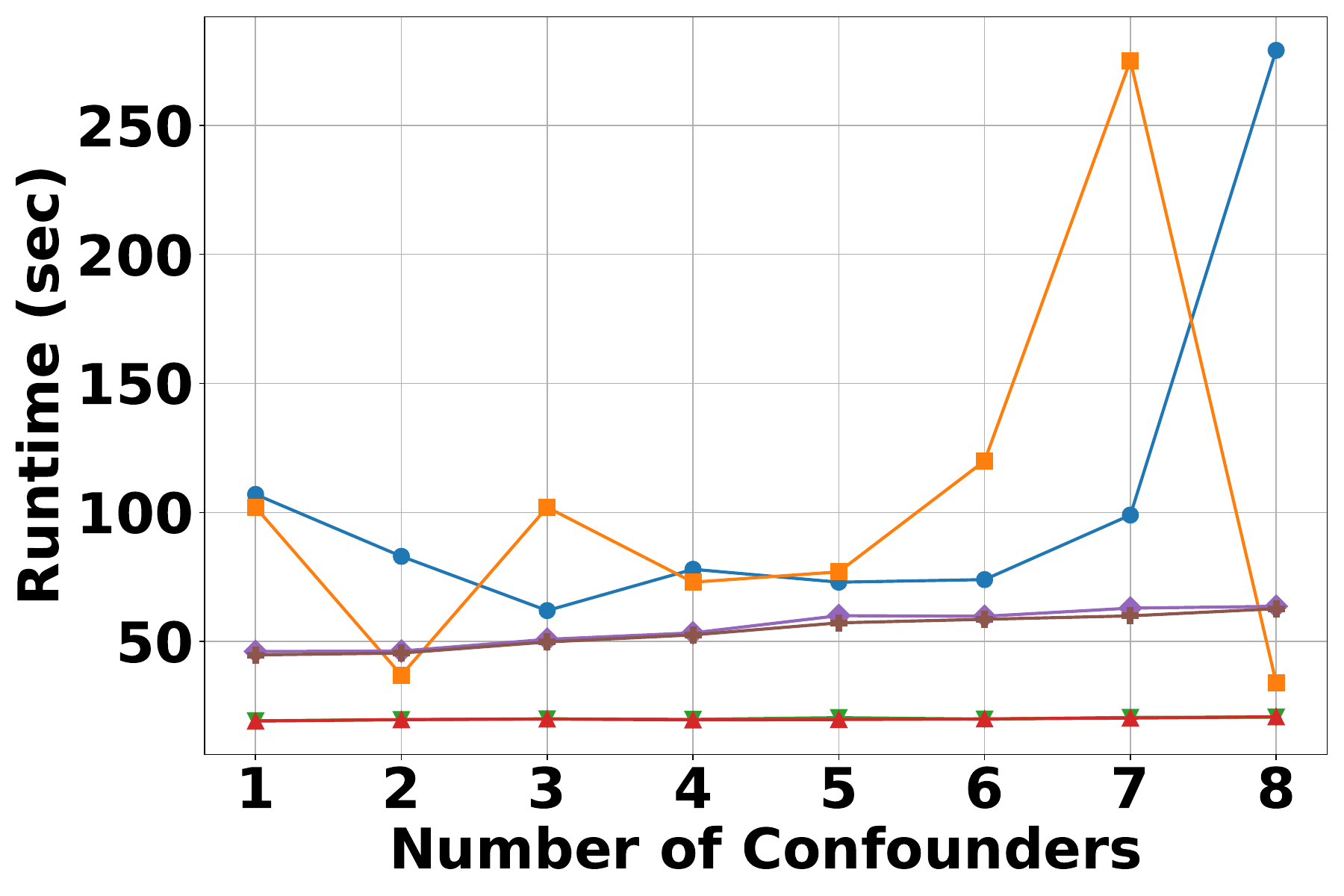}
        \caption*{(b) Stack Overflow}
        \label{fig:image2}
    \end{minipage}
    \begin{minipage}[t]{0.25\textwidth}
        \centering
        \includegraphics[height=2.4cm]{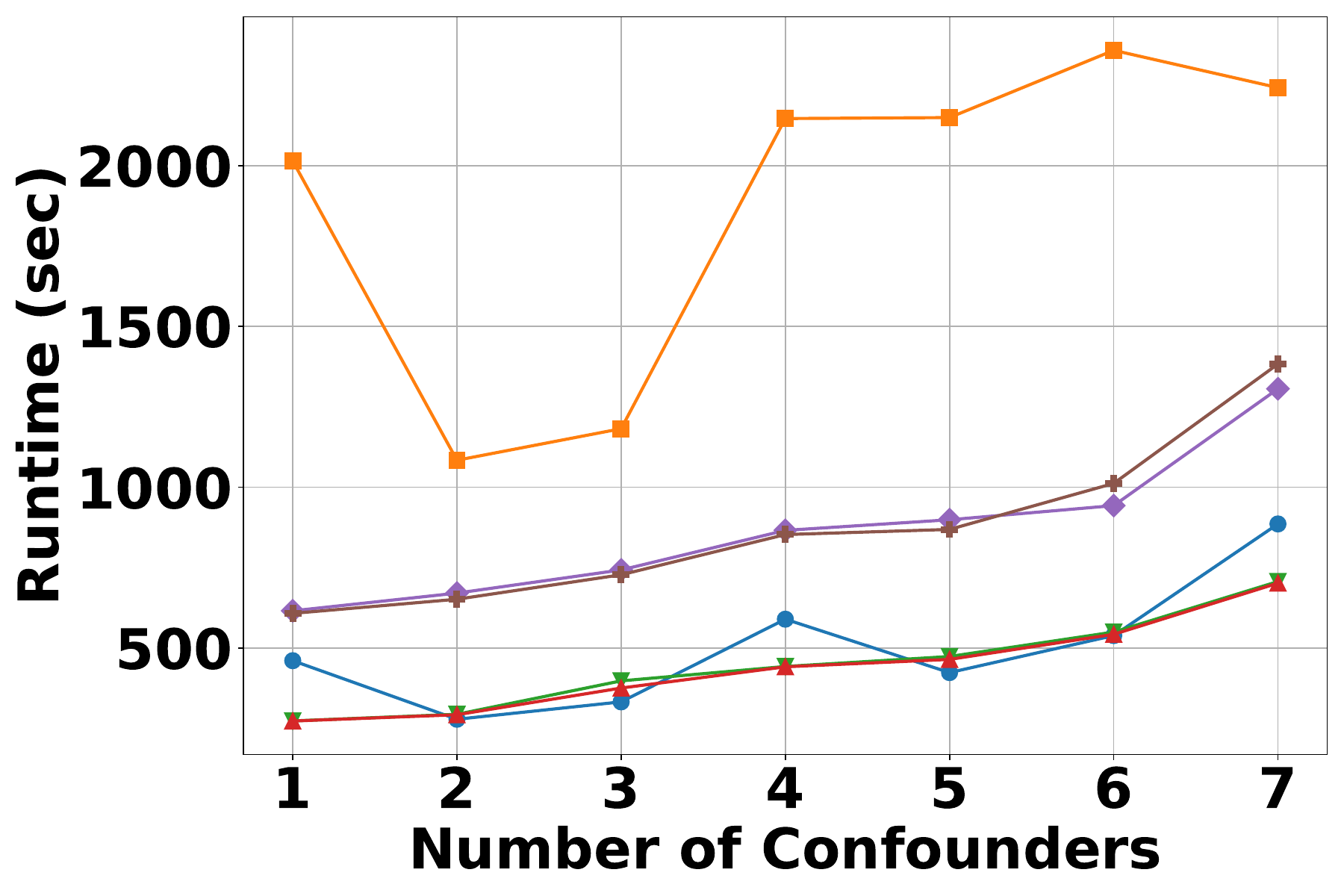}
        \caption*{(c) ACS}
        \label{fig:image3}
    \end{minipage}
    \caption{Runtime as a function of number of confounders.}
    \label{fig:runtime_confounders}
\end{figure*}

\subsubsection{Factors Influencing the Performance \algoNamePattern}
We evaluate how the number of random walks affects both runtime and the algorithm's ability to identify impactful, small subpopulations. Too few walks risk missing valid solutions, while too many increase runtime. Experiments on real-world datasets show that capping random walks at 1000 provides a good trade-off between quality and efficiency. A similar balance is achieved by terminating a walk when the size of the current group exceeds 20\% of the data. These parameters are user-configurable and can be adjusted to control runtime.

\paragraph*{Runtime vs. number of confounders}
The results are shown in Figure~\ref{fig:runtime_confounders} illustrates the effect of the number of confounders on runtime. With more confounders, each ATE computation takes longer, and thus, as expected, it affects the runtime of all methods.

\end{document}
\endinput